%% file: mainArxiv.tex
\documentclass[a4paper,UKenglish,cleveref,thm-restate]{lipics-v2021}
\hideLIPIcs
\nolinenumbers
\input{macros}

\newcommand{\app}{ (see Appendix) }

\title{Bandwidth of Timed Automata:  3 Classes
}

\funding{This work was supported by the ANR project MAVeriQ ANR-CE25-0012 and by the ANR-JST project CyPhAI.}

\relatedversion{To appear (without appendix) in FSTTCS'23: \href{https://www.doi.org/10.4230/LIPIcs.FSTTCS.2023.6}{DOI 10.4230/LIPIcs.FSTTCS.2023.6}}


\author{Eugene Asarin}{Université Paris Cité, CNRS, IRIF, Paris, France}{asarin@irif.fr}{https://orcid.org/0000-0001-7983-2202}{}
\author{Aldric Degorre}{Université Paris Cité, CNRS, IRIF, Paris, France}{adegorre@irif.fr}{https://orcid.org/0000-0003-2712-4954}{}
\author{C\u at\u alin Dima}{LACL, Université Paris-Est Créteil, France} {dima@upec.fr}{https://orcid.org/0000-0001-5981-4533}{}
\author{Bernardo Jacobo Inclán}{Université Paris Cité, CNRS, IRIF, Paris, France}{jacoboinclan@irif.fr}{https://orcid.org/0009-0009-5323-7945}{}
\acknowledgements{We thank the anonymous reviewers for their careful reading and valuable comments which helped to improve the paper. }
\authorrunning{Asarin, Degorre, Dima, and Jacobo Inclán}
\Copyright{E. Asarin, A. Degorre, C. Dima, and B. Jacobo Inclán} 

\begin{document}
\maketitle

\begin{abstract}
Timed languages contain sequences of discrete events (``letters'') separated by real-valued delays, they can be recognized by timed automata, and represent behaviors of various real-time systems.
The notion of bandwidth of a timed language defined in \cite{formats2022} characterizes the amount of information per time unit, encoded in words of the language observed with some precision $\varepsilon$.

In this paper, we identify three classes of timed automata according to the asymptotics of the bandwidth  of their languages with respect to this  precision $\varepsilon$: automata are either meager, with an $O(1)$ bandwidth,  normal, with a $\Theta\left(\log\frac{1}{\varepsilon}\right)$ bandwidth, or obese, with $\Theta\left(\frac{1}{\varepsilon}\right)$ bandwidth. We define two structural criteria and prove that they partition 
timed automata into these 3 classes of bandwidth, implying that there are no intermediate asymptotic classes. The classification problem of a 
 timed automaton is  \PSPACE-complete.

Both criteria are  formulated using morphisms from paths of the timed automaton to some finite monoids extending Puri's orbit graphs; the proofs are based on Simon's factorization  forest theorem.

\end{abstract}

\section{Introduction}
   \input{introArxiv}

\section{Background and Problem Statement}\label{sec:prelim}
\input{background}

\section{Distinguishing between Meager and Other Automata}\label{sec:meager}

\input{thin}

\section{Distinguishing between Obese and Other Automata}\label{sec:obese}

\input{obese}

\section{Gathering Stones: Classification}\label{sec:classif}

\input{classify}

\section{Conclusion}\label{sec:conclusion}
    \input{conclusion}

\bibliographystyle{plainurl}
\bibliography{entro}


\appendix
\input{appendix}

\end{document}

%% file: macros.tex
\usepackage{cite}
\usepackage{dsfont}
\usepackage{complexity}
\newclass{\TWOEXP}{2\!-\!EXPTIME}
\newclass{\THREEEXP}{3\!-\!EXPTIME}
\newclass{\EXPTIME}{EXPTIME}
\usepackage{tikz}
\usetikzlibrary {arrows.meta}
\usetikzlibrary{arrows,trees,backgrounds,automata,shapes,plotmarks,decorations}
\tikzstyle{block}=[rectangle,draw, thin, inner sep=3pt, text centered,fill=orange!20!yellow!20] 
\tikzstyle{net}=[draw,cloud,fill=yellow!20,aspect=3,inner sep=1pt]
\tikzstyle{dev}=[draw,circle,fill=yellow!20,aspect=2,inner sep=1pt,minimum size=.6cm]
\tikzstyle{pre}=[<-,shorten <=1pt,>=stealth']
\tikzstyle{post}=[->,shorten >=1pt,>=stealth']
\tikzstyle{bi}=[<->,shorten >=1pt,shorten <=1pt, >=stealth']
\tikzstyle{every initial by arrow}=[initial text={},initial distance=1em,post]
\tikzstyle{every state}=[minimum size=0.6cm,fill=cyan!20!yellow!20]
\tikzstyle{transition}= [post,shorten >=1pt,node distance=2cm, inner sep=2pt,bend angle=20]

\theoremstyle{plain}
  \newtheorem{thm}{Theorem}
  \newtheorem{con}[theorem]{Construction}
\Crefname{thm}{Thm.}{Thms.}
\Crefname{con}{Constr.}{Constrs.}
\Crefname{proposition}{Prop.}{Props.}
\Crefname{lemma}{Lem.}{Lemmas}
\Crefname{definition}{Def.}{Defs.}
\Crefname{corollary}{Cor.}{Cors.}
\Crefname{figure}{Fig.}{Figs.}
\Crefname{section}{Sect.}{Sections}

\newcommand{\one}{\mathds{1}}
\newcommand{\real}{\mathds{R}}
\newcommand{\nat}{\mathds{N}}
\newcommand{\integer}{\mathds{Z}}

\newcommand{\aut}{\mathcal{A}}
\newcommand{\ent}{\mathcal{H}}
\newcommand{\capa}{\mathcal{C}}

\newcommand{\bandh}{\mathcal{BH}}
\newcommand{\bandc}{\mathcal{BC}}
\newcommand{\trans}[1]{\stackrel{#1}{\to}}

\newcommand{\src}{\mathop{src}\nolimits}
\newcommand{\dst}{\mathop{dst}\nolimits}
\newcommand{\lbl}{\mathop{lbl}\nolimits}
\newcommand{\guard}{\mathfrak{g}}
\newcommand{\reset}{\mathfrak{r}}
\newcommand{\card}[1]{\#\!#1}

\newcommand{\instant}{\mathbf{instant}}
\newcommand{\fast}{\mathbf{fast}}
\newcommand{\slow}{\mathbf{slow}}
\newcommand{\none}{\mathsf{0}}

\newcommand{\narrow}{\mathbf{narrow}}
\newcommand{\wide}{\mathbf{wide}}
\newcommand{\0}{\mathbf{0}}
\newcommand{\1}{\mathbf{1}}

\newcommand{\dr}{\overrightarrow{d}}

\newcommand{\Path}{\mathop{path}}
\newcommand{\Paths}{\mathop{Paths}}
\newcommand{\Word}{\mathop{word}}
\newcommand{\Let}{\mathop{letters}}
\newcommand{\Runs}{\mathop{Runs}}
\newcommand{\Still}{\mathop{Still}}
\newcommand{\ssum}{\mathop{sum}}
\newcommand{\proj}{\mathop{proj}}
\newcommand{\timing}{\mathop{timing}}

\newcommand{\mon}{\mathbf{M}}

\newfunc{\myexp}{exp}
\newlang{\myreach}{Reach}
\newlang{\Obese}{Obese}
\newlang{\Meager}{Meager}
\newlang{\Normal}{Normal}

\newcommand{\aaa}{node[blue]{$\bullet$}}
\newcommand{\bbb}{node[red]{$\circ$}}

\newcommand{\Ver}{V}

\newcommand{\Reach}{\mathop{Reach}}
\newcommand{\SCC}{\mathop{SCC}}
\newcommand{\Conv}{\mathop{ConvHull}}
\newcommand{\Low}{\mathop{low}}
\newcommand{\Hi}{\mathop{high}}

\newcommand{\DBM}{\textsc{DBM}}
\newcommand{\DTA}{\textsc{DTA}}
\newcommand{\RTA}{\textsc{RsTA}}
\newcommand{\stA}{\textsc{smA}}
\newcommand{\nsoA}{\textsc{nsoA}}

\newcommand{\vvvert}{|\!|\!|}

\ifdefined\unbounded

\fi

\newcommand{\ck}{c_{e\varepsilon \varkappa}}
\newcommand{\xk}{x_{e\varepsilon \varkappa}}

\newcommand{\tuple}[1]{\left\langle #1 \right\rangle}


%% file: introArxiv.tex
The study of the growth rate  of formal languages  was initiated in the 1950s in the pioneering \cite{ChomskyMiller}, where the growth of the number of words of length $n$ with respect to $n$ in regular languages was analyzed. Even earlier, visionary   Shannon \cite{shannon48} related word counts to the  capacity of a discrete noiseless channel. In the early '60s, Kolmogorov \cite{3approaches} conceptualized the counting (``combinatorial'') approach to the quantity of information, as compared to the probabilistic and the algorithmic ones. According to this approach, the quantity of information conveyed by an element of a finite set $S$ is just $\log_2 \#S$ bits, and in most 
cases, asymptotics of this amount w.r.t.~some size parameter is studied. Nowadays, counting-based asymptotic analysis of  formal languages  constitutes a background of the theory of codes\cite{perrin},   and of practical protocols \cite[Chapter 6]{imminkBook} implemented in every hard disk drive and DVD \cite{efmplus}. In a nutshell, one can encode a source language  with a growth rate $A$ (that is with $\sim 2^{An}$ words of size $n$) into a channel with growth rate $B$ only if  $A\leq B$. 
The growth rate is also closely related  to the entropy rate of subshifts in symbolic dynamics  \cite{Marcus}.

In the long run, we are porting this approach to timed languages and automata, with both theoretical and practical ambitions. The intrinsic difficulty here is the continuous infinity of such languages and the infinite-state nature of timed automata. For growth rate (or information quantity) per event, we have mostly solved the problem in  \cite{entroJourn} and \cite{timedCoding}.

 However, in the above, time was seen as data carried by the language, which contrasts with many natural applications, where time describes the actual execution time of the system under study. For this reason, in \cite{formats2022} we formalized the notion of bandwidth, the quantity of information per time unit, given some observation precision, and proved its actual relevance to bounded delay channel coding. Unfortunately, how this bandwidth can be computed for any timed automaton is not yet known. Thus, as a preliminary step, in this paper, we propose a rough qualitative classification of automata into three classes according to the asymptotic of the bandwidth w.r.t.~precision.
 We believe that this qualitative result
 \begin{itemize}
\item is of interest by itself and provides a relevant classification of qualitative behaviors of timed automata;
\item  is in line with other classification results for discrete languages;
\item involves various and interesting techniques and tools;
\item  paves the way to the more precise computation of bandwidth.  
 \end{itemize}

As a comparison, for  multiple classes of languages, the growth rate provides a relevant classification of languages into big (exponential) and small (polynomial) ones, giving important behavioral insights;
a gap between polynomial and exponential growth in regular languages is immediate from \cite{ChomskyMiller}, for context-free languages it was established  in \cite{growth1}. An efficient algorithm for distinguishing between the polynomial  and exponential growth of regular and context-free languages is presented  in \cite{growth2}. Finally, for indexed languages, there is no such gap and an intermediate growth is possible \cite{grigorchuk}. In the same vein but for a slightly different aspect, in \cite{ambiguity}, nondeterministic finite automata are classified into three classes: finitely, polynomially, and exponentially ambiguous; and the classes are characterized by the presence of specific patterns, as in our study.

For timed regular languages, the growth rate  of their volume with respect to the number of events was explored in \cite{entroJourn}. 
In particular, a dichotomy between thin and thick languages was established and related to the Zeno phenomena. In a nutshell, in a thick language (with exponential growth) most timed words  are non-Zeno, and have nice properties (they satisfy  a pumping lemma, admit discretization, etc.),  while in  a thin language, all timed words long enough are, in some sense, Zeno and involve very short delays. 

In this paper, as already said,  we 
look at the growth rate of timed languages \textbf{per time unit}. 
More precisely, instead of volumes,  we study the size of the discretizations of the language, up to precision $\varepsilon$, using $\varepsilon$-nets and $\varepsilon$-separated sets \cite{kolmoEpsilon}. 
The logarithms of these quantities are called, resp., $\varepsilon$-entropy and $\varepsilon$-capacity, and the limits of these quantities divided by the time bound of the language 
are called, resp., entropic and capacitive bandwidth. In \cite{formats2022} we justify the notion of bandwidth proving (up to some details) that a timed source $S$ can be encoded to a timed channel $C$ with bounded delay, iff the bandwidth of $S$ does not exceed the bandwidth of $C$.

In this article, we partition timed  regular languages (and timed automata) into three classes according to the asymptotic behavior of the bandwidth,
as $\varepsilon$ becomes small. We call them \emph{meager}, 
with an $O(1)$ (smallest) bandwidth,  \emph{normal}, with a $\Theta\left(\log\frac{1}{\varepsilon}\right)$ bandwidth, or \emph{obese}, with $\Theta\left(\frac{1}{\varepsilon}\right)$ (highest) bandwidth. 
Intuitively, in an obese language, information can be encoded into time delays with a very high frequency; in a normal language, information can be encoded 
every few time units; finally, a meager language is extremely constrained and information can be mostly encoded into discrete choices.
This trichotomy can be deduced from two structural criteria on timed automata, 
which exactly characterize the asymptotic classes. Using these structural criteria we also prove that deciding to which class belongs a timed automaton is \PSPACE-complete.

Our results build on several ingredients. The bandwidth of a timed language has been defined in \cite{formats2022}, itself based on notions of $\varepsilon$-entropy and $\varepsilon$-capacity of functional spaces as studied in \cite{kolmoEpsilon} and on a pseudo-distance on timed words from \cite{distance}.  
Structural information on timed automata is formulated in terms of orbit graphs from \cite{puri}, which are discrete objects summarizing reachability relations between states. We had to enrich these graphs with some qualitative information in the form of elements from two finite semirings which represent (1) a rough estimate of the delay between symbols in a timed word and (2) how much choice one may have when taking an edge between two
vertices. Structural criteria of meagerness and obesity are formulated in terms of those labels. More precisely, we identify non-meagerness and obesity patterns, 
and prove that a timed regular language is non-meager resp.~obese, whenever it is recognized by a timed automaton with a cycle labeled with an orbit graph containing a corresponding pattern. 

 The lower bounds take thus the form
``whenever a cycle with a pattern is present, the bandwidth is high'', and we prove them by providing a direct construction of a big bunch of $\varepsilon$-separated runs along this cycle.  
The proofs of upper bounds  ``whenever a cycle with a pattern is absent, the bandwidth is low'' are much more involved and rely on Simon's factorization forest theorem, a powerful result relating Ramsey style combinatorics and formal languages theory (see the original article \cite{simon} and a modern presentation in \cite{bojanczyk}). We pinpoint particular properties of Simon factorization 
whenever one of the non-meagerness or obesity patterns is absent, and build a small $\varepsilon$-net  guided by  this factorization.

The paper proceeds as follows: in \cref{sec:prelim}, we recall some notions and tools, and state the problem of classification of timed automata according to their bandwidth. In \cref{sec:obese,sec:meager}, we solve the problem for deterministic timed automata by establishing  two structural criteria --- for very low and very high bandwidth. In \cref{sec:classif}, we come up with a classification theorem (and establish the  complexity of classification), extend it to the non-deterministic case, and compare our new classes with those from \cite{entroJourn}. We conclude with some perspectives. For sake of readability, many technical details and proofs are relegated to  the Appendix.

%% file: background.tex
We start this section by recalling our approach \cite{formats2022} to bandwidth (information by time unit) in timed languages and its background: classical $\varepsilon$-entropy and capacity \cite{kolmoEpsilon} and our pseudo-distance on timed words \cite{distance}. Next, we recall timed automata \cite{AD} and related constructions, including
variants of  the region-split form \cite{entroJourn}
and orbit graphs from \cite{puri}. 
 Finally, we present the central notion of the paper: the three classes of timed automata.

\subsection{Measuring information in  timed languages} \label{sec:entropy}
Given $\Sigma$, a finite alphabet of discrete events, a \emph{timed word} over $\Sigma$ is an element from $\left(\Sigma\times\real_+\right)^*$ of the form $w = (a_1,t_1)\dots (a_n,t_n)$, 
with $0 \leq t_1 \leq t_2 \cdots \leq t_n$. Number $t_i$ should be interpreted as the \emph{date} at which the event $a_i$ happens. We always put $t_0=0$.
For any timed word $w = (a_1,t_1)\dots (a_n,t_n)$, we denote its \emph{discrete length} as $|w|\triangleq n$ and its  \emph{duration} as $\tau(w)\triangleq t_n$. 
The \emph{timing} of $w$ is the sequence $t_1\dots t_n$, and its \emph{untiming} is the sequence $a_1\dots a_n$. 
A \emph{timed language} over $\Sigma$ is a set of timed words over $\Sigma$.

For determining the quantity of information contained in a timed language, we adopt the approach introduced in \cite{kolmoEpsilon} (for compact functional spaces). 
Intuitively, when only words that are very close to one another can be observed, then one may not be able to distinguish many different words, thus little information may be conveyed. On the contrary, if there are many words that are set sufficiently far from each other, the fact of choosing among them conveys a large amount of information.

A second point of view is how many words are necessary to approximate the language up to some precision.
Both points of view depend on the precision with which we observe, as formalized by the following definition, taken from \cite{kolmoEpsilon} 
and adapted to pseudo-metric spaces\footnote{In pseudo-metric spaces one may have $d(x,y)=0$ for $x\neq y$, but all other axioms of the distance hold. Moreover we also allow infinite distances.}:

\begin{definition}
Let $(U,d)$ be a pseudo-metric space and $A \subseteq U$, then:
\begin{itemize}
    \item $M \subseteq A$ is an $\varepsilon$\emph{-separated subset} of $A$ if $\forall x\neq y \in M, d(x,y) > \varepsilon$;
    \item $N \subseteq U$ is an $\varepsilon$\emph{-net} of $A$  if $\forall y \in A, \exists s \in N \text{ s.t. } d(y,s) \leq \varepsilon$.
\end{itemize}
\end{definition}
\begin{figure}[t]
\begin{tikzpicture}[scale=0.4]
 \draw[fill=black] (-10.5,2) -- (10.5,2);
 \draw[fill=black] (-10.5,2.4) -- (-10.5,1.6);
 \draw[fill=black] (10.5,2.4) -- (10.5,1.6);

 \foreach \a in {-13,-11,...,13}
 \foreach[evaluate={\z=int(16*\a*\a+169*\b*\b);}] \b in {2}{
 \ifnum \z<2600
 {
 \draw[black,fill=black](\a,\b) circle(0.08);
 }
 \fi
 }

\draw[blue] (-9,2.35) -- (-9,1.65);
\draw[blue] (-7,2.35) -- (-7,1.65);
\draw[blue,dotted] (-9,2.25) -- node[above]{\large$>\varepsilon$}(-8,2.25) --(-7,2.2);

\draw[red] (-4.2,2.35) -- (-4.2,1.65);
\draw[red] (-5.0,2.35) -- (-5.0, 1.65);
\draw[dotted,red] (-4.2,2.25) -- node[above]{\large$\leq\varepsilon$} (-5.0, 2.25);

\draw[black] (14.2,1.35) -- (14.2,0.65);
\draw[black] (15.8,1.35) -- (15.8, 0.65);
\draw (14.2,1) -- node[above]{$\large\varepsilon$}(15.6,1) --(15.8,1);
 \end{tikzpicture}	
 \caption{The set of black dots is $\varepsilon$-separated and an $\varepsilon$-net in the 
 segment}
 \label{fig:net}
\end{figure}
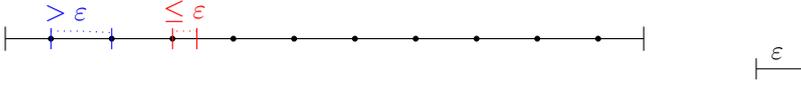

An $\varepsilon$-separated subset of $A$ contains elements of $A$ distant from each other by at least $\varepsilon$, hence distinguishable with precision $\varepsilon$. 
An $\varepsilon$-net of $A$ allows approximating every element of $A$ with precision $\varepsilon$, see \cref{fig:net}. 
They give rise to two measures of the information content:

\begin{definition} 
For $(U,d)$ a pseudo-metric space and $A\subseteq U$, the $\varepsilon$\emph{-capacity}  $\capa_\varepsilon(A)$  and $\varepsilon$\emph{-entropy}  $\ent_\varepsilon(A)$ are defined as follows\footnote{all logarithms in this article are base 2}:
\[
  \capa_\varepsilon(A) \triangleq  \log\max\{\card{M}\mid  M \text{ $\varepsilon$-separated set of } A\};\hfill
  \ent_\varepsilon(A) \triangleq  \log\min\{\card{N} \mid  N  \text{ $\varepsilon$-net of $A$}\}.
\]
\end{definition}
The following inequality \cite{kolmoEpsilon} is known to hold for any $A$ subset of $U$:
\begin{equation}\label{prop:inequality}
\capa_{2\varepsilon}(A) \leq \ent_\varepsilon(A) \leq \capa_\varepsilon(A).
\end{equation}
%
Two basic cases provide an intuition on the sense of $\capa$ and $\ent$ as information measures:
\begin{itemize}
    \item For $\Sigma$  a finite alphabet (with the discrete metrics), $\capa_\varepsilon(\Sigma)=\ent_\varepsilon(\Sigma)=\log \#\Sigma$; in words, a letter in $\Sigma$ brings  $\log \#\Sigma$ bits of information. 
    \item Consider now a segment $[0,1]$, then $\capa_\varepsilon[0,1]\approx \ent_{\varepsilon/2}[0,1]\approx \log 1/\varepsilon$. In words, a real number observed with precision $\varepsilon$ contains $\log 1/\varepsilon$ bits of information, intuitively corresponding to $\log 1/\varepsilon$ digits after the binary point.
\end{itemize}

In the developments thereafter, $U$   will typically be $U_\Sigma$, the universal timed language on the working alphabet $\Sigma$, while different regular timed languages on $\Sigma$ will play the role of $A$. The pseudo-metric will be the one we describe next.


To give a meaningful interpretation of capacity and entropy, an appropriate metric modeling the ability of the observer is required,
that is, the ability to distinguish between two words. To this purpose, we imagine an observer that reads the discrete letters of the word exactly (they can determine whether or not a letter has occurred) but with some imprecision w.r.t.~time, so they cannot determine when two letters are very close to one another, which one came before the other, and not even how many times a letter was repeated within a short interval. 
Hence, we use the following pseudo-distance (similar to the Hausdorff distance between sets):
\begin{definition}[\!\cite{distance}] The \emph{pseudo-distance} $d$ between two timed words 
$w = (a_1,t_1)\dots(a_n,t_n)$ and $v = (b_1,s_1)\dots(b_m,s_m)$ is defined by (with the convention $\min \emptyset = \infty$)
\[
 \dr(w,v)\triangleq \max_{i\in\{1..n\}}\min_{j\in\{1..m\}} \{ |t_i-s_j|:a_i=b_j \};\qquad
 d(w,v) \triangleq  \max( \dr(w,v) , \dr(v,w) ).
\]
\end{definition}
\begin{figure}[t]
\begin{center}
{
\usetikzlibrary {arrows.meta,positioning} 
\begin{tikzpicture}
\draw (0,.8) -- (.7,.8) \aaa  --(1.8,.8) \bbb  -- (3,.8) \aaa-- (4,.8) \bbb-- (4.1,.8) \aaa --
(4.5,.8)node[anchor=west]{$u=(a,.7),(b,1.8), (a,3), (b,4), (a,4.1)$};

\draw (0,0) -- (.6,0) \aaa   -- (1,0) \aaa--(1.7,0) \bbb  -- (3,0) \aaa-- (4.1,0) \aaa-- (4.2,0)\bbb --
(4.5,0)node[anchor=west]{$v=(a,.6),(a,1), (b,1.7), (a,3), (a,4.1),(b,4.2)$};

\draw [dotted, arrows={
->[width=1mm,length=1.2mm,sep=1.8mm]}]
(.6,0) edge (.7,.8)
(.7,.8) edge (.6,0)
(1,0) edge (.7,.8)
(1.7,0) edge (1.8,.8)
(1.8,.8) edge (1.7,0)
(3,0) edge (3,.8)
(4,.8) edge (4.2,0)
(3,.8) edge (3,0)
(4.2,0) edge (4,.8)
(4.1,0) edge (4.1,.8)
(4.1,.8) edge (4.1,0)
;
\end{tikzpicture}
}
\end{center}
\caption{Pseudo-distance between two timed words (dotted lines with arrows represent directed matches between letters).  	$\dr(u,v)=0.2;\  \dr(v,u)=0.3$, thus $ d(u,v)=0.3$.}\label{fig:dist}
\end{figure}
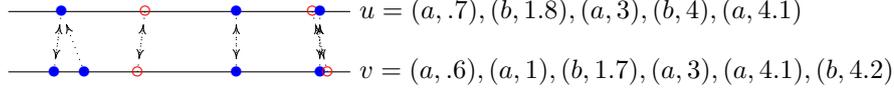
The definition is illustrated on \cref{fig:dist}.
We remark that $d(w,v)=0$ whenever $w$ and $v$ only differ by order and quantity of simultaneous letters, e.g.~for $w=(a,1)(b,1)$ and $v=(b,1)(b,1)(a,1)$. For this reason, $d$ is only a pseudo-distance.

There are two (related) reasons for using this pseudo-distance rather than more straightforward ``uniform'' distances such as defined in \cite{tubes}:
(1) it provides non-trivial values even when the two words $w$ and $v$ have a different number of events; (2) the set of timed words with duration bounded by some $T$ is compact, 
which ensures that finite $\varepsilon$-nets always exist.

We now recall the notion of bandwidth, as introduced in \cite{formats2022}, a characterization of the ability of a language to produce or convey information \emph{per time unit} in the long run. 
Bandwidths come in two versions, depending on whether it is about producing or transmitting information.
Bandwidths are asymptotics, with respect to time, of the growth rate of $\varepsilon$-capacity or $\varepsilon$-entropy, which finally yield quantities expressed in bits per time unit.
We use here the notion of \emph{timed slice} of a language $L$ at some $T\in \real_+$, defined as $L_T \triangleq  \{ w: w \in L, \tau(w) \leq T \}$.

\begin{definition}[\!\cite{formats2022}]The $\varepsilon$\emph{-entropic} and $\varepsilon$\emph{-capacitive bandwidths} of a language  $L$ are 
\[
\bandh_\varepsilon(L) \triangleq \limsup_{T\rightarrow\infty}\frac{\ent_\varepsilon(L_T)}{T}; \qquad
\bandc_\varepsilon(L) \triangleq  \limsup_{T\rightarrow\infty}\frac{\capa_\varepsilon(L_T)}{T}. 
\]
\end{definition}
The following relationship between the two bandwidths is immediate from \eqref{prop:inequality}:
 \begin{equation}\label{eq:inequality}
 \bandc_{2\varepsilon}(L) \leq \bandh_\varepsilon(L) \leq\bandc_\varepsilon(L).  
 \end{equation}

\subsection{Timed Automata}
\emph{Regular} timed languages are timed languages recognized by timed automata introduced in \cite{AD}. In most of this paper, as in \cite{puri,entroJourn} we deal with deterministic timed automata  (\DTA).
For a set of variables $\Xi$, $G_{\Xi}$ denotes the set of constraints expressible as conjunctions of inequalities of the form $\xi \sim b$ with $\xi \in \Xi$, $\sim \in \{<,\leq, >, \geq\}$ and $b$ an integer $\geq 0$.    
\begin{definition}
A  \emph{timed automaton} is a tuple 
$(Q, X, \Sigma, \Delta, S, I, F)$, where
\begin{itemize}
    \item $Q$ is the finite set of discrete locations;
    \item $X$ is the finite set of clocks;
    \item $\Sigma$ is a finite alphabet;
    \item $S,I,F: Q \rightarrow G_{X}$  define respectively the starting, initial, and final clock values for each location;
    \item $\Delta \subseteq Q \times Q \times \Sigma \times G_{X} \times 2^X$ is the transition relation, whose elements are called edges.
\end{itemize}
A timed automaton is \emph{deterministic}  if $\{ (q,x) | x\models I(q)\}$ is a singleton and for any two edges 
$(q,q_1,a,\guard_1,\reset_1)$ and $(q,q_2,a,\guard_2,\reset_2)$ with $q_1 \neq q_2$, the constraint $\guard_1 \wedge \guard_2$ is   non-satisfiable. 
\end{definition}
The role of $S,I,F$ is the following: $S(q)$ denotes the ``starting constraint'' which must be satisfied by clock values when entering location $q$. 
Accepting runs should start in configuration $(q,x)$ with $x$ satisfying $I(q)$ and end by a transition to $(q',x')$ with $x'$ satisfying $F(q')$.

Throughout the paper, we use as running examples the 10 automata described on \cref{fig:5examples}.
\begin{figure}[t] 
\begin{tikzpicture}
\node at (0,6) {$\aut_1:$};
\node[state, initial, accepting] at (1,6) (q1) {$q$};
\draw (q1) edge[loop below] node{$a,b$} (q1);
\node at (3,6) {$\aut_2:$};
\node[state, initial, accepting] at (4,6) (q1) {$q$};
\draw (q1) edge[loop below] node{$a,b,\,x<1$} (q1);
\draw (q1) edge[loop above] node[right]{$c,x\in (5,6),\{x\}$} (q1);
\node at (6.5,6) {$\aut_3:$};
\node[state, initial, accepting] at (7.5,6) (q1) {$q$};
\draw (q1) edge[loop below] node{$a,b,\,x<5$} (q1);
\node at (9,6) {$\aut_4:$};
\node[state, initial] at (10,6) (q1) {$q$};
\node[state, accepting] at (12,6) (q2) {$p$};
\draw (q1) edge[->,above, bend left] node{$a,b,\,x\in(3,4)$} (q2)
(q2) edge[->,below, bend left] node{$\delta_2:b, y\in(5,6), \{x,y\}$}(q1);
%
\node at (0,3.5) {$\aut_5:$};
\node[state, initial, accepting] at (1,3.5) (q1) {$q$};
\node[state, accepting] at (3,3.5) (q2) {$p$};
\draw (q1) edge[->,above, bend left] node{$a,b,\,x=3$} (q2)
(q2) edge[->,below, bend left] node{$b,\, y=5, \{x,y\}$} (q1);
\node at (4.5,3.5) {$\aut_6:$};
\node[state, initial, accepting] at (5.5,3.5) (q1) {$q$};
\node[state, accepting] at (7.5,3.5) (q2) {$p$};
\draw (q1) edge[->, above, bend left] node{$\delta_1:a, x<1, \{x\}$} (q2)
(q2) edge[->, below, bend left] node{$\delta_2:b, y\in(1,2), \{y\}$} (q1);
\node at (9,3.5) {$\aut_7:$};
\node[state, initial, accepting] at (10,3.5) (q1) {$q$};
\node[state, accepting] at (12,3.5) (q2) {$p$};
\draw (q1) edge[->, above, bend left] node{$a, x<1, \{x\}$} (q2)
(q2) edge[->, below, bend left] node{$b, y<1, \{y\}$} (q1);
%
\node at (0,1) {$\aut_8:$};
\node[state,initial] at (1,1) (q) {$q$};
\node[state,accepting] at (3,1) (p) {$p$};
\draw (q) edge[->, above right, bend left] node{$a, x<1, \{x\}$} (p)
(p) edge[->, below, bend left] node{$b, y\in(1,2), \{y\}$} (q)
(q) edge[loop above] node{$c,x<1,y<1$} (q);
\node at (4.5,1) {$\aut_9:$};
\node[state,initial] at (5.5,1) (q) {$q$};
\node[state,accepting] at (7.5,1) (p) {$p$};
\draw (q) edge[->, above right, bend left] node{$a, x<1, \{x\}$} (p)
(p) edge[->, below, bend left] node{$b, y\in(1,2), \{y\}$} (q)
(q) edge[loop above] node{$c,\, x,y<1, \{y\}$} (q);
\node at (9,1) {$\aut_{10}:$};
\node[state,initial] at (10,1) (q) {$q$};
\node[state,accepting] at (12,1) (p) {$p$};
\node[state,accepting] at (11,0) (r) {$r$};
\draw (q) edge[->, above,inner sep=4mm] node{$a, x<1, \{x\}$} (p)
(p) edge[->, below, very near start,inner sep=4mm, sloped] node{\hspace{10mm}$b, y\in(1,2), \{y\}$} (r)
(r) edge[->, below left] node{$c, y<1$} (q);
\end{tikzpicture}

\vspace{-15mm}
\caption{10 running examples of timed automata, $I(q)=\{\0\}$, $F(\cdot)=\mathbf{true}$ for marked locations.} \label{fig:5examples}
\end{figure}
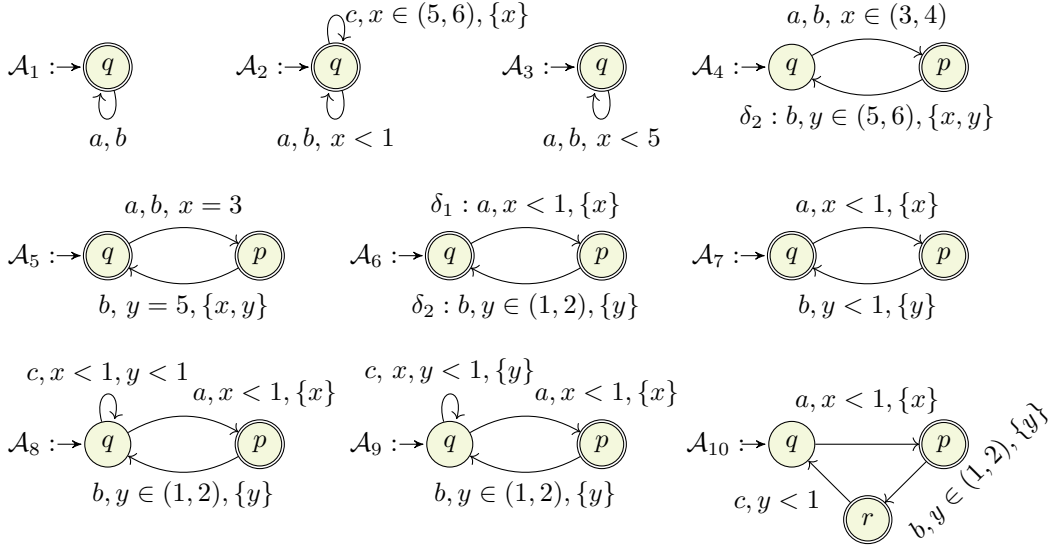

\subsubsection{Semantic details and notations}
We define two operations that can be applied to any clock vector $x \in [0,\infty)^X$: first, given $t\geq 0$, the \emph{timed successor} of $x$ by $ t$, denoted $x+t\in  [0,\infty)^X$, is the clock vector such that for all clocks $c \in X$, $(x+t)_c = x_c+t$; second, given a set of clocks  $\reset\subset X$, the \emph{$\reset$-reset of $x$} is the clock vector $x[\reset]\in [0,\infty)^X$ defined as $x[\reset]_c=0$ whenever $c\in\reset$ and $x[\reset]_c=x_c$ otherwise.

 The following objects will be associated with a timed automaton, and used in the sequel (for all types of sequences, $\circ$ denotes the concatenation and $\epsilon$ the empty sequence
; the concatenation is lifted in the usual way to sets of sequences, seen as languages):
\begin{itemize}
\item For any $\delta\in \Delta$ we define the projections $\src_\delta$ (source), $\dst_\delta$ (destination), $\lbl_\delta$ (label), $\guard_\delta$ (guard) and $\reset_\delta$ (set of reset clocks) such that $\delta=(\src_\delta, \dst_\delta, \lbl_\delta, \guard_\delta, \reset_\delta)$.
\item \emph{Edges} are elements of $\Delta$, \emph{edge sequences} are elements of $\Delta^*$, \emph{paths} are edge sequences $\delta_1\dots\delta_n$ with $\dst_{\delta_i} =\src_{\delta_{i+1}}$ for all $i$. 
\item A \emph{state} is  a pair of a discrete location and of a clock vector, i.e.~$(q, x)\in Q\times [0,\infty)^X$.
\item A \emph{run} is a sequence of the form: $\rho = (q_0,x_0) \trans{(\delta_1,t_1)} (q_1,x_1)\dots \trans{(\delta_n, t_n)}(q_n,x_n)$,
such that for each $i$ it holds that (by convention $t_0=0$)
\begin{equation}\label{eq:run-constraints}
x_i\models S(q_i); \src_{\delta_i}=q_{i-1}; \ 
\dst_{\delta_i}=q_{i};\  x_{i-1}+t_i-t_{i-1}\models\guard_{\delta_i}; \ x_i=\reset_{\delta_i}(x_{i-1}+t_i-t_{i-1}).
\end{equation}
\item The following projections from runs to locations, paths, and timed words are defined: $\src(\rho)\triangleq(q_0,x_0)$, $\dst(\rho)\triangleq(q_n, x_n)$, 
$\Path(\rho)\triangleq\delta_1\dots \delta_n$, $\Word(\rho)\triangleq (\lbl_{\delta_1},t_1)\dots (\lbl_{\delta_n},t_n)$.

 \item If $\rho_1$ and $\rho_2$ are runs, $\rho_1\circ\rho_2$ is defined provided: either $\rho_1=\epsilon$, $\rho_2=\epsilon$ or $\dst(\rho_1)=\src(\rho_2)$.

\item To any edge sequence $\pi$ we associate its set of runs $\Runs(\pi)\triangleq \left\{\rho: \Path(\rho)=\pi\right\}$. This implies $\Runs(\pi)\neq \emptyset$ only when $\pi$ is a path. Also $\Runs(\pi_1\circ\pi_2)=\Runs(\pi_1)\circ\Runs(\pi_2)$.

\item Given a path $\pi$ and two clock vectors $x,x'$ we define $\Runs(\pi,x,x')$ as the subset of $\Runs(\pi)$ starting with clock values $x$ and ending with clock value $x'$. We denote the timed language $\left\{\Word(\rho)| \rho\in\Runs(\pi,x,x')\right\}$ by $L_\pi(x,x')$.
\item We also define the language $L_{\bar{\pi}}(x,x')$, where all the guards are replaced by their closures.\footnote{Note that closing guards in an automaton may create new paths but doing so for a single path is ``safe''.}

\item A timed automaton $\aut$ defines a timed language $L(\aut)$: a timed word $w$ belongs to $L(\aut)$ iff it is in $L_\pi(x,x')$ for some clock vectors $x$ and $x'$ and some path  $\pi$ going from some location $q$ to some location $q'$  with $x \models I(q)$ and $x' \models F(q')$ (so we require also the first state to be initial and the last one to be accepting).
\end{itemize}

For a small example consider a 3-edge path $\delta_1\delta_2\delta_1$ in $\aut_6$. One of its runs is 
\[
(q,0,0) \trans{(\delta_1,0.8)} (p,0,0.8) \trans{(\delta_2, 1.5)}(q,0.7,0)\trans{(\delta_1,1.7)} (p,0,0.2),
\]
 and the corresponding timed word is $(a,0.8), (b,1.5), (c,1.7)$.

\subsection{Region Abstraction and Region-Splitting}
First, let us denote by $M$ the maximal  constant  appearing in the constraints used in the automaton.  Also, for $\xi \in \real$, we denote $\{\xi\}$  its fractional part.
\begin{definition}[\!\cite{AD}]
For a clock vector $x$, we define $x_\infty=\{c\in X|x_c>M\}$ and $x_\nat=\{c\in X| x_c \in \nat\}$.
Two clock vectors $x,y\in\real_+^X$ are \emph{region-equivalent} whenever
\begin{itemize}
    \item $x_\infty=y_\infty$ and $x_\nat\setminus x_\infty = y_\nat\setminus y_\infty$;
    \item  $\forall c\in X\setminus x_\infty, \lfloor x_c \rfloor = \lfloor y_c \rfloor$;
    \item and $\forall c_1, c_2\in X\setminus x_\infty,$ $\{x_{c_1}\}\leq \{x_{c_2}\}$ iff $\{y_{c_1}\} \leq \{y_{c_2}\}$.
\end{itemize}
A \emph{region} is an equivalence class of this relation.
\end{definition}
A region $R$ is \emph{$M$-bounded } whenever $R \subseteq [0,M]^X$.
 Any  $M$-bounded region $R$ is  a simplex of some dimension $d\leq \#X$. We denote by $\Ver(R)$ the set of all  $d+1$ vertices of $R$.

A folklore result says that each timed automaton can be transformed into a normal form, which can be considered as a (timed) variant of the region graph defined in \cite{AD}. 
The following definition adapts this to our framework while ensuring that starting constraints are bounded:   
\begin{definition}
A region-split $M$-bounded TA (\RTA) is a \DTA\ $(Q,X,\Sigma,\Delta,S,I,F)$, such that, for any location $q\in Q$: 
\begin{itemize}
\item the starting constraint $S(q)$ defines a non-empty $M$-bounded region;
\item all states in $\{q\}\times S(q)$ are reachable from the initial state and co-reachable to a final state;
\item either $I(q)=S(q)$ is a singleton\footnote{by definition of DTA this is possible for a unique location $q$} or $I(q)=\emptyset$;
\item for any edge $(q,q',a,\guard,\reset)\in\Delta$,  it holds that  $\left(\{S(q)+t\mid t\in \real_+\}\cap \guard\right)[\reset]=S(q')$.
\end{itemize}
\end{definition}

\begin{restatable}{proposition}{propregionsplit}\label{prop:DTA2RTA}
 For any \DTA, there exists an \RTA\  accepting the same language, with the same number of clocks, and an exponentially larger number of locations. 
\end{restatable}

\subsection{Orbits Monoid}
Now we define the monoid, based on  Puri's orbit graphs \cite{puri}, and summarizing reachability along paths of an \RTA.
We fix some numbering of vertices of regions: for a location $q$, such that the simplex of the  region $S(q)$ has dimension $d_q$, we have       $V(S(q)) = \{v_1^q,\dots,v_{d_q+1}^q\}$.

\begin{definition}\label{def:mon:puri}The finite monoid $\mon$ of p-orbits contains $\0,\1$ and, for any two locations $q$ and $q'$  all triples $\tuple{q,A,q'}$ with $A$ a Boolean $(d_q+1)\times (d_{q'}+1)$-matrix.  All p-orbits with zero matrices are identified with $\0$.
Composition rules for $\0$ and $\1$ are standard; and
\[
\tuple{p,A,p'}\tuple{q,B,q'} =
\begin{cases}
\tuple{p,AB,q'} &\text{ if } p'=q;\\
\0 &\text{otherwise} 
\end{cases}
\]
(matrices are multiplied using Boolean operations\footnote{that is $(AB)_{ij}=\bigvee_k(A_{ik}\wedge B_{kj})$}).
\end{definition}
It is convenient to visualize a p-orbit as a bipartite graph (or an oriented graph whenever $q=q'$), from  vertices of $S(q)$ to those of $S(q')$, with edges corresponding to 1s in the matrix.

Following \cite{puri}\footnote{in fact, \cite{puri} only considers cyclic paths in the region automaton} we associate a p-orbit to each edge sequence.

\begin{definition}[The  p-orbit of a path summarizes reachability] Given an \RTA, we define a function $\gamma$ from edge sequences to $\mon$:
\begin{itemize}
    \item $\gamma(\epsilon)=\1$ and whenever sequence $\pi$ is not a path, $\gamma(\pi)=\0$;
    \item for $\pi$ a non-empty path from $q$ to $q'$ let $\gamma(\pi)=\tuple{q,A,q'}$, with $A_{ij}=1$ iff 
     $L_{\bar{\pi}}(v_i^q,v_j^{q'})\neq\emptyset$.
\end{itemize}
\end{definition}
In graph terms, there is an edge in the orbit graph of $\pi$ from a vertex $v$ to $v'$ iff $(q',v')$ is reachable from $(q,v)$ along (the closure of) $\pi$.
\begin{proposition}[probably \cite{entroJourn}]
$\gamma$ is a monoid morphism.
\end{proposition}
Thanks to this property, we can easily compute $\gamma(\pi)$: for one edge using the definition, for a longer path  by composition: $\gamma(\delta_1\dots\delta_k)=\gamma(\delta_1)\dots\gamma(\delta_k)$.

It follows from the theory developed in \cite{puri}, that  p-orbits are an efficient abstraction for reachability: for each path $\pi$, its p-orbit $\gamma(\pi)$ fully describes the reachability relation 
between a clock vector $x$ in the starting region of $\pi$ and a clock vector $x'$ in its ending region, through a run along $\pi$.
Intuitively, the reason is that   any run from $(q,x)$ to $(q',x')$ can be expressed as a convex combination of runs of surrounding vertices of $S(q)$ to vertices of $S(q')$, summarized in the orbit $\gamma(\pi)$.

Consider the running example $\aut_6$, its \RTA\ has its main cycle\footnote{there are two more transient locations, not considered here} between $q$ and $p$ with $S(q)=(0,1)\times \{0\}, S(p)=\{0\}\times(0,1) $. 
The p-orbits of its edges are the two graphs on the left of \cref{fig:puri-twin}, and their matrix form is:
\[
\gamma(q\to p)=\gamma_1=\tuple{ q, \begin{pmatrix}
1&1\\
1&0
\end{pmatrix},p};
 \ 
 \gamma(p\to q)=\gamma_2=\tuple{p,
\begin{pmatrix}
0&1\\
1&1
\end{pmatrix},q}.
\]
The orbits of the two cycles, i.e.~the two graphs on the right of \cref{fig:puri-twin}, are the following:
\[
\gamma(q\to p\to q)= \gamma_1\cdot\gamma_2=
\tuple{q,
\begin{pmatrix}
1&1\\
0&1
\end{pmatrix},
q};\ 
\gamma(p\to q\to p)=  \gamma_2\cdot\gamma_1=
\tuple{ p,
\begin{pmatrix}
1&0\\
1&1
\end{pmatrix},
p}. 
\]

\begin{figure}
\begin{tikzpicture}
\node at (0,0.7) {\large$\aut_6$};
\node at (2,1) {$q\to p$};
  \node[circle, inner sep=0.5mm,fill=blue] (a) at (0,0) {};
  \node[circle, inner sep=0.5mm,fill=blue]  (b) at (1,0) {};
  \node[circle,inner sep=0.5mm,fill=blue] (c) at (3,0) {};
  \node[circle,fill=blue,inner sep=0.5mm] (d) at (3,1) {}; 
  \draw (a)node[above]{0}
  (b)node[above]{1}
  (c)node[right]{0}
  (d)node[right]{1};
 
  \path[-,ultra thick,blue] 
  (a) edge (b)
  (c) edge (d);
 \path[-{Latex},  thin,black,every loop/.append style=-{Latex}] 
 (a) edge[bend right]  (c)
 (a) edge (d)
 (b) edge (c);
%
%
\node at (6,1) {$p\to q$};
  \node[circle, inner sep=0.5mm,fill=blue] (a) at (4,0) {};
  \node[circle, inner sep=0.5mm,fill=blue]  (b) at (5,0) {};
  \node[circle,inner sep=0.5mm,fill=blue] (c) at (7,0) {};
  \node[circle,inner sep=0.5mm,fill=blue] (d) at (7,1) {};  
  \path[-,ultra thick,blue] 
  (a) edge (b)
  (c) edge (d);
 \path[-{Latex},  thin,black,every loop/.append style=-{Latex}] 
 (c) edge (b)
 (d) edge (b)
  (d) edge (a);
%
%
\node at (8.5,1) {$q\to p\to  q$};
  \node[circle, inner sep=0.5mm,fill=blue] (a) at (8,0) {};
  \node[circle, inner sep=0.5mm,fill=blue]  (b) at (9,0) {};
  \path[-,ultra thick,blue] 
  (a) edge (b);
 \path[-{Latex},  thin,black,every loop/.append style=-{Latex}] 
 (a) edge[bend left] (b)
 (a) edge[loop left] (a)
 (b) edge[loop right] (b)
 ;
\node at (12,1) {$p\to q\to  p$};
  \node[circle,inner sep=0.5mm,fill=blue] (c) at (11,0) {};
  \node[circle,inner sep=0.5mm,fill=blue] (d) at (11,1) {};  
  \path[-,ultra thick,blue] 
  (c) edge (d); 
 
 \path[-{Latex},  thin,black,every loop/.append style=-{Latex}] 
 (d) edge[bend left] (c)
 (d) edge[loop left] (d)
 (c) edge[loop left] (c)
 ;
\end{tikzpicture}
    \begin{tikzpicture}
		\draw (0,0) -- (.9,0) \aaa   -- (1.2,0) \bbb--(1.8,0) \aaa  -- (2.4,0)\bbb -- (2.7,0)\aaa-- (3.5,0) \bbb-- (3.6,0)\aaa-- (4.52,0) \bbb-- (4.59,0)\aaa-- (5.53,0) \bbb-- (5.58,0)\aaa-- (6.535,0) \bbb-- (6.575,0)\aaa
  ;
\end{tikzpicture} 
\caption{P-orbits for $\aut_6$\label{fig:puri-twin} and a typical word in $L(\aut_6)$: $(a,.9)(b,1.2)(a,1.8)(b,2.4)(a,2.7)(b,3.5)(a,3.6)(b,4.52)(a,4.59)(b,5.53)(a,5.58)(b,6.535)(a,6.575)
$\label{fig:thin-twin-run}}
\end{figure}
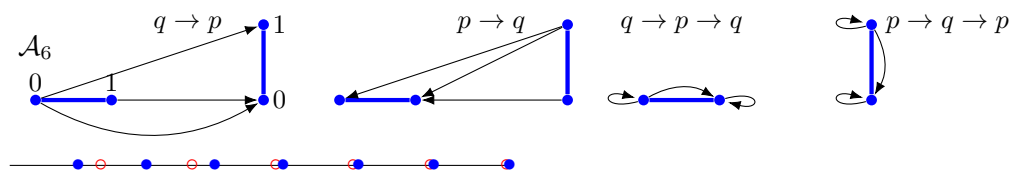


\subsection{Central notion of the paper: the three classes}
%
\begin{definition} A timed language $L$ is 
\begin{itemize}
    \item \emph{meager}
    whenever $\bandh_\varepsilon(L) = O(1)$;
    \item \emph{normal} whenever $\bandh_\varepsilon(L) =  \Theta(\log \frac{1}{\varepsilon})$;
    \item \emph{obese} whenever $\bandh_\varepsilon(L) =  \Theta(\frac{1}{\varepsilon})$,
\end{itemize}
as $\varepsilon\to 0$.
\end{definition}
It follows immediately from \eqref{eq:inequality} that a similar and  equivalent characterization  can be done in terms of capacity by replacing $\bandh$  by $\bandc$.
We will now give some intuition concerning the three classes, using running examples. 
\begin{description}
\item[Automata with obese language:]  information is  produced with a high frequency.  Thus, in $\aut_1$ every $\varepsilon$ seconds one can produce an $a$, a $b$, both  or nothing, which corresponds to two bits of information, and yields a bandwidth of $2/\varepsilon$ bit/sec. In $\aut_2$ such a high-frequency encoding is possible during $1$ second, after that one should wait 4 seconds and take the $c$ loop to reset  $x$. In total, the bandwidth will be $\frac{2}{5\varepsilon}$.
\item[Automata with normal language:] during one unit of time they typically make $O(1)$ discrete choices, and $O(1)$ real-valued choices. The former convey $O(1)$ bit/sec, the latter $O(\log1/\varepsilon)$. Thus $\aut_4$ makes a cycle in 5 to 6 seconds, with two real-valued choices (duration of stay in $p$ and in $q$), and the bandwidth is $\frac25 \log1/\varepsilon$.
\item[Automata with meager language:] In the long run such automata should not allow real-valued choices. $\aut_5$ has only discrete choices since its transitions happen  at discrete dates only. $\aut_3$ can produce a huge amount of information ($2/\varepsilon$ bits every second), but only during the first five seconds of its life. Its bandwidth is $\lim_{T\to \infty} 2/\varepsilon T=0$. The automaton $\aut_6$ is much less evident. It requires $a$ and $b$ to interleave, with $a$ being $<1$ second apart and $b$ on the contrary $>1$ second apart, as shown on \cref{fig:thin-twin-run}. Such interleaving is possible for any duration $T$, but real-valued intervals between symbols  become more and more constrained. We prove later that this language is indeed meager.
\end{description}

In the rest of the paper, we propose a structural and decidable characterization of \DTA\  establishing membership 
of their languages into one of these 3 classes, also proving that no other possibility exists.

%% file: thin.tex
For the next two sections, our plan is to characterize the bandwidth of the timed language
generated by a cycle looping over a region  and to extend this analysis to the whole language.
As bandwidth is about choice per unit of time, we enrich the monoid of p-orbits with one of the following two types of information: 
(1) the amount of choice one may have when taking an edge between two vertices, or (2) a rough estimate of the interval of 
durations of timed words generated by the cycle, starting in one vertex and ending in another vertex.  
The first type of information is useful for discriminating meager from non-meager languages, 
while the second type will distinguish obese from non-obese languages. We start with the former, slightly simpler to formulate. 

\subsection{Orbits with Abstracted Freedom}

The structural criterion for meagerness is based on an enrichment of the monoid of p-orbits in which matrix entries 
summarize the dimension of the reachability polytope between vertices.
In more detail, matrix entries belong to $\Lambda_f\triangleq \{\none,\narrow,\wide\}$, having a semiring structure\footnote{also known as super-Boolean semiring $\mathds{SB}$, \cite{rhodes}},
 defined in \cref{fig:freedom-monoid-plus}.
The intuition 
is the following. Elements of $\Lambda_f$ summarize reachability polytopes between region vertices through some path: no reachability for $\none$ (empty polytope), reachability without continuous choice for $\narrow$ (singleton polytope), presence of at least one continuous choice for $\wide$ (polytope of dimension $\geq 1$). The  addition of $\Lambda_f$ represents the convex union of two compatible sets of runs, while the multiplication corresponds to the concatenation of sets of runs, in particular, two distinct singletons logically add up to a one-dimensional polytope, while the dimension of two concatenated sets of runs is the sum of the dimensions. 
Hence the monoid describes the dimension of the reachability polytope of a path: its degree of \emph{freedom}.

\begin{table}[ht]
\begin{center}
    \begin{tabular}{|c|c|c|c|}
         \hline $+$&$\none$&$\narrow$&$\wide$\\
         \hline
         $\none$&$\none$&$\narrow$&$\wide$\\
         $\narrow$&$\narrow$&$\wide$&$\wide$\\
         $\wide$&$\wide$&$\wide$&$\wide$\\
         \hline
    \end{tabular}
    \hfill
        \begin{tabular}{|c|c|c|c|}
         \hline $\times$&$\none$&$\narrow$&$\wide$\\
         \hline
         $\none$&$\none$&$\none$&$\none$\\
         $\narrow$&$\none$&$\narrow$&$\wide$\\
         $\wide$&$\none$&$\wide$&$\wide$\\
         \hline
    \end{tabular}
 \end{center}   
\caption{Addition and multiplication of the semiring $\Lambda_f$\label{fig:freedom-monoid-plus}}
\vspace*{-20pt}
\end{table}

\begin{proposition}
    $\left(\Lambda_f,+,\times, \none, \narrow\right)$ is a commutative semiring.
\end{proposition}

Thus the definition of the monoid $\mon_f$ of f-orbits is identical with \cref{def:mon:puri}, with the only difference, that  matrix entries belong to the semi-ring $\Lambda_f$. Graphically, we represent all non-$\none$ elements of $\Lambda_f$ by labeled (or colored) edges, and $\none$ by an absence of edge, see \cref{fig:forbits:twins}.

\begin{definition}[The  f-orbit of a path abstracts its freedom] \label{def:f-orbit}Given an \RTA, we define a function $\gamma_f$ from edge sequences to $\mon_f$.
\begin{itemize}
    \item $\gamma_f(\epsilon)=\1$. Whenever sequence $\pi$ is not a path, $\gamma_f(\pi)=\0$.
    \item For $\pi$ a non-empty path from $q$ to $q'$, let $\gamma_f(\pi)=\tuple{q,B,q'}$, with the following matrix entries: 
    for vertices $v=v_i^q\in V(S(q))$ and $v'=v_j^{q'}\in V(S(q'))$
    \begin{itemize}
      \item $B_{ij}=\0$ iff $L_{\bar{\pi}}(v,v')=\emptyset$;
    \item $B_{ij}=\narrow$ iff $L_{\bar{\pi}}(v,v')$ is a singleton;
        \item $B_{ij}=\wide$ otherwise.
    \end{itemize}
\end{itemize}
\end{definition}
We remark that in the last case $L_{\bar{\pi}}(v,v')$ has a dimension $\geq 1$.
\begin{restatable}{proposition}{propForbit}\label{prop:forbit-morphism}
$\gamma_f$ is a monoid morphism.
\end{restatable}
Again, this property allows  compositional computation of $\gamma_f(\pi)$.\footnote{{\cref{def:f-orbit} may not seem effective but, for a single edge $\delta$, $\gamma_f(\delta)$ can be computed in polynomial time in $\#X$, by solving the constraints of $\delta$, while $\gamma_f(\cdot)$ for longer paths can be computed using the multiplication.}}

\begin{definition}
    A cycle of an \RTA\ is \emph{structurally meager} if its f-orbit has no  $\wide$ self-loop (in matrix terms,  a $\wide$ on the main diagonal). An \RTA\ is \emph{structurally meager} when all its cycles are structurally meager.
\end{definition}


Consider the running example  $\aut_6$. The top row of  \cref{fig:forbits:twins}  (which is a colorful version of \cref{fig:puri-twin}) illustrates f-orbits $\gamma_f(q\to p)=\gamma'_1$ and $\gamma_f( p\to q)=\gamma'_2$, in matrix form:
 \[
 \tuple{ q,
 \begin{pmatrix}
\narrow&\narrow\\
\narrow&0
\end{pmatrix},
p};
\text{ and }
\tuple{p,
\begin{pmatrix}
0&\narrow\\
\narrow&\narrow
\end{pmatrix},
q}
.
 \]
and the f-orbits of the cycles $q \to p\to q$ and $p \to q\to p$, i.e.~the products $\gamma'_1\gamma'_2$ and $\gamma'_2\gamma'_1$: 
 \[
 \tuple{q,
 \begin{pmatrix}
\narrow&\wide\\
0&\narrow
\end{pmatrix},
q}; 
\text{ and }
\tuple{p,
\begin{pmatrix}
\narrow&0\\
\wide&\narrow
\end{pmatrix},
p}.
 \]
Similarly, the bottom row of \cref{fig:forbits:twins} presents f-orbits of $\aut_7$, and one can observe that   
  $\aut_6$ is structurally meager (no $\wide$ self-loops), while $\aut_7$ is not.

\begin{figure}
\begin{tikzpicture}
\node at (0,0.7) {\large$\aut_6$};
\node at (2,1) {$q\to p$};
  \node[circle, inner sep=0.5mm,fill=blue] (a) at (0,0) {};
  \node[circle, inner sep=0.5mm,fill=blue]  (b) at (1,0) {};
  \node[circle,inner sep=0.5mm,fill=blue] (c) at (3,0) {};
  \node[circle,fill=blue,inner sep=0.5mm] (d) at (3,1) {}; 
  \draw (a)node[above]{0}
  (b)node[above]{1}
  (c)node[right]{0}
  (d)node[right]{1};
 
  \path[-,ultra thick,blue] 
  (a) edge (b)
  (c) edge (d);
 \path[-{Latex},  thin,black!50!green] 
 (a) edge[bend right]  (c)
 (a) edge (d)
 (b) edge (c);
%
%
\node at (6,1) {$p\to q$};
  \node[circle, inner sep=0.5mm,fill=blue] (a) at (4,0) {};
  \node[circle, inner sep=0.5mm,fill=blue]  (b) at (5,0) {};
  \node[circle,inner sep=0.5mm,fill=blue] (c) at (7,0) {};
  \node[circle,inner sep=0.5mm,fill=blue] (d) at (7,1) {};  
  \path[-,ultra thick,blue] 
  (a) edge (b)
  (c) edge (d);
 \path[-{Latex}, thin,black!50!green] 
 (c) edge (b)
 (d) edge (b)
  (d) edge (a);
%
%
\node at (8.5,1) {$q\to p\to  q$};
  \node[circle, inner sep=0.5mm,fill=blue] (a) at (8,0) {};
  \node[circle, inner sep=0.5mm,fill=blue]  (b) at (9,0) {};
  \path[-,ultra thick,blue] 
  (a) edge (b);
 \path[-{Latex},  thin,black!50!green,every loop/.append style=-{Latex}] 
 (a) edge[bend left,very thick,red] (b)
 (a) edge[loop left] (a)
 (b) edge[loop right] (b)
 ;
\node at (12,1) {$p\to q\to  p$};
  \node[circle,inner sep=0.5mm,fill=blue] (c) at (11,0) {};
  \node[circle,inner sep=0.5mm,fill=blue] (d) at (11,1) {};  
  \path[-,ultra thick,blue] 
  (c) edge (d); 
 
 \path[-{Latex},  thin,black!50!green,every loop/.append style=-{Latex}] 
 (d) edge[bend left,very thick,red] (c)
 (d) edge[loop left] (d)
 (c) edge[loop left] (c)
 ;
\end{tikzpicture}\\
\begin{tikzpicture}
\node at (0,0.7) {\large$\aut_7$};
\node at (2,1) {$q\to p$};
  \node[circle, inner sep=0.5mm,fill=blue] (a) at (0,0) {};
  \node[circle, inner sep=0.5mm,fill=blue]  (b) at (1,0) {};
  \node[circle,inner sep=0.5mm,fill=blue] (c) at (3,0) {};
  \node[circle,fill=blue,inner sep=0.5mm] (d) at (3,1) {};  
  \path[-,ultra thick,blue] 
  (a) edge (b)
  (c) edge (d);
 \path[-{Latex},  thin,black!50!green,every loop/.append style=-{Latex}] 
 (a) edge[bend right]  (c)
 (a) edge (d)
 (b) edge (c);
  \draw (a)node[above]{0}
  (b)node[above]{1}
  (c)node[right]{0}
  (d)node[right]{1};
%
%
\node at (6,1) {$p\to q$};
  \node[circle, inner sep=0.5mm,fill=blue] (a) at (4,0) {};
  \node[circle, inner sep=0.5mm,fill=blue]  (b) at (5,0) {};
  \node[circle,inner sep=0.5mm,fill=blue] (c) at (7,0) {};
  \node[circle,inner sep=0.5mm,fill=blue] (d) at (7,1) {};  
  \path[-,ultra thick,blue] 
  (a) edge (b)
  (c) edge (d);
 \path[-{Latex},  thin,black!50!green,every loop/.append style=-{Latex}] 
 (c) edge[bend left] (a)
 (c) edge (b)
 (d) edge (a);
%
%
\node at (8.5,1) {$q\to p\to  q$};
  \node[circle, inner sep=0.5mm,fill=blue] (a) at (8,0) {};
  \node[circle, inner sep=0.5mm,fill=blue]  (b) at (9,0) {};
  \path[-,ultra thick,blue] 
  (a) edge (b);
 \path[-{Latex},  thin,black!50!green,every loop/.append style=-{Latex}] 
 (a) edge[bend left] (b)
 (b) edge[bend left] (a)
 (a) edge[loop left,red,very thick] (a)
 (b) edge[loop right] (b)
 ; 
 \node at (12,1) {$p\to q\to  p$};
  \node[circle,inner sep=0.5mm,fill=blue] (c) at (11,0) {};
  \node[circle,inner sep=0.5mm,fill=blue] (d) at (11,1) {};  
  \path[-,ultra thick,blue] 
  (c) edge (d); 
 
 \path[-{Latex},  thin,black!50!green,every loop/.append style=-{Latex}] 
 (d) edge[bend left] (c)
  (c) edge[bend left] (d)
 (d) edge[loop left] (d)
 (c) edge[loop left,very thick,red] (c)
 ;
\end{tikzpicture}
\caption{F-orbits for $\aut_6$ and $\aut_7$ ; green thin arrows are $\narrow$, red thick ones $\wide$.} \label{fig:forbits:twins}
\end{figure}
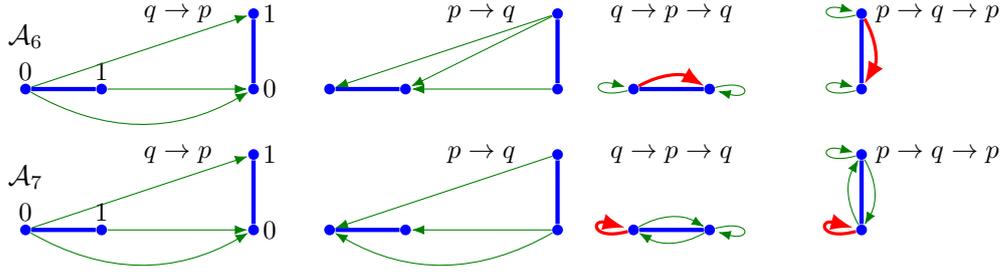

The decision procedure for structural meagerness uses the following ``pumping lemma'':

\begin{restatable}[adapted from {\cite[Lem.~8.7.1]{ocan}}]{lemma}{lemsimonminimallength}\label{lem:simon-minimal-length}
Given a finite alphabet $\Gamma$, a finite monoid $\mon $, and a  morphism $\beta:\Gamma^*\longrightarrow \mon$, for any $e\in \mon $ having  $\beta^{-1}(e)\neq \emptyset$, there exists a word  $w\in\beta^{-1}(e)$ of a length $\leq \#\mon$.
\end{restatable}

\begin{thm}\label{thm:complexity:meager}
    Structural meagerness of an \RTA\ of description size $n$ with $m$ clocks is decidable in $\SPACE(\log^2(n)\poly(m))$. 
\end{thm}
Here  $\poly$ stands for  a polynomial. We will also use $\myexp$ for an exponential function below.
\begin{proof}[Proof sketch]
The  \RTA\ has at most $n$ locations, each orbit can be represented as a couple of locations and an (at most) $(m+1)\times(m+1)$-matrix with elements in $\Lambda_f$ (where $\# \Lambda_f=3$). 
Thus
$
\# \mon_f \leq n^2\cdot 3^{(m+1)^2}=\poly(n)\myexp(m).
$
An automaton is structurally meager whenever  it contains no cyclic path $\pi$ with orbit $\gamma_f(\pi)$ having a $\wide$ entry on the matrix diagonal. By \cref{lem:simon-minimal-length}, only cycles up to length $\#\mon_f$ need to be checked. 
We use now the reachability method   \cite[Sect.~7.3]{papa} to detect such a cycle  in a  space-efficient way, see \cref{algo}. 
An auxiliary function \lstinline{isAPathOrbit(e, h)}, true iff there is a path $\pi\in \gamma_f^{-1}(e)$ of length $\leq 2^h$ in  the \RTA, 
is computed by a divide-and-conquer algorithm. 
\begin{lstlisting}[mathescape=true,language=Java,float,caption=Deciding structural meagerness,label=algo]
Boolean function isAPathOrbit(e, h)
    if h = 0 return (e=$\1$ || $\exists \delta  : \gamma_f(\delta)=$e)
    for all orbits e1, e2 
        if e1$\cdot$e2 = e && isAPathOrbit(e1, h-1) && isAPathOrbit(e2, h-1)
            return true
    return false
 Boolean function isStrucMeager()
    for all cyclic orbits e
        if (containsWideOnDiag(e) && isAPathOrbit(e, $\log \#\mon_f$)) 
            return false
    return true
\end{lstlisting}
The algorithm requires a call stack of depth $\log\#\mon_f$, the size of each stack frame is $O(\log n)+\poly(m)$, thus the 
 computation takes  space   $\left(O(\log n)+\poly(m)\right)\cdot \log\#\mon_f=\log^2(n)\poly(m)$.
\end{proof}

\subsection[Structural Meagerness <=> Meagerness]{Structural Meagerness $\Leftrightarrow$
Meagerness}
The upper bound can be formulated as follows.

\begin{restatable}{proposition}{propmeagerupper}\label{prop:meager:upper}
The timed language of a structurally meager $\RTA$ is meager. 
\end{restatable}

The proof \app\ is done in several steps, two of them using Simon's theorem.  Let $\aut$ be structurally meager. First, we  explore the languages $L((q,x),(q,x))$, of words, looping from a state to itself, and observe that, for any structurally meager cycle $\pi$, $L_\pi(x,x)$ is a singleton. Intuitively, words of $L((q,x),(q,x))$ do only discrete  choices, and thus convey only  $O(1)$ bit/second of information. 

Next, we make a finite partition of all the regions into small ($\varepsilon$-sized) cells and define a set $\Still_\varepsilon$ of words accepted along a cycle from such a cell to the same cell. Such words are close to singletons and thus convey $O(1)$ bit/second of information when observed with precision $\varepsilon$.
On the other hand,  visits  to different cells by the same cycle (and even by different cycles with the same orbit) always respect a certain partial order $\succeq$ over cells.  

Last, we use Puri's reachability characterization \cite{puri} and our technique of linear  Lyapunov functions \cite{entroJourn}, together with Simon's theorem to factorize any word in $L_T(\aut)$ in factors that are either one letter or belong to $\Still^*_\varepsilon$. The total number of factors does not depend on $T$, but only on $\varepsilon$.
This factorization makes it possible to build a small $\varepsilon$-net and conclude.

On the other hand, we need a lower bound.
\begin{proposition}\label{prop:meager:lower}
If an \RTA\ is not  structurally meager then its bandwidth is $\Omega(\log(1/\varepsilon))$.
\end{proposition}
The proof \app\ idea is to take a cycle that is not structurally meager, and show that it alone provides a bandwidth of $\log(1/\varepsilon)$. 
The main theorem of this section is now immediate from \cref{prop:meager:upper,prop:meager:lower}:
\begin{thm}\label{thm:meager} The language of an \RTA\ is meager iff the \RTA\ is structurally meager.
\end{thm}

%% file: obese.tex
 \subsection{Orbits with Abstracted Duration}
The structural criterion for obesity is based on another enrichment of 
p-orbits, in which the  entries of the matrix $A$  belong to a finite semiring $\Lambda_d \triangleq  \{\instant,\fast,\slow,
\none\}$, with operations detailed
in \cref{fig:foreign-monoid-times}.
Intuitively, the elements of $\Lambda_d$ represent 4 abstracted classes of durations 
for sets of runs having the same  path, same source, and same destination vertices; similarly to $\Lambda_f$, the  addition of $\Lambda_d$ represents the convex union of two compatible sets of runs, while the multiplication corresponds to the concatenation of sets of runs.

\begin{table}
\begin{center}
    \begin{tabular}{|c|c|c|c|c|c|}
         \hline $+$&$\instant$&$\fast$&$\slow$
         &$\none$\\
         \hline
         $\instant$&$\instant$&$\fast$&$\fast$
         &$\instant$\\
         $\fast$&$\fast$&$\fast$&$\fast$
         &$\fast$\\
         $\slow$&$\fast$&$\fast$&$\slow$
         &$\slow$\\
         $\none$&$\instant$&$\fast$&$\slow$
         &$\none$\\
         \hline
     \end{tabular}\hfill
     \begin{tabular}{|c|c|c|c|c|c|}
         \hline $\times$&$\instant$&$\fast$&$\slow$
         &$\none$\\
         \hline
         $\instant$&$\instant$&$\fast$&$\slow$
         &$\none$\\
         $\fast$&$\fast$&$\fast$&$\slow$
         &$\none$\\
         $\slow$&$\slow$&$\slow$&$\slow$
         &$\none$\\
         $\none$&$\none$&$\none$&$\none$
         &$\none$\\
         \hline
    \end{tabular}
    \end{center}
\caption{Addition and multiplication of the semiring $\Lambda_d$\label{fig:foreign-monoid-times}}
\vspace*{-20pt}
\end{table}

\begin{proposition}
$ \left(\Lambda_d,+,\times,\none,\instant\right)$ is a commutative semiring.
\end{proposition}
Similarly to $\mon_f$ in the meager case, we define $\mon_d$, the monoid of d-orbits, where the only difference is that matrix entries belong to $\Lambda_d$ instead of $\Lambda_f$. As in previous cases, we represent elements of $\mon_d$ as labeled graphs.

\begin{definition}[The  d-orbit of a path abstracts its speed] Given an \RTA, we define a function $\gamma_d$ from edge sequences to $\mon_d$.
\begin{itemize}
    \item $\gamma_d(\epsilon)=\1$. Whenever sequence $\pi$ is not a path, $\gamma_d(\pi)=\0$.
    \item For $\pi$ a non-empty path from $q$ to $q'$ let $\gamma_d(\pi)=\tuple{q,B,q'}$, with the following matrix entries: for vertices $v=v_i^q \in V(S(q))$ and $v'=v_j^{q'} \in V(S(q'))$
    \begin{itemize}
    \item $B_{ij}=\0$ iff $L_{\bar{\pi}}(v,v')=\emptyset$;
    \item $B_{ij}=\instant$ iff $L_{\bar{\pi}}(v,v')$ is a singleton $\{w\}$, with $\tau(w)=0$;
        \item $B_{ij}=\slow$ iff $L_{\bar{\pi}}(v,v')\neq\emptyset$ and $\tau(L_{\bar{\pi}}(v,v'))\subseteq [1,\infty)$;
        \item $B_{ij}=\fast$  otherwise.
    \end{itemize}
\end{itemize}
\end{definition}
It is not difficult to see that in the last case $\tau(L_{\bar{\pi}}(v,v'))\supseteq [0,1]$.

\begin{restatable}{proposition}{propRorbit}\label{prop:dorbit-morphism}
$\gamma_d$ is a monoid morphism.
\end{restatable}
Again, this property allows a compositional  computation of $\gamma_d(\pi)$.

\subsection{The Obesity Pattern}
\begin{figure}[t] 
\vspace*{-10pt}
\begin{center}
\begin{center}
\begin{tikzpicture}[node distance={20mm}, thick, main/.style = {draw, circle}] 
\node(1) {$u$}; 
\draw[-{Latex}] (1) to[out=135,in=225,looseness=4] node[midway, above =7pt, left =8pt, pos=1] {$\fast$} (1);
\end{tikzpicture}
\hspace{3cm}
\begin{tikzpicture}[node distance={20mm}, thick, main/.style = {draw, circle}] 
\node(1) {$u$}; 
\node(2)[right of=1] {$v$};
\draw[-{Latex}] (1) to[out=135,in=225,looseness=4] node[midway, above =7pt, left =8pt, pos=1] {$\instant$} (1);
\draw[-{Latex}] (1) -- node[midway, above =7pt, left =7pt, pos=1] {$\slow$} (2);
\draw[-{Latex}] (2) to[out=315,in=45,looseness=4] node[midway, below =5pt, right =9pt, pos=1] {$\instant$} (2);
\draw[-{Latex},dashed] (2) to[out=225,in=315,looseness=1] node[midway, below =16pt, right =6pt, pos=1] {} (1);
\end{tikzpicture}
\end{center}
\caption{Obesity patterns. Left: Type I. Right:  Type II.}
\label{fig:var:obesity}
\end{center}
\end{figure}
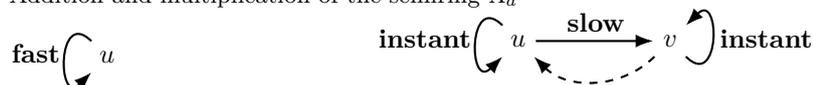
\begin{definition}
A  cycle in an \RTA\ is \emph{structurally obese} iff its d-orbit has the obesity pattern of one of the two types presented on  \cref{fig:var:obesity}, in the sense that either (a) $\gamma_d(\pi)$ has a $\fast$ diagonal element, or (b), for some indices $u$ and $v$, it has two $\instant$ diagonal elements at $(u,u)$ and $(v,v)$, the element at position $(u,v)$ is $\slow$ and some other realizable d-orbit on the same region has a non-$\none$ element at position $(v,u)$.

An \RTA\ is \emph{structurally obese if} it has a structurally obese cycle.
\end{definition}
Some comments are in order.  The above definitions capture the intuition that obese behaviors are ``fast'' behaviors (cycles of unbounded frequency) that can be repeated during an unbounded total duration, and this can happen in two ways only: 
\begin{description}
\item [Type I structural obesity:] there exists a fast cycle that can be repeated with unbounded frequency for an unbounded accumulated duration. 
\item [Type II structural obesity:] there exists a ``Zeno'' cycle $\pi$ (with a d-orbit having a $\slow$ edge from a vertex $u$ to another vertex $v$, and $\instant$ self-loops around both vertices) intersecting with another ``resetting'' cycle $\pi'$ (with a p-orbit having an edge from $v$ to $u$).
The cycle $\pi$ can be repeated with unbounded frequency for a total duration of $\leq 1$ time unit, until $\pi'$ is executed and resets this total, allowing, infinitely often, to spend another time unit in $\pi$.
\end{description}

Let us illustrate this definition on examples from \cref{fig:5examples}. 
Obesity patterns for $\aut_1$ and $\aut_2$ can be easily seen. Consider now $\aut_8$. The d-orbit of the self-loop above its state $q'$ is drawn in \cref{fig:rs-and-orbit-graphs} (right).
The type II obesity pattern occurs because, by taking the path $q'\rightarrow p \rightarrow q \rightarrow q'$, one may 
return from vertex $(1,1)$ to $(0,0)$ of the region $S(q')$.
This orbit has an edge from $(1,1)$ to $(0,0)$, symbolized by the dashed line edge in \cref{fig:rs-and-orbit-graphs} (right).
Finally, the reader can check that $\aut_9$ is not structurally obese. 

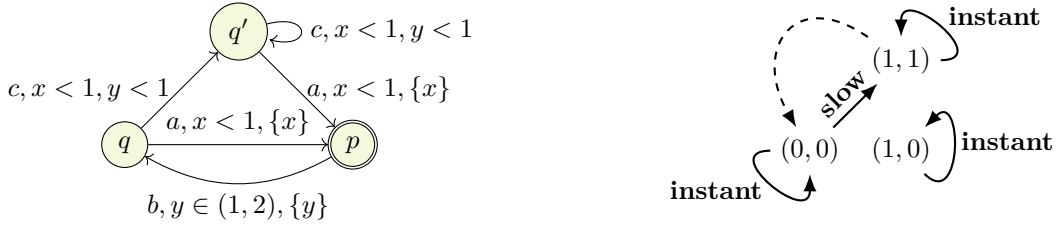
\begin{figure}[t]
\vspace*{-20pt}
\begin{center}    
\begin{tikzpicture}
\node[state] at (0,1) (q) {$q$};
\node[state,accepting] at (3,1) (p) {$p$};
\node[state] at (1.5,2.5) (q') {$q'$};
\draw (q) edge[->, above right] node[above]{$a, x<1, \{x\}$} (p)
(p) edge[->, below, bend left] node{$b, y\in(1,2), \{y\}$} (q)
(q') edge[loop right] node[right]{$c,x<1,y<1$} (q')
(q) edge[->, above right] node[left]{$c, x<1, y<1$} (q')
(q') edge[->, below right] node[right]{$a, x<1, \{x\}$} (p)
;
\end{tikzpicture}
\hfill
\begin{tikzpicture}[node distance={12mm}, thick, main/.style = {draw, circle}] 
\node(1) {$(0,0)$}; 
\node(2)[right of=1] {$(1,0)$};
\node(3)[above of =2] {$(1,1)$};
\draw[-{Latex}] (1) to[out=180,in=270,looseness=3] node[midway, below =6pt, left =14pt, pos=1] {$\instant$} (1);
\draw[-{Latex}] (1) -- node[midway, above,sloped] {$\slow$} (3);
\draw[-{Latex}] (2) to[out=315,in=45,looseness=3] node[midway, below =5pt, right =10pt, pos=1] {$\instant$} (2);
\draw[-{Latex}] (3) to[out=0,in=90,looseness=3] node[midway, above =8pt, right =14pt, pos=1] {$\instant$} (3);
\draw[-{Latex},dashed] (3) to[out=145,in=120,looseness=2](1);
\end{tikzpicture}
\end{center}
\vspace{-10pt}
\caption{Type II obesity for $\aut_8$. Left: its region-split form (transient locations omitted). Right:  d-orbit for the self-loop on $q'$ in the simplex $S(q')$; each tuple represents the value for clocks $x$ and $y$ respectively. The ``resetting'' cycle is $q'\rightarrow p \rightarrow q \rightarrow q'$.}
\label{fig:rs-and-orbit-graphs}
\end{figure}

Similarly to  \cref{thm:complexity:meager}, structural obesity is in  \PSPACE.

\begin{restatable}{thm}{thmcomplexityobese}\label{thm:complexity:obese}
    Structural obesity of an \RTA\ of description size $n$ with $m$ clocks is decidable in $\SPACE(\log^2 (n)\poly(m))$. 
\end{restatable}

\subsection[Structural Obesity <=> Obesity]{Structural Obesity $\Leftrightarrow$ Obesity}
Let us state first the upper bound for non structurally obese automata.
\begin{proposition}\label{prop:obese:upper}
If an \RTA\ $\aut$ is not structurally obese then its  bandwidth  is $O(\log(1/\varepsilon))$.
\end{proposition}
The proof \app\ is quite technical. At its first stage, we carefully explore Zeno cycles, which can convey a lot of information, but only for a short lapse of time. In the second stage we  use Simon's theorem on factorization forests \cite{simon}, and factorize words  of duration $T$ in the language, the number of factors is $\sim T$, some of them are Zeno (their total duration is only $O(1)$), others contain only one event or have duration $0$. Finally, we make a profit from the factorization and build an $\varepsilon$-net for $L_T(A)$ of the required size.

The lower bound is formulated as follows.

\begin{restatable}{proposition}{corstrobeseisobese}\label{prop:strobese}
The language of a structurally obese \RTA\ is obese.
\end{restatable}

The principle of the proof \app\ is simple: we identify a cycle with an obesity pattern, 
and show  that it may be used to transmit information with high frequency (every $\varepsilon$ sec) which yields the required large bandwidth. 
When the clocks become too high, to continue looping through the cycle, a ``resetting'' path is taken.

The main result of this section is now immediate.
\begin{thm}\label{thm:obese}
The language of an \RTA\ is obese iff the automaton is structurally obese.
\end{thm}

%% file: classify.tex
\subsection{Classifying \DTA\ According to Bandwidth Asymptotics}

\begin{restatable}{thm}{thmclassif}\label{thm:classif}
 A timed  language  accepted by a \DTA\  is either meager, normal, or obese. Given a \DTA, the three problems ``is it meager?'', ``is it normal?'', ``is it obese?''  are \PSPACE-complete.
\end{restatable}

\begin{proof}
First, we put the \DTA\ into \RTA\ form, with as many clocks but exponentially more states, as ensured by \cref{prop:DTA2RTA}. 
If   the \RTA\ is structurally meager (resp.~structurally obese), then it is meager (obese) by  \cref{thm:meager,thm:obese}.  
If it is neither structurally meager nor structurally obese, then by   \cref{prop:obese:upper,prop:meager:lower} its bandwidth is both $O(\log(1/\varepsilon))$ and $\Omega(\log(1/\varepsilon))$, hence it is normal.
\PSPACE\ membership follows from
\cref{thm:complexity:meager,thm:complexity:obese}.
\PSPACE-hardness can be proved\app\ by reduction of reachability \cite{AD}.
\end{proof}

\begin{restatable}{corollary}{corextend}
The same classification and complexity results also hold for  nondeterministic timed automata without $\epsilon$-transitions.
\end{restatable}

\subsection{Comparing with Thin-Thick Classification}\label{sec:comparing}
Previously in \cite{entroJourn}, we established a dichotomy between thin (whose language for $n$ events is subexponentially vanishing as $n$ goes to $\infty$) and thick automata (whose language has volumes growing exponentially) and proposed a structural characterization according to orbits: thick automata  are those which contain a reachable \emph{forgetful} cycle, i.e.~a cycle with a p-orbit having a complete graph.

It is natural to compare the approach from \cite{entroJourn} with the new one. 
First, their goals differ: the former assesses how much a time language grows at every discrete event, and the latter handles growth per unit of time. 
Second, in \cite{entroJourn}, only bounded guards were considered, because paths through such guards would have infinite volumes. Third, for a similar reason, paths with punctual transitions were ignored, because of their empty volumes, often entailing a thin verdict for trivial reasons.
In the new approach, paths with punctual guards can contribute to meager, normal, and even obese, behaviors, while the presence of unbounded guards in a path is not a determining factor with respect to bandwidth class (efficient coding strategies tend to favor as short delays as possible). Nonetheless, in the bounded and non-punctual case, it is possible to examine the same automaton under both approaches and obtain the following result:

\begin{restatable}{thm}{thmthinvsmeager}
    An \RTA\ with only bounded non-punctual guards cannot be both thick and meager. Also, it cannot be thin and obese of Type I. All other combinations are possible.
\end{restatable}

\begin{proof}[Proof sketch]
    Consider a thick automaton, with its forgetful cycle $\pi$. If its orbit has a single vertex, since $\pi$ has a non-punctual transition it can be realized by several runs, hence the only edge of $\gamma_f(\pi)$ is $\wide$. 
    Otherwise, $\gamma(\pi)$ has a complete SCC of size at least 2, and then all its vertices have $\wide$ self-loops in $\gamma_f(\pi^2)$.
    
    Now consider an automaton with a structurally obese cycle $\pi$ of Type I. By definition, $\gamma_d(\pi)$ has a $\fast$ self-loop. 
    It can be shown that such a cycle must have a complete p-orbit, hence  the automaton cannot be thin. 

    For all the other combinations we just exhibit as examples the automata from \cref{fig:5examples}:
    \begin{itemize}
        \item 
        $\aut_6$ is thin and meager; 
        \item 
        $\aut_{10}$ is thin and normal;
        \item 
        $\aut_8$ is thin and obese (type II);
        \item 
        $\aut_4$ is thick and normal;
        \item 
        $\aut_1$ is thick and obese (type I).
        \item 
        $\aut_2$ is thick and obese (type II). \qedhere
    \end{itemize}
\end{proof}

%% file: conclusion.tex
We defined three classes of bandwidth and two structural criteria to classify automata into these three classes.
We believe these classes will serve in the theory of timed automata since they correspond to very different kinds of behavior. We believe that techniques specific to the class of automaton would be useful for the random generation of its runs, and improve the performance and scalability of tools such as  \textsc{Wordgen} \cite{wordgen}. Our results have been formulated for timed words (time event sequences), and we think that porting them to timed signals would bring interesting insights.  We believe that  our results can be rephrased in terms of Kolmogorov complexity of $\varepsilon$-approximations of timed words.

In the future, we would like to be able to do more than classify automata: we would also like to compute the constant $\alpha$ that multiplies the main term in the bandwidth formulas. This step will be essential for applications to timed information coding, in the spirit of \cite{timedCoding}.  We hope that many tools introduced in this article will serve to reach this ambitious aim.


%% file: appendix.tex
\section{Preliminary details}

\subsection{Difference bound matrices}\label{sec:DBM}

\emph{Difference bound matrices} --- \DBM \cite{DBM} are a compact way to represent and manipulate systems of difference constraints, they are extensively used in the study of timed automata. We will need only one specific property of a restricted class of \DBM s.

\begin{definition}\label{def:DBM}
    An $n$-\DBM\ is a matrix $A=(a_{ij})_{i,j= 0..n}$ with all $a_{ij}\in\real\cup\{\infty\}$. We say that a vector $x\in\real^n$ satisfies a DBM $A$ whenever for all $i,j= 0..n$ it holds that $x_j-x_i\leq a_{ij}$, where by convention $x_0=0$. We denote the set of all vectors satisfying a DBM $A$ by $[A]$ and call it a \emph{closed zone}.
\end{definition}

Timings of runs in a timed automaton can be described by \DBM{}s, as stated here:
\begin{proposition}\label{prop:timingDBM}
Given a timed automaton with closed guards, a path $\pi$ of length $n$, and clock vectors $x,y$, the set of timings of  words $L_\pi(x,y)$ can be described by an $n$-\DBM\ $A(x,y)=(a_{ij})$ of the following parametric form: 
\begin{itemize}
\item its entries ``in the middle'', i.e.~$a_{ij}$ for $i,j= 1..n-1 $, take values in $\integer\cup \{\infty\}$;
\item entries ``on the border'', that is $a_{0k}, a_{k0}, a_{nk}, a_{kn}$ with $ k\in 1..n-1$ can take values of the form $\min(x_k+b,d)$,
$\min(-x_k+b,d)$, $\min(y_k+b,d)$,
$\min(-y_k+b,d)$ respectively,  with $b,d\in \integer\cup \{\infty\}$;
\item finally $a_{0n}=-a_{n0}=y_k-x_k$ for some $k$ or $a_{0n}=a_{n0}=\infty$.
\end{itemize}
\end{proposition} 
\begin{proof}
    We translate all the conditions on timings in the run as in \eqref{eq:run-constraints} into a \DBM.  We remark that the initial clock  vector $x$ is used to compute clock values, and we check that final clock values coincide with the coordinates of $y$. The translation can be  succinctly described as follows (variables $i$ ranges over $1..n-1$ and $j$ over $1..n$):
\begin{itemize}
\item We initialize all the entries of the \DBM\ at $\infty$;
\item we treat all the constraints in the guards along the path as follows: 
\begin{itemize}
   \item whenever a clock $c_k$, reset at $t_i$, is tested at $t_j$ with the guard $b\leq c_k\leq d$, this can be rewritten as $b\leq t_j-t_i\leq d$, and translated to  the \DBM\   $a_{ij}:=\min(a_{ij},d)$ and $a_{ji}:=\min(a_{ji},-b)$;
\item whenever a non-reset clock $c_k$ (initially equal  to  $x_k$) is tested at $t_j$ with the guard $b\leq c_k\leq d$, this can be rewritten as $b\leq x_k+t_j\leq d$, and translated to the \DBM\   $a_{0j}:=\min(a_{0j},d-x_k)$ and $a_{j0}:=\min(a_{j0},x_k-b)$;
\end{itemize}
\item and  we make sure that the final values of clocks equal $y$ as follows:
\begin{itemize} 
\item if a clock $c_k$ is reset the last time at $t_i$, then the final value of the clock should be attained at $t_n$, i.e.~$t_n-t_i=y_k$, which translates into
$a_{in}:=\min(a_{in},y_k)$ and $a_{ni}:=\min(a_{ni},-y_k)$. 
\item finally, if a clock  $c_k$ is never reset along $\pi$, then $y_k=x_k+t_n$, which yields $a_{0n}=-a_{n0}=y_k-x_k$. \qedhere
\end{itemize}
\end{itemize}    
\end{proof}
A \emph{path} in  a \DBM\ is just a sequence of indices $\sigma=i_1,i_2,\dots,i_k$ all belonging to $0..n$, it is a \emph{cycle} whenever $i_k=i_1$. The sum of a path is defined as follows: $\ssum(\sigma)=\sum_{j=1}^{k-1}a_{i_j,i_{j+1}}$.

We will use the well-known result:
\begin{proposition}\label{prop:proj-DBM}
$[A]=\emptyset$, iff there exists a cycle $\sigma$ with $\ssum(\sigma)<0$. For nonempty $[A]$, its projection on $x_i$ can be characterized as follows: $\proj_i[A]=[L_i,U_i]$ where $U_i=\min\{\ssum(\sigma) \mid \sigma \text{ a path from } 0\text{ to }i\}$ 
and $L_i=\max\{-\ssum(\sigma) \mid \sigma \text{ a path from } i\text{ to }0\}$. The $\min$ and $\max$ can also be taken w.r.t.~simple paths.
\end{proposition}
We remark for knowledgeable  readers that $U_i=a'_{0i}$ and $L_i=-a'_{i0}$ where $A'=(a'_{ij})$ is the canonical form of \DBM\ $A$.

Later on, we will apply the following corollary on Lipshitz continuity.
\begin{corollary}\label{cor:Lipshitz} Suppose  that  two $n$-\DBM{}s $A$ and $A'$ with nonempty $[A]$ and $[A']$ are:
\begin{description}
    \item[equal in the middle:] $a_{ij}=a'_{ij}$ for all $i,j= 1..n-1$;
    \item[$\varepsilon$-close on the border:]  $|a_{0i}-a'_{0i}|<\varepsilon$ for all $i= 0..n-1$; similarly for $a_{i0}, a_{in},a_{ni}$.
\end{description}
and their projections on $x_i$ are $\proj_i[A]=[L_i,U_i]$ and $\proj_i[A']=[L'_i,U'_i]$. Then these projections are close to each other: 
$
|L'_i-L_i|, |U'_i-U_i|<3\varepsilon.
$
\end{corollary}
\begin{proof}
For any simple path $\sigma$ from $0$ to $i$ , at most three weights are different in $A$ and $A'$, and at most by $\varepsilon$: one of the form $a_{0\bullet}$, and possibly one $a_{\bullet n}$ and one $a_{n\bullet}$. Thus $\ssum_A(\sigma)$ can differ from $\ssum_{A'}(\sigma)$ at most by $3\varepsilon$. Since $U_i$ and $U'_i$ are the minima of finitely many such sums, we conclude that $|U'_i-U_i|<3\varepsilon$. The reasoning for $L_i$ is symmetric.
\end{proof}
\subsection{Region-split form}
We establish a preliminary lemma:
\begin{lemma}\label{lem:bounded-starting}
Given a TA $\aut$, it is possible to construct another TA $\mathcal{B}$ recognizing the same timed language, such that its starting constraint for every location is bounded for every clock.
\end{lemma}

\begin{proof}[Proof sketch]
We extend the state space of the automaton $\aut$ 
to obtain the automaton $\mathcal{B}=\left( Q\times 2^X, X, \Sigma, \Delta', S',I',F'\right)$ defined as follows:
    
\begin{itemize}
    \item For all $q\in Q$ and $Y\subseteq X$, $S'((q,Y))$, $I'((q,Y))$ and $F'((q,Y))$ are obtained from, respectively, $S(q)$, $I(q)$ and $F(q)$ by adding, for each clock $c\in X\setminus Y$, the constraint $x_c\leq M$ and replacing, for each clock $c \in Y$, all the constraints involving $x_c$ by $x_c = 0$.
    
    \item For any edge $\delta\in(q,a,\guard,\reset,q')\in\Delta$, and $Y, Y'\subseteq X$, $\Delta'$ contains all edges of the following form: $\left((q,Y), a, \guard \wedge \left(\bigwedge_{c\in Y'} x_c > M\right)\wedge \left(\bigwedge_{c\in X\setminus (Y'\cup \reset)}  x_c\leq M \right), \reset\cup Y', (q',Y')\right)$, except those having a guard with a constraint of the form $x_c < a$ or $x_c\leq a$ for $a\leq M$ and $c\in Y$ (a ``large'' clock cannot satisfy upper bounded constraints).

\end{itemize}           

    Intuitively, we split every state with respect to the set $Y$ of ``large'' clocks, forcing them to reset to $0$ every time a location where this clock is large is entered. $Y$ is updated as follows: when an edge from $\aut$ resets a clock, it becomes non-large (is removed from $Y$) in the destination state of all matching edges of $\mathcal{B}$; when a guard checks that a clock is $>M$, it becomes large (is added to $Y$) in the destination (and thus is reset).

    We obtain an automaton with the same language, but such that the starting constraint of every location is bounded.
\end{proof}

\propregionsplit*
\begin{proof}
We start by constructing an equivalent automaton $\mathcal{B}$ with bounded starting constraints, using \cref{lem:bounded-starting}.
Then we use the construction from \cite{entroJourn} to obtain an \RTA\  $\mathcal{C}$.

The construction goes as follows (simplified for the current setting where starting constraints are already given):
\begin{enumerate}
    \item Split every location $q \in Q$ into sublocations $q_R$ for each region $R$ included in $S(Q)$  (and we set $S(q_R)=R$).
    \item For every $q_R$, set $I(q_R)=I(q)\wedge R$ and $F(q_R)=F(q)\wedge R$.
    \item For every edge from $q$ to $q'$, copy it for every version of $q$ and $q'$ with its guard intersected with the time successor of $S(q)$ and the reset predecessor of $S(q')$ (under the condition that the resulting  guard is satisfiable).
    \item Remove all the locations and transitions not reachable from the initial states or not co-reachable from the final states.
\end{enumerate}

Note that the guards, which are zones, may still contain an unbounded region, projecting to the region $S(q')$ after the reset by $\reset$. Nonetheless, by construction, if we chose  $q,x,t,q',x',\reset$, such that $x\in S(q)$, $x'\in S(q')$ and $\reset(x+t)=x'$, if $R$ is the region of $x$, $R'$ the region of $x'$, then
\begin{itemize}
    \item if $\delta=(q,a,\guard,\reset,q')$ is an edge of $\mathcal{B}$, then there is an edge of $\mathcal{C}$, $\delta' = (q_R,a,\guard',\reset,q'_{R'})$ ($\guard'$ as defined above) such that $x+t\models \guard$  iff $x+t\models \guard'$ (i.e.~one can be fired iff the other can be fired with the same timing, going from the same origin clock vector to the same destination clock vector);
    \item if  there is an edge of $\mathcal{C}$ of the form $\delta' = (q_R,a,\guard',\reset,q'_{R'})$, then there is an edge $\delta$ of $\mathcal{B}$, $\delta=(q,a,\guard,\reset,q')$ is an edge of $\mathcal{B}$, such that $x+t\models \guard$  iff $x+t\models \guard'$.
\end{itemize}

Using this, by induction on the length of runs, we deduce that for any run $\rho$ starting from an initial state of $\mathcal{B}$ there is a run $\rho'$ in $\mathcal{C}$, starting from an initial state, such that both runs have the same word and one ends in an accepting state if and only if the other one does. So the languages of $\mathcal{B}$ and $\mathcal{C}$ are the same.
\end{proof}

%
%
%

\subsection{Closed reachability of region-split automata}

First, we fix some definitions concerning the ``closure'' of various objects depending on clocks. For any constraint on $\real_+$, its closure consists in replacing $<$ and $>$ by $\leq$ and $\geq$.
 The closure of a set of constraints/polytope/zone/region consists in replacing all its constraints with their closure (this definition coincides with the topological closure of the set of points satisfying the constraints). Finally we also define the closure of an edge $\delta=(q,a, \guard,\reset,q')$ as $\bar\delta\triangleq (q,a, \bar\guard,\reset,q')$ and extend this notion to any path $\pi=\delta_1\dots \delta_n$: we denote $\bar{\pi}\triangleq\bar{\delta}_1\dots \bar{\delta}_n$.

Next, let us remind, and adapt to \RTA, a couple of folk results about the shape of successors and predecessors, which will be useful in further developments. In the first one, we deal with regions, and in the second with their closures. 

In both proofs, we use standard properties of regions and zones, extensively used  in the verification of timed automata, see \cite{Bengtsson2004}.
\begin{lemma}\label{lem:region-reach-open}
For an 
\RTA\ 
edge $\delta\in\Delta$, 
\begin{itemize}
    \item the successor of $S(\src_\delta)$ by $\delta$ is exactly $S(\dst_\delta)$;
    \item the predecessor of $S(\dst_\delta)$ by $\delta$ is exactly $S(\src_\delta)$.
\end{itemize}
\end{lemma}
\begin{proof}
By construction of the \RTA, the successor, and the predecessor are non-empty.
 A region is a zone, its successors and predecessors w.r.t.~an edge $\delta$ are  also zones.
 Also, the successor of $S(\src_\delta)$ and the predecessor of $S(\dst_\delta)$ are subsets of, resp., $S(\dst(\delta))$, and $S(\src(\delta))$, which are bounded regions. 
 A non-empty zone subset of a bounded region coincides with this region (since bounded regions are minimal non-empty zones), and we can conclude.
%
\end{proof}
\begin{lemma}\label{lem:region-reach}
For an 
\RTA\ 
edge $\delta\in\Delta$ and a region $R$, 
\begin{itemize}
    \item if $R\subseteq \overline{S(\src_\delta)}$,  then its successor $\{y\mid \exists\; x \in \bar{R} \text{ s.t. } (x,y)\in \Reach_{\bar{\delta}}\}$ is the closure of some region $R'\subseteq \overline{S(\dst_\delta)}$;
    \item if $R\subseteq \overline{S(\dst_\delta)}$,  then its predecessor  $\{x\mid \exists\; y \in \bar{R} \text{ s.t. } (x,y)\in \Reach_{\bar{\delta}}\}$ is the closure of some region $R'\subseteq \overline{S(\src_\delta)}$.
\end{itemize}
Corollary: the same results hold by replacing the single edge $\delta$ by a path $\pi$.
\end{lemma}
\begin{proof}
{
Let us prove first that the successor is non-empty. 
We start by noticing that for $x,y\in[0,M]^X$ if $(x,y)\in \Reach_{\delta}$, then $x\trans{(\delta,t)}y$ with some $t\in[0,M+1]$.
Let $x\in \bar{R}$, and $x_1,x_2,\dots \in S(\src_\delta)$ a sequence  converging to $x$. By \cref{lem:region-reach-open}, for each of those $x_i$ there exists an $y_i\in S(\dst_\delta)$ and $t_i\in[0,M+1]$ such that  $x_i\trans{(\delta,t_i)}y_i$ with all $t_i\leq M+1$. By compactness, one can choose a subsequence of indices $i_k$ with  $y_{i_k}$ converging to some $y'\in \overline{S(\dst_\delta)}$ and $t_{i_k}$ converging to some $t'\in[0,M+1]$. Then $x\trans{(\bar{\delta},t')}y'$, and the successor contains at least $y'$ and is thus non-empty.
The proof for the predecessor is similar.
}

The successor, resp.~predecessor, of $\bar R$ is a non-empty closed zone with integer vertices, included in $\overline{S(\dst_\delta)}$, resp.~$\overline{S(\src_\delta)}$, which are region closures. Therefore it must be a facet of this simplex, hence a region closure.
%
%
\end{proof}

\begin{corollary}\label{cor:orbit-incoming-outgoing}
In the (p,f,d)-orbit $\langle p, M, q\rangle$ of any path, the matrix $M$ has a non-$\0$ entry in every line and every column (i.e., in the graph representation, every vertex of the destination location has an incoming edge and every vertex of the source location has an outgoing edge).
\end{corollary}
The above was already stated in \cite{puri} and remains true in the current context, thanks to \cref{lem:region-reach} (considering singleton regions $\{v\}$, where $v$ is a vertex).

Finally, the following lemma justifies that we can study the information content of the language of a closed path instead of that of its open counterpart:

\begin{lemma}\label{lem:closed-semantics}
In an \RTA, for any  path $\pi$ and clock vectors $x,y$:
\begin{itemize}
 \item $L_{\pi^*,T}(x,y)\subseteq L_{\bar{\pi}^*,T}(x,y)$;
 \item $\ent_\varepsilon(L_{\pi^*,T}(x,y)) = \ent_\varepsilon(L_{\bar{\pi}^*,T}(x,y))$;
 \item $\capa_\varepsilon(L_{\pi^*,T}(x,y)) \leq \capa_\varepsilon(L_{\bar{\pi}^*,T}(x,y))$;
 \item for any $0<\varepsilon'<\varepsilon$, $\capa_{\varepsilon'}(L_{\pi^*,T}(x,y)) > \capa_{\varepsilon}(L_{\bar{\pi}^*,T}(x,y))$.
\end{itemize}
\end{lemma}

\begin{proof}
The first item is immediate, considering that the system of constraints defining $L_{\bar{\pi}^*,T}$ is less constraining than that of $L_{\pi^*,T}$. The third item immediately follows.

The second item is a consequence of the inclusion of the languages (for one direction) and the fact that  $L_{\bar{\pi}^*,T}$ is the topological closure of $L_{\pi^*,T}$ (for the other direction) because closed $\varepsilon$-balls covering $L_{\pi^*}$ also cover its closure.

For the last item, we fix some $t\in L_{\pi^*,T}(x,y)$, as such a run exists, and define the homothety $h_{\lambda}: \real^{|\pi|}\to \real^{|\pi|}, u\mapsto (1-\lambda) t + \lambda u$. We extend the notation $h$ to words of $L_{\bar{\pi}^*}$: 
$h_{\lambda}(w) \triangleq w'$ where $\Let(w')=\Let(\pi)$ and $\timing (w')= h_{\lambda}(\timing (w))$.

Now consider $S$, an $\varepsilon$-separated set of $L_{\bar{\pi}^*,T}(x,y)$ and let us take 2 elements $s_1$ and $s_2$ of $S$. Without loss of generality, there is some index $i$ such that $d(s_1,s_2) = \min_j | s_1[i] - s_2[j] | \geq \varepsilon$. Then the following inequalities hold:
\begin{multline*}
    d(h_\lambda(s_1),h_\lambda(s_2)) \geq \min_j | h_\lambda(s_1)[i] - h_\lambda(s_2[j]) |
    \geq \min_j \left|\lambda (s_1[i]-s_2[j]) + (1 - \lambda) (t[i]-t[j])\right|\\ 
    \geq \min_j \left(\lambda |s_1[i] - s_2[j]| - (1 - \lambda)|t[i]-t[j]|\right)
    \geq \lambda\varepsilon - (1-\lambda)T = \lambda(\varepsilon+T) - T.
\end{multline*}
Hence, if we fix $\varepsilon'\in (0,\varepsilon)$, then the set $S'\triangleq\{h_{ \frac{T+\varepsilon'}{T+\varepsilon} }(w),  w\in S\}$ is $\varepsilon'$-separated. Furthermore $S_{\lambda}\subseteq L_{\pi^*,T}$: indeed, $L_{\pi^*,T}$ is convex and $L_{\bar\pi^*,T}$ is its topological closure.

This finishes proving all the intended results.
\end{proof}

Note this result can only be used to assess the information content of a single path. Indeed, the automaton obtained by closing all the guards of an \RTA\ can exhibit many qualitative behaviors (paths) that were not originally present and thus a bandwidth larger than that of the \RTA.

\subsection {Morphisms to finite monoids and Simon's theorem}
We will need the following fundamental result 
\begin{thm}[Simon\cite{simon,bojanczyk}]\label{thm:simon}
Let $\Gamma$ be a finite alphabet, $\mon $ a finite monoid, and $\beta:\Gamma^*\to \mon $ a monoid morphism. 
Then the sequence of sets 
\begin{equation}\label{eq:simon}
X_0=\Gamma\cup\{\epsilon\};\qquad 
X_{h+1}=X_h\cdot X_h\cup\bigcup_{ee=e}(X_h\cap \beta^{-1}(e))^*
\end{equation}
stabilizes at some finite level $K$ (in fact $K\leq 3\cdot\#\mon $) with $X_K=\Gamma^*$.
\end{thm}

And now we will prove the pumping lemma.
\lemsimonminimallength*
\begin{proof} 
In the hypotheses of the Lemma, let $w=a_1\dots a_N$ be the shortest word in $\beta^{-1}(e)$, and consider, for any $i\in 0..N$, the prefix image $f_i=\beta(a_1\dots a_i)\in \mon$. If $f_i=f_j$ for some $i<j$, then $\beta(a_0\dots a_ia_{j+1}\dots a_N)=e$, which would contradict the minimality of 
$w$. We conclude that all $f_i$ are different, and thus $N+1\leq \#\mon$.
\end{proof}

\subsection{Reachability and Lyapunov functions}
We recall here Puri's characterization of reachability in terms of the orbit graph. The presentation is close to that of \cite{entroJourn}.

Given a region $R$ and a clock vector $x \in R$, we denote $\lambda(x) = \big(\lambda_v(x)\big)_{v\in \Ver(R)}$ the barycentric coordinates of $x$ w.r.t.~the vertices of $R$.
Hence,  $x = \sum_{v\in \Ver(R)} \lambda_v(x) v$.

\begin{lemma}[\!\cite{puri}]\label{lem:puri}
The reachability in  an \RTA~along an edge sequence $\pi$  with $\gamma(\pi) =\tuple{q,A,q'}$ can be characterized as follows. 
For any two clock vectors $x\in S(q)$ and  $x'\in S(q')$, 
we have that $(x,x')\in \Reach_\pi$ iff there exists a stochastic matrix $P\leq A)$ (the inequality is element-wise) 
such that $\lambda(x) P = \lambda(x')$. 
Consequently, two paths have the same p-orbit iff they have the same reachability relation.
\end{lemma}

To be more precise about this citation:
\begin{itemize}

    \item The lemma was originally stated for cyclic paths only, but the fact $\pi$ is a cycle was not used in its proof.
    \item Also only one direction of the implication is stated (reachability $\implies$ existence of a matrix), but the final argument can be reversed: the existence of a matrix allows writing $y$ as a convex combination $\sum_i \lambda_i(x) s_i$ with $s_i=\sum_j P_{ij} w_{ij}$, where the $w_{ij}$ are vertices of $S(\dst(\pi))$ reachable from vertex $v_i$ of $S(\dst(\pi))$. So $y\in \sum_i\lambda_i(x) R(v_i)=R(\sum_i\lambda_i(x) v_i)=R(x)$, where $\forall x, R(x)\triangleq \{y|L_{\bar\pi}(x,y)\neq\emptyset\}$. 

    Note this reasoning (for both directions) uses the key identity $R(\lambda x + (1-\lambda) y) = \lambda R(x) + (1-\lambda)R(y)$, which is proved in another lemma.
   \item Last, the result from \cite{puri} was established for a setting where all regions are bounded.
Fortunately, the original proof can be interpreted in the case where the set of vectors reachable in every intermediary step is an arbitrary zone (in particular, an unbounded region).

Indeed, the identity $R(\lambda x + (1-\lambda) y) = \lambda R(x) + (1-\lambda)R(y)$ just comes from the fact that the time successor and reset operations are affine operations from and to convex sets. This is still the case in our context.



Nonetheless, it is still important that the starting and ending regions of the path are both bounded, for the final result to be stated in terms of vertex-to-vertex reachability and barycentric coordinates.
\end{itemize}

\begin{definition}[\!\cite{entroJourn}]\label{def:Lyap}
 Given a cycle  $\pi$ on a region $R$, that is with $R=S(\src(\pi))$, we call $\ell:\bar R \to [0;1]$ a \emph{Lyapunov function}  for $\pi$ whenever the following holds: if $y\in \bar R$ can be reached from $x\in \bar R$ through $\bar \pi$, then $\ell(x)\geq \ell(y)$.   
\end{definition}

We will slightly generalize the construction of  a Lyapunov function  for a cycle in timed automata from \cite{entroJourn}. 

\begin{con}\label{con:one-lyap}Let  $e$ be an orbit graph and $V=\Ver(R)$ its supporting set of vertices. Let $I\subsetneq V$ be an initial set, i.e.~such that all incoming   edges of $e$ to $I$ have origins in $I$.  Given $x\in \bar R$, let $\lambda_i$ be the barycentric coordinates of $x$ w.r.t.~$V$. We define  the function $\ell_I(x) = \sum_{i\in I}\lambda_i$.
\end{con}
\begin{lemma}
    The construction above always yields a Lyapunov function for any path $\pi$ such that $\gamma(\pi)=e$.
\end{lemma}
The proof is identical to that of \cite[Lemma 11]{entroJourn}.

To have many independent Lyapunov functions, we proceed as follows:
\begin{con}\label{con:lyapu} Given an orbit graph $e$ and its supporting set $V$, we partition $V$ into SCCs, sort these SCCs in a topological order ($S_0,S_1, S_2,\dots, S_s$) and define initial sets $I_i=\bigcup_{j=0}^{i-1} S_j$, their family $\mathcal{I}= \left(I_i\right)_{i\in 1..s}$, and the corresponding family of Lyapunov functions $\Lambda=\{\ell_I \mid I\in\mathcal{I}\}$. We also write $\Lambda(x)$  for the vector $(\ell_I(x))_{I\in\mathcal{I}}$. In further developments, for such a Lyapunov function $\ell_I$, we call the initial set $I$ the \emph{support} of $\ell_I$.
\end{con}

We also remark that $V \setminus I_i$ are final sets of vertices.

\begin{lemma}\label{lem:independent-lyapunov}
Given any cycle $\pi$ such that $e=\gamma(\pi)$,  the construction above yields a family $\Lambda$ of $s=\#\SCC(e) - 1$  independent Lyapunov functions for $\pi$. 
\end{lemma}
The independence means that for any subset of $\Lambda$, values of functions inside the subset do not determine the value of a function outside it. For a set of affine functions, this is equivalent to linear independence of their slopes (and to linear independence of the set $\Lambda \cup \{x\in \bar{R} \mapsto 1\}$).
\begin{proof}
Choose any $i\in 1..s$ and consider the affine function $v:[0;1]\to R$ defined as 
\[
v(\lambda)=\frac{1-\lambda}{\#S_i}\cdot\sum_{j\in S_i}v_j+\frac{\lambda}{\#S_{i+1}}\cdot\sum_{j\in S_{i+1}}v_j.
\]
 The segment $v([0;1])$ is exactly the one that links the barycenter of $S_i$ and that of $S_{i+1}$. For any $j > i$, this segment is entirely included in $\Conv(I_j)$, hence $\ell_{I_j}(v(\cdot))$ is the constant $1$; for $j<i$, the segment is included in $\Conv(\Ver(r)\setminus I_j)$, so $\ell_{I_j}(v(\cdot))$ is the constant $0$. But the image of $\ell_{I_i}(v(\cdot))$ is the full interval $(0;1)$ (indeed $v(0)\in I_i$, so $\ell(v(0))=1$ and $v(1)$ has null barycentric coordinates on $I_i$, so $\ell(v(1))=0$).
Hence $\ell_i$ cannot be obtained by any linear combination of the other Lyapunov functions (i.e.~$\left(\ell_j\right)_{j\neq i}$).
\end{proof}

\begin{corollary}\label{cor:mut-reach-dim}
 Any set of clock vectors mutually reachable through the cycle $\bar\pi$  on a region $R$ can be contained in an affine set of dimension $\dim R - \#\SCC(\gamma(\pi)) + 1$.
\end{corollary}

\begin{proof}
If $x$ and $y$ are mutually reachable, then they should satisfy $\ell(x)=\ell(y)$ for any Lyapunov function $\ell$. Thus any set of mutually reachable  clock vectors is contained in an affine set of the form $\{x\mid\forall \ell\in\Lambda\ \ell(x)=c_\ell\}$, where $\Lambda$ is a set of $\#\SCC(\gamma_\pi)-1$ independent affine functions, hence the dimension.
\end{proof}

\begin{lemma}\label{lem:lyapunov-derivatives}
  In the case when $\#\SCC(\gamma(\pi))=\dim R+1$ (i.e.~all SCCs are singletons), $||x-y||_\infty\leq K||\Lambda(x)-\Lambda(y)||_\infty$ (i.e.~$\Lambda$ is a dilatation), for some positive real number $K$ depending only on $\gamma(\pi)$. In this case, $\Lambda$ is a bijection between $\bar R$ and the simplex $0\leq\ell_1\leq\cdots\leq \ell_{\dim R}\leq 1$.
\end{lemma}

\begin{proof}

We use the notations of \cref{con:lyapu}.

For any clock vector $x$, $\Lambda(x)$ is an affine function of the form $x\mapsto LPB(x-v_s)$ where:
\begin{itemize}
    \item $v_s$ is the unique vertex of $S_s$.
    \item $B$ the square matrix such that $B(x-v_s)$ is the vector of barycentric coordinates of x, excluding the coordinate in $v_s$ (clocks and vertices both in canonical order: i.e.~clocks sorted by fractional part value and vertices sorted lexicographically).




   \item $P$ is the permutation matrix reordering rows such that output vector coordinates correspond to the topological order from \cref{con:lyapu} instead of the lexicographical order.

    \item $L$ is the matrix associating a vector of barycentric coordinates to its Lyapunov values, one per line. Hence $L$ is such that its $i$th row has a $1$ for each position of index v such that $v\in I_i$ and $0$ everywhere else.
    Because vertices are in topological order and each SCC has exactly one more vertex than the previous one, $L$ actually is the lower triangular matrix having only $1$'s in the diagonal and below.

\end{itemize}

Since all three matrices $L$, $P$ and $B$ are invertible, we can write $ (x-v_s) = (LPB)^{-1}\Lambda(x)$, hence $ (x-y) = (LPB)^{-1}(\Lambda(x)-\Lambda(y))$ and thus $||x-y||_\infty\leq ||B^{-1}||_\infty\cdot||P^{-1}||_\infty\cdot||L^{-1}||_\infty\cdot||\Lambda(x)-\Lambda(y)||_\infty$.
\end{proof}

\section{Proofs on obesity}
For technical reasons (in particular, since \cref{lem:singletons} depends on \cref{lem:instant+slow=forgetful}), we proceed first with proofs of results from \cref{sec:obese}, and next from \cref{sec:meager}.

In this section and the next one, for $\bullet\in\{d,f\}$, we use the notation $\gamma_\bullet(\pi)[u,v]$ to denote the label (in $\Lambda_\bullet$) associated to the edge from vertex $u$ to vertex $v$ in the $\bullet$-orbit of a path $\pi$ 
(if there is no such edge, then $\gamma_\bullet(\pi)[u,v]=\none$).


\subsection{Monoid properties}
\propRorbit*
\begin{proof}
We only need to prove, for any two edge sequences $\pi_1$, $\pi_2$, that $\gamma_d(\pi_1\circ\pi_2)=\gamma_d(\pi_1)\gamma_d(\pi_2)$.
Notice that this is trivially true if any of the two sequences is $\epsilon$, or if $\pi_1\circ\pi_2$ is not a path.
So assume both $\pi_1$ and $\pi_2$ are paths of length $\geq 1$ such that $\pi_1\circ\pi_2$ also is a path and let us call $R$ the destination region of $\pi_1$.

Let us choose any source vertex $u$ and any destination vertex $v$ for $\pi_1\circ\pi_2$. Then $\left(\gamma_d(\pi_1)\gamma_d(\pi_2)\right)[u,v]$ has to take one of the four values of $\Lambda_d$:

\begin{itemize}
    \item If it is $\instant$, it means there is at least one vertex $w$ of $R$ such that  $\gamma_d(\pi_1)[u,w]=\instant$ and $\gamma_d(\pi_2)[w,v]=\instant$. By definition, for a run from $u$ to $v$ through $w$, the prefix to $w$ has duration $0$ and also the suffix from $w$, hence the whole run lasts $0$. Moreover all vertices $w$ of $R$ are either like this or such that either $\gamma_d(\pi_1)[u,w] = \none$ or $\gamma_d(\pi_2)[w,v] = \none$, in which case there is no run from $u$ to $v$ through $w$. Hence all runs from $u$ to $v$ last $0$ second and thus $\gamma_d(\pi_1\circ \pi_2)[u,v]=\instant=\left(\gamma_d(\pi_1)\gamma_d(\pi_2)\right)[u,v]$.
    \item If it is $\slow$, it means there is at least one vertex $w$ of $R$ such that  \linebreak$\{\slow\}\subseteq\{\gamma_d(\pi_1)[u,w], \gamma_d(\pi_2)[w,v]\} \subseteq \{\instant, \fast, \slow \}$. In this case, there is a run through $w$ such that either its prefix or suffix lasts more than $1$, hence the full run lasts more than $1$. Moreover all vertices $w$ of $R$ are either like this or such that either $\gamma_d(\pi)[u,w] = \none$ or $\gamma_d(e)[w,v] = \none$, in which case there is no run from $u$ to $v$ through $w$. Hence all runs from $u$ to $v$ last more than $1$ second and thus $\gamma_d(\pi_1\circ \pi_2)[u,v]=\slow=\left(\gamma_d(\pi_1)\gamma_d(\pi_2)\right)[u,v]$.
    \item If it is $\fast$, it means either there is at least one vertex $w$ of $R$ such that $\{\fast\} \subseteq \{\gamma_d(\pi_1)[u,w], \gamma_d(\pi_2)[w,v]\} \subseteq \{\instant, \fast\}$ or there exists $w_1$ and $w_2$ such that \linebreak$\gamma_d(\pi_1)[u,w_1]=\instant$ and $\gamma_d(\pi_2)[w_1,v]=\instant$ and \linebreak$\{\slow\}\subseteq\{\gamma_d(\pi_1)[u,w_2], \gamma_d(\pi_2)[w_2,v]\} \subseteq \{\instant, \fast, \slow \}.$ In the first case, there are runs through $w$ with durations $0$ and $1$. In the second case, there is a run through $w_1$ with duration $0$ and a run through $w_2$ with duration $\geq 1$. Hence in all cases there are runs of length $0$ and $1$ (and therefore all durations in between), thus $\gamma_d(\pi_1\circ \pi_2)[u,v]=\fast=\left(\gamma_d(\pi_1)\gamma_d(\pi_2)\right)[u,v]$.
    \item If it is $\none$, it means that $\{w| \gamma_d(\pi_1)[u,w]\neq \none\}\cap \{w| \gamma_d(\pi_2)[w,v]\neq \none\} = \emptyset$. By \cref{lem:puri}, it implies there is no clock vector in $R$ 
    such that it can be reached from $u$ through $\pi_1$ and that $v$ can be reached from it through $\pi_2$, in other words there is no run from $u$ to $v$ through $\pi_1\circ\pi_2$. Hence, $\left(\gamma_d(\pi_1)\gamma_d(\pi_2)\right)[u,w]= \none = \left(\gamma_d(\pi_1)\gamma_d(\pi_2)\right)[u,v]$ \qedhere
\end{itemize}



\end{proof}

\begin{lemma}[Orbits of zero-time transitions]\label{lem:zero-time-edges}
Let $\delta\in \Delta$ be an \RTA\ edge that can be traversed in $0$ time units from the vertex $v_0$ having the lowest coordinates in its source region. Then for any source vertex $u$ there is a destination vertex $v$ such that $\gamma_d(\delta)[u,v]\in\{\instant, \fast\}$.
\end{lemma}
\begin{proof}
Let us just observe that $v_0\models \guard_\delta$. For any source vertex $u$, there is a delay $d\geq 0$ such that $u+d\models \guard_\delta$ (the automaton is region-split), so $\guard_\delta$ is satisfied by both $v_0$ and $u+d$. Due to the shape of the guards, it means it is also satisfied by all clock vectors that are greater or equal to that of $v_0$ and smaller or equal to that of $u+d$. This rectangle contains the coordinates of $u$, hence $u\models\guard_\delta$.

So $\delta$ can be traversed in zero time from $u$. By \cref{lem:puri}, there is at least one such run going to some vertex $v$ of the target region, such that $\gamma_d(\delta)[u,v]\in \{\instant, \fast\}$.
\end{proof}

\begin{lemma}[Self-loops on a cycle]\label{lem:d-orbit-self-loops}
    For a cyclic path $\pi$, if its d-orbit contains a vertex $u$ with a zero-timed ($\instant$ or $\fast$) self-loop, then
    \begin{itemize}
    \item if it contains a vertex $v$, different from $u$, with a $\slow$ self-loop, then the d-orbit also contains a $\slow$ edge from $u$ to $v$ and a zero-timed ($\instant$ or $\fast$) edge from $v$ to $u$;
    \item it cannot contain a vertex $v$, different from $u$, having a $\fast$ self-loop. 
    \end{itemize}
\end{lemma}

\begin{proof}
We assume $v\neq u$; $\gamma_d(\pi)[u,u]\in \{\instant, \fast\}$ and $\gamma_d(\pi)[v,v]\in \{\fast,\slow\}$. 

    Let $x=(x_1,\dots,x_c)$ and $y=(y_1,\dots,y_c)$ be the clock vectors for $u$ and $v$ respectively.
    As there is a positive-timed self-loop on $v$, it is possible to take the cycle $\pi$ in a non-immediate manner (with duration $d > 0$). For this to be possible, every  clock has to be reset during $\pi$ (otherwise there is some $i$ such that $y_i = y_i + d$, which is impossible).
    Since all  clocks are reset during the cycle, doing it with duration $0$ necessarily leads to the vertex of clock coordinates $(0,\dots,0)$, hence $x=(0,\dots,0)$.
     Moreover, the existence of this zero-time self-loop implies the possibility to traverse each transition without waiting. 
    This means, starting from $(0,\dots,0)$, each automaton edge can be traversed without waiting and thus staying in $(0,\dots,0)$. By \cref{lem:zero-time-edges}, this implies that for each automaton edge in $\pi$, for all source vertices, there is an $\instant$ or $\fast$ outgoing edge.
     So, if we concatenate the orbits, all vertices of the d-orbit of $\pi$ have a zero-time outgoing edge. Because the cycle resets all the clocks, these edges all necessarily go to $u$. This also implies that there cannot be a zero-time realization of the self-loop around $v$, hence this self-loop was actually $\slow$.

    For the edge from $u$ to $v$, consider the cycle started at $u$, with the same delays as the ones of the concrete path realizing the self-loop of $v$
     (it is possible to take such a cycle, as all clock values are smaller in $u$ than in $v$). The clock vector after such a run of the cycle is the same as $v$'s (every  clock has been reset, so its value in $v$ is the sum of delays after its last reset). So $v$ can be reached from $u$ and there is an edge from $u$ to $v$. This edge cannot be $\fast$ or $\instant$ because of the uniqueness of the zero-timed run from $u$.
\end{proof}

\begin{lemma}\label{lem:trivial-instant} 
In any idempotent d-orbit,
$\instant$ loops only exist in singleton SCCs.
\end{lemma}
\begin{proof}
Let $o$ be an idempotent d-orbit. Assuming $o$
has a vertex $u$ with an $\instant$ loop $e$ in an SCC contains several vertices. Thus $u$ must have another outgoing edge, which must be $\slow$. But one can go back from the target of this edge, hence there is a $\slow$ circuit $c$ of length $l$ going through $u$, but $o=o^l$ contains both $c$ and $e^l$, hence the label of the loop through $u$ should have been $\slow+\instant^l=\fast$ and not $\instant$.
\end{proof}

Finally, let us state an intermediary result, an immediate corollary of \cref{lem:d-orbit-self-loops}, that we use in several proofs in this subsection.
\begin{lemma}\label{lem:instant+slow=forgetful}
A cycle having an idempotent p-orbit and a d-orbit with a zero-timed self-loop and a positive-time  self-loop actually has a complete p-orbit. If the d-orbit is also idempotent, then the cycle is structurally obese.
\end{lemma}

\begin{proof}
    In an idempotent p-orbit of a region cycle, all initial and terminal vertices have self-loops, so in the d-orbit of the same cycle, all these vertices have 
    loops. But the condition implies, according to \cref{lem:d-orbit-self-loops}, that all vertices with self-loops are in the same SCC. 
    Hence the d-orbit is strongly connected,
    thus the p-orbit graph is complete. Since the vertices with zero-time loops also belong to positive-time cycles, if the d-orbit is idempotent, it means the self-loop was both zero-time and positive-time, hence $\fast$, which is a realization of the obesity pattern.
\end{proof}

\subsection{Upper bound}
Given a region-split automaton that is not structurally obese -- or \nsoA\ for short -- 
we want to show  that its bandwidth is $O(\log(1/\varepsilon))$ (and hence it is not obese).  

At the \textbf{first stage} we explore cycles in  the automaton corresponding to idempotents of $\mon_d$, with a particular interest to Zeno behaviors (defined below). Let us split idempotent d-orbits of any automaton into four disjoint sets:
\begin{description}
\item[$E_0$] (null) its elements have no realization;
\item[$E_I$] (instant) its elements are realizable and all their realizations have duration $0$;
\item[$E_F$] (fast) its elements have at least one realization of duration in $(0;1)$;
\item[$E_S$] (slow) its elements are realizable and any  realization has duration $\geq 1$.
\end{description}

We expect and will show in the following, that elements of $E_I$ should not be able to create much bandwidth, because the distance  collapses simultaneous events together; that those of $E_S$ should not be able to do so either, because they are not fast enough; and, for obvious reasons, non-realizable orbits in $E_0$ do not contribute at all. The case of $E_F$ is more involved. Intuitively, its elements correspond to fast behaviors which, in the case of non-structurally obese automata, should not be able to run for an unbounded duration, as they would otherwise be a source of obese bandwidth. More specifically, we will prove that elements of $E_S$ satisfy the Zeno property:
\begin{definition}[Zeno property]
A d-orbit $e$ is Zeno when, for any run factorized as $\zeta_0\rho_1\eta_1\zeta_1\rho_2\eta_2\dots\rho_n\eta_n\zeta_n$ such that $\gamma_d(\rho_i)=\gamma_d(\eta_i)=e$ for all $i$, we have that $\sum_i \tau(\rho_i)\leq 1$.
\end{definition}
I.e.~the accumulated duration of all occurrences of cycles with orbit $e$, except the last one in each pack, is bounded by 1. The last iteration needs to be excluded for purely technical reasons.

To establish that elements of $E_f$ are Zeno, we formulate a sufficient condition (in terms of the monoid) --- structural Zenoness, show that in a \nsoA, elements of $E_F$ satisfy this property (\cref{lem:st-non-obese&fast=>st-zeno}), and finally that structural Zenoness implies  Zenoness (\cref{lem:st-non-obese&st-zeno=>zeno}).
\begin{definition}
A d-orbit is \emph{structurally Zeno} iff all its self-loops are $\instant$ and it has at least one $\slow$ edge.
\end{definition}
\begin{lemma}\label{lem:st-non-obese&fast=>st-zeno}
For an \nsoA, all elements $e$ in the set $E_F$ are structurally Zeno.
\end{lemma}

\begin{proof}
First, observe that any edge with a $\fast$ self-loop would already be an instance of the obesity pattern (with $u=v$). Moreover, according to \cref{lem:instant+slow=forgetful}, if there were both an $\instant$ and a $\slow$ self-loop in an idempotent d-orbit, the orbit would also contain the obesity pattern.

Let us take any element $e\in E_F$.

The d-orbit $e$, being in $E_F$, must have both a zero time edge and a positive time edge. If they are the same edge, it is a $\fast$ edge, which cannot be a self-loop (as we just observed), so there must be 2 distinct vertices. If the zero time and positive time edges are different edges, they must either have different sources or different destinations (otherwise they are the same, $\fast$, edge), hence there are at least 2 vertices in all cases. Either both these vertices have a self-loop, or at least one does not (we call such a vertex  transient), but in this case, it is reachable and co-reachable from 2 different vertices having self-loops. In all cases, there are 2 vertices with self-loops, which cannot be $\fast$.

But it is also impossible for all self-loops to be $\slow$: since
$e$ is idempotent and contains a zero-time ($\instant$ or $\fast$) edge,  then it must contain an arbitrarily long (non-elementary) zero-time path of edges, which may only have zero-time factors. However, sufficiently long paths  must contain an (also zero-time) cycle, the starting vertex of which must have a zero-time self-loop (because of idempotency), which contradicts the hypothesis.

Thus necessarily, all loops in 
$e$ are $\instant$. Just notice in this case that all SCC of 
$e$ are singletons (\cref{lem:trivial-instant}). Additionally, we also know that non-self-loop outgoing edges of 
$e$
from vertices having a self-loop must be $\slow$.
\end{proof}

Before proving that structural Zenoness implies Zenoness, let us describe the precise shape of structurally Zeno orbits.

\begin{lemma}\label{lem:zeno-bounds}
    Let $e$ be a structurally Zeno d-orbit. Then, for all  clocks $c$
    there is a natural number $\Low_c$ such that, in all regions traversed by any realization $\pi$, $c$ has values staying in an interval $[\Low_c,\Hi_c)$ (where $\Hi_c=\Low_c+1$).

    For clocks $c$ reset in some realization of $e$ $\Low_c=0$ and $\Hi_c=1$.

\end{lemma}
\begin{proof}
    Any clock $c$ never reset in any realization $\pi$ must have increasing values along any run of $\pi^*$. So it is not possible for $\pi$ to visit regions with different intervals for $c$ (since different intervals are disjoint): otherwise, let $I$ and $J$ be two different visited intervals; wlog we assume all elements of $I$ are smaller than all elements of $J$; after visiting $J$ in one iteration of $\pi$, $I$ would be visited in the next iteration, with $c$ getting a strictly smaller value than one obtained before when visiting $J$. Since the common interval is defined by region constraints, the interval must be a singleton or an open interval of length 1 and natural bounds. Singleton is impossible (this would imply all runs have duration $0$, so there could not be any $\slow$ edge), so it is open of length 1  (indeed, at least the starting region is bounded, and all clocks in this region have a value either contained in a singleton or in an interval of length 1). 
    
    Since all paths realizing $e$ share at least a common first region, the interval is common to all regions traversed by all realizations. 
    
    So the statement is true for non-reset clocks.

    Observe that there has to be at least one non-reset   clock in any $\pi$ realizing $e$: indeed a zero-time run along $\bar\pi$ is possible from any non-transient vertex (all loops are $\instant$); but with all  clocks being reset, such runs all go to $(0,\dots,0)$, implying there is only one self-loop, which is only possible if there is only one vertex, which contradicts there must also be a $\slow$ edge.
    
    So, there is one  non-reset clock, and its value remains contained in an open interval of length $1$, common to all regions traversed by any realization of $e$. This implies the total duration of any run along any factor of an element of $\Paths(e)^*$ is strictly smaller than $1$, this is true in particular if the run starts from a region where a clock was just reset. So clocks reset during $\pi$ can have values in $[0,1)$ only.
\end{proof}

\begin{lemma}[Shape of structurally Zeno d-orbits]\label{lem:zeno-shape}
Let $e$ be a structurally Zeno d-orbit. We define:
\begin{itemize}
    \item  $X_{nr}$ as the set of  clocks not reset by any realization of $e$;
    \item and $X_r$ as those that are reset by some realization;
    \item $v_0$, the vertex of $e$ such that for any clock, $(v_0)_c = \Low_c$ (as defined in \cref{lem:zeno-bounds});
    \item $V_h$, the set of vertices $v$ such that for any clock $c$ in  $X_{nr}$, $v_c=\Hi_c$;
    \item $v_1$, the vertex of $V_h$ such that for any clock in $X_r$, $(v_1)_c = \Low_c$.
\end{itemize}
Then:
\begin{itemize}
    \item $v_0\in \Ver(e)$ and $v_1\in \Ver(e)$; 
    \item the vertices $v$ such that, for all clock $c\in X_r$ $v_c = 0$, have an $\instant$ self-loop and are the only ones with a self-loop;
    \item all other vertices $v$ are transient, belong to $V_h$ and only have one 
    outgoing edge; this edge is $\instant$ and goes to $v_1$;
    \item from $v_0$ to any vertex $v\in V_h$ there is a $\slow$ edge and all the $\slow$ edges are like this.
\end{itemize}
Furthermore, if the automaton is not structurally obese, then there is no 
edge from $v_1$ to $v_0$ in any realizable p-orbit.
\end{lemma}

\begin{proof}
Let us define $e'\triangleq \gamma_{np}(e)$ (the p-orbit common to all paths of $\Paths(e)$,  it is idempotent).

Observe that  clocks that are positive for at least one vertex with a self-loop (necessarily $\instant$) must be in $X_{nr}$ (otherwise the outgoing zero-time edge would have to target a vertex having $0$ as coordinate for this clock, hence this edge would not be a self-loop).

Also, edges in $e'$ can only go from a vertex to another one with equal coordinates on $X_{nr}$ (for an $\instant$ edge in $e$) or one with all coordinates on $X_{nr}$ strictly higher (for a $\slow$ edge in $e$), this means that for a vertex to have an outgoing $\slow$ edge in $e$, it needs to have all coordinates in $X_{nr}$ set to $\Low$ (let us call the set of such vertices $V_l$), otherwise a delay of 1 time unit would go outside $R$, the support region of $e$. Also, this edge must go to a vertex where all clocks in $X_{nr}$ are $\Hi$ (i.e.~a vertex of $V_h$). Remark that at least one $\slow$ edge actually exists from $V_l$ to $V_h$, otherwise $e$ would not be structurally Zeno. Moreover, due to this structure, $V_h$ contains a terminal SCC of $e'$, hence at least one non-transient vertex (idempotency), but such a vertex can only be the vertex $v_1$ having all clocks in $X_r$ set to $0$. Since this is the only vertex in $V_h$ that is terminal for $e'$, it is reachable by all vertices in $V_h$. By idempotency, it means there is a $\slow$ edge from $V_l$ to $v_1$. Dually, $V_l$ contains an initial SCC of $e'$, hence a non-transient vertex, which is the unique vertex $v_0$ such that all clocks in $X_r$ have value $0$. $v_0$ has an outgoing $\slow$ edge into $V_h$, and repeating the argument of idem-potency, it has a $\slow$ edge to $v_1$.

Moreover, since $v_1$ has all clocks in $X_r$ equal to $0$ and those in $X_{nr}$ equal to $\Hi$, inside the region, clocks in $X_r$ must have a smaller fractional part than those in $X_{nr}$. Therefore the component of edges having coordinates in $X_{nr}$ set to $\Low$, must also have the resetting ones set to $\Low=0$, hence $V_l$ actually is the singleton $\{v_0\}$.

To finish, concerning non-transient vertices, vertices neither in $V_l$ nor $V_h$ are only connected to other vertices having the same coordinates in $X_{nr}$. Hence, each class of common coordinates contains at least one initial and one terminal SCC, which have to be non-transient singletons, for which all clocks in $X_r$ are set to $0$. Therefore each such class actually is a non-transient, isolated, singleton (so all transient vertices are in $V_h$).

Concerning transient vertices, let us consider any given region path $\pi$ realizing $e$. For any $i=0..|\pi|$, every vertex $v$ of $\gamma_d(\pi[i])$ such that, for at least one $c\in X_{nr}$, $v_c=\Hi_c$, has exactly one
successor: its successor by reset. Notice this is actually an $\instant$ edge. This successor, hence having the same values as $v$ for every clock of $X_{nr}$. So, by induction on $i$, we obtain that, in $e=\gamma_d(\pi)$, any vertex of $V_h$  has exactly one 
outgoing edge, going to another vertex with the same coordinates on $X_{nr}$ and this edge is $\instant$. This is, in particular, true for transient vertices, as we already know they belong to $V_h$. But as they need to have $v_1$ as a successor (transitivity), $v_1$ actually is their only 
successor and it is reachable by an $\instant$ edge.

Then, we have to prove that there is a $\slow$ edge from $v_0$ to every vertex in $V_h$ other than $v_1$. Consider $v$ such a vertex. Being transient, it should have an incoming edge from elsewhere (by \cref{cor:orbit-incoming-outgoing}). This can either be an $\instant$ edge from another vertex of $V_h$ (which should have its own incoming edge) or a $\slow$ edge coming from $v_0$ (as the only element in $V_l$). As elements in $V_h$ cannot have self-loops, and as the number of vertices is finite (if a vertex from $V_h$ is repeated, by idempotency this will mean a self-loop), eventually the second option will have to be taken, and by idempotency, this means a $\slow$ edge from $v_0$ to $v$. 

Lastly, concerning realizable p-orbits, 
just remark that going from $v_1$ to $v_0$ would create an obesity pattern of type II (with $v_0$ in the role of $u$ and $v_1$ in the role of $v$).
\end{proof}

Now we established a sufficiently precise characterization of structurally Zeno orbits and thus, thanks to \cref{lem:puri}, that of the stochastic matrices that may transform any clock vector, by executing any realization of a given structurally Zeno orbit. We use this knowledge to prove the next lemma: 
\begin{lemma}\label{lem:st-non-obese&st-zeno=>zeno}
In an \nsoA, any structurally Zeno d-orbit is Zeno.
 \end{lemma}
\begin{proof}
Let us fix a structurally Zeno idempotent $e$.

To prove the lemma, first, we describe a ranking function $\eta$ from the support region of $e$ to the real interval $[0,1]$ such that its value cannot decrease along a run; then we show that in any sequence of $k+1$ repetitions of $e$ (for some positive integer $k$), the accumulated duration of the first $k$ consecutive realizations of $e$ is smaller or equal to the increase of $\eta$ over the $k+1$ realizations. So, on a full run, the total duration spent in all realizations of $e$ (except the last of each repetitive sequence) is a fortiriori smaller than the total increase of $\eta$ over the run, which is smaller than $1$.

In the developments below, we abundantly use the definitions and characterizations from \cref{lem:zeno-shape}.

\proofsubparagraph{Construction of the ranking function $\eta$.}

We partition the vertices of the region into the sets 
$V_1\triangleq  \{ v \mid \exists \pi \text{ s.t. } (v_1, v) \text{ is an edge of } \gamma(\pi) \}$,
and $V_0\triangleq \Ver(R) \setminus V_1$. Since there is no return path from $v_1$ to $v_0$ (no structural obesity), this means $v_0\in V_0$. Note that this partitioning  
only depends on $R$, and not on $e$, knowing $R$.

We consider the Lyapunov function $\ell_{V_0}$, as described in \cref{con:one-lyap}.
As $V_0$ is an initial set of vertices of $\gamma(\pi)$ for any cycle $\pi$ starting from $R$ 
(its complement, $V_1$, is defined as a reachability cone),
for any clock vectors $x$ and $x'$ in $R$, such that $x'\in\Reach_\pi(x)$, it holds that $\ell_{V_0}(x)\geq \ell_{V_0}(x')$ (still according to the proof of the same corollary).

Hence, defining $\eta$ as the function $x\mapsto 1 - \ell_{V_0}$, we obtain a non-decreasing ranking function as required.

\proofsubparagraph{Proof of the link of $\eta$ with the accumulated duration of $e$.}

For $x_0$, $x_1$,  $x_2$, clock vectors and $\rho_1$ and $\rho_2$, realizations of the structurally Zeno d-orbit $e$ such that $x_1$ is reached from $x_0$ by reading $\rho_1$ and $x_2$ is reached from $x_1$ by reading $\rho_2$ , we will prove that $\tau(\rho_1)  \leq \eta(x_2)-\eta(x_0)$ (which is enough to deduce the Zeno property of $e$).

According to \cref{lem:puri}, there exist two stochastic matrices $P_1$ and $P_2$ such that $\lambda(x_1)=\lambda(x_0)P_1$, $\lambda(x_2)=\lambda(x_1)P_2$, where for any clock vector $x$ in the region $R$, $\lambda(x)$ is the vector of its barycentric coordinates with respect to vertices of $R$, sorted in clock order. Furthermore, according to the characterization from \cref{lem:zeno-shape}, $P_i$ must have the following form, for $\mu_j^i\geq 0$ and $\sum_j \mu_j^i = 1$:
\[P_i=\left(\begin{matrix}
\mu_0^i&0&0&\dots&0&0&\mu_1^i&\mu_2^i&\dots &\mu_{\#X_{r}}^i\\
0&1&0&\dots&0&0&0&0&\dots&0\\
0&0&1&\ddots&0&0&0&0&\dots&0\\
\vdots&\vdots&\ddots&\ddots&\ddots&\vdots&\vdots&\vdots&\ddots&\vdots\\
0&0&0&\ddots&1&0&0&0&\dots&0\\
0&0&0&\dots&0&1&0&0&\dots&0\\
0&0&0&\dots&0&0&1&0&\dots&0\\
0&0&0&\dots&0&0&1&0&\dots&0\\
\vdots&\vdots&\vdots&\vdots&\ddots&\vdots&\vdots&\vdots&\ddots&\vdots\\
0&0&0&\dots&0&0&1&0&\dots&0
\end{matrix}\right)
=
\left(\begin{matrix}
\mu_0^i&0&\mu_1^i&\mon_i\\
0&Id_{\#X_{nr} - 1}&0&0\\
0&0&1&0\\
0&0&1_{\#X_r-1}&0
\end{matrix}\right),
\]
where $Id_{\#X_{nr} - 1}$ is the identity matrix of size $\#X_{nr} - 1$, $\mon_i$ is the line vector $(\mu_2^i,\dots,\mu_{\#X_r}^i)$ and $1_{\#X_r-1}$ is the column vector $\left(\begin{matrix}1\\\vdots\\1\end{matrix}\right)$ of height $\#X_r-1$.

Since the matrix $P_i$ describes a barycentric coordinates transform where vertices are sorted in clock order, the first line and column concern $v_0$, the next $\#X_{nr}-1$ (5 here) lines and columns concern the isolated vertices with a self-loop, the $\#X_{nr}+1$-th line and column concerns $v_1$, and last, the last $\#X_r-1$ (3 here) lines and columns concern the transient vertices.

Remark that running $e$ in $t$ units of time must increase all non-reset clocks by $t$, in particular $c$ the one with the smallest fractional part. Hence going from $x_0$ to $x_1$, if we decompose $x_i$ in the form $\lambda_i v_0 + (1-\lambda_i) v$, with $v\in\Conv(\Ver(R)\setminus\{v_0\})$ and $\lambda_i$ stands for $\lambda(x_i)_{v_0}$, then projecting on coordinate $c$, we get $\tau(\rho_1)=x_c^1-x_c^0=(1-\lambda_1)-(1-\lambda_0)=\lambda_0-\lambda_1$.
Using the notations from the matrix $P_1$, we get $\tau(\rho_1)=\lambda_0 -\lambda_0\mu_0^1$. 

Now, observe that $\eta(x_2)-\eta(x_0)=\lambda(x_2)\one_{V_1} - \lambda(x_0)\one_{V_1}=\lambda(x_0)(P_1P_2-Id)\one_{V_1}$ (where $\one_{V_1}$ is the column vector having $1$ for all entries corresponding to the index of a vertex of $V_1$ and having $0$ everywhere else).
Moreover 
\[P_1P_2-Id = \left(\begin{matrix}
\mu_0^1\mu_0^2-1&0&\kappa &\mu_0^1\mon_2\\
0&0&0&0\\
0&0&0&0\\
0&0&\one_{\#X_{r}-1}&-Id_{\#X_{r}-1}
\end{matrix}\right)
\]
and thus 
\[(P_1P_2-Id)\one_{V_1} = \left(\begin{matrix}
1-\mu_0^1 + \mu_0^1\sum_{i\in V_1}\mu_i^2\\
0\\
0\\
1 \text{ for } V_0, 0 \text{ for } V_1
\end{matrix}\right),
\]
where $\kappa= \mu_0^1\mu_1^2+\sum_{i=1}^{\# X_{r}}\mu_i^1= \mu_0^1\mu_1^2+1-\mu_0^1$.
Hence \[\eta(x_2)-\eta(x_0)=\lambda(x_0)(P_1P_2-Id)\one_{V_1}=\lambda_0(1-\mu_0^1+\sum_{i\in V_1}\mu_0^1\mu_i^2)+\sum_{i\not\in V_1}\lambda_i(x) \geq \lambda_0(1-\mu_0^1)= \tau(\rho_1).\]
\end{proof}

Now we proceed to the \textbf{second stage} of the upper bound proof. Using Simon's theorem, we factorize words in the language into three types of subwords, where type $F$ factors correspond to Zeno subwords. 
\begin{lemma}\label{lem:factorize}
For any \nsoA\ $\aut$, there exist constants $a,d$, such that any timed word  $w\in L(A)$ admits a factorization $w=u_1u_2\cdots u_N$ and a labeling $\ell_1\ell_2\dots\ell_N\in \{B,I,F\}^*$  with $N\leq a\tau(w)+a$ satisfying
\begin{itemize}
     \item whenever $\ell_i=B$, the corresponding word $u_i$ has a form $(t,a)$ with $t\geq 0, a\in\Sigma$;
    \item whenever $\ell_i=I$, the word is instant: $\tau(u_i)=0$;
    \item time spent with label $F$ is bounded: $\sum_{i:\ \ell_i=F}\tau(u_i)\leq d$.
\end{itemize}
\end{lemma}
\begin{proof}
Consider the set $P$ of all paths   in $\aut$, we have that $P\subset \gamma^{-1}(\mon_d\backslash \{0\})$. To alleviate notation, given $e\in \mon_d$ we write $P_e$ for $\gamma^{-1}(e)$. 

Thus, by Simon's  theorem, $P\subset  P_K$ for some level $K$ of the following recurrence: 
\begin{equation}\label{eq:paths}
    P_0 =P_B= \Delta_{\aut}\cup\{\epsilon\};\,
    P_{h+1} =P_h\cdot P_h \cup \bigcup_{e\in E_I} (P_e\cap P_h)^*\cup 
    \bigcup_{e\in E_F}  (P_e\cap P_h)^*\cup
    \bigcup_{e\in E_S}  (P_e\cap P_h)^*.
\end{equation}
We will slightly modify this inductive construction: 
consider the recurrence
\[
    P'_0 =P_B\cup P_I \text{ (with $P_I=\bigcup_{e\in E_I}P_e$ )};\quad
    P'_{h+1} =P'_h\cdot P'_h \cup \bigcup_{e\in E_F} P_e\cdot (P_e\cap P'_h)\cup \bigcup_{e\in E_S} (P_e\cap P'_h)^*.
\]
By induction at each level $P_h\subset P'_h$, and thus $P\subset P'_K$.  In the union over $E_F$ we overapproximate all the iterations except the last one by $P_e$ and use the fact that $P_e^+=P_e$ for idempotent $e$.

The previous recurrence concerned paths in $\aut$, we need its version for time-bounded runs.  
For any $T$, let $R_T$ be the set of all runs of $\aut$ with duration $\leq T$. We also define run counterparts of all sets of paths defined above, so $R'_{hT}=\Runs( P'_h)\cap R_T$ etc.

First, we consider short runs (of duration <1). Applying $\Runs(\cdot)\cap R_{<1}$ to \eqref{eq:paths} we obtain
\[
    R'_{0,<1} =R_{B,<1}\cup R_I;\qquad
    R'_{h+1,<1} =((R'_{h,<1}\circ R'_{h ,<1})\cap R_{<1})\cup \bigcup_{e\in E_F} R_{e}\circ (R_{e}\cap R'_{h,<1})
\]
(we have omitted the last term of \eqref{eq:paths} since it does not produce short runs). The level $K$ of the induction generates all the short runs: $R_{<1}\subset R'_{K,<1}$.

For any short run $\rho\in R_{<1}$, by induction on $h$ we produce a factorization and labeling required in the  statement of the lemma as follows :
\begin{itemize}
    \item for the base case  $\rho\in R'_{0,<1}$ we take only one factor, and label it $I$ if $\rho\in R_I$ or $B$ otherwise;
    \item for the inductive case $\rho\in R'_{h+1,<1}$, whenever $\rho =\rho_1\circ\rho_2$ (with $\rho_i\in R'_{h,<1}$) we factorize $\rho_i$ by inductive hypothesis and merge the factorizations and labelings;
    \item finally, for the inductive case of $\rho=\rho_1\circ \rho_2$ with $\rho_1\in R_{e}$ and $\rho_2\in R_{e}\cap R'_{h,<1}$ with $e\in E_F$, we label the whole $\rho_1$ by $F$. We factorize and label $\rho_2$ using the inductive hypothesis.  
\end{itemize}
We obtain at level $h$ at most $2^h$ factors, hence at most  $2^K$ factors for any short run $\rho\in R_{<1}$.  

The labeling of the short runs has the following Zeno-like property: in any run\linebreak $\zeta_0\rho_1\zeta_1\rho_2\zeta_2\dots\rho_n\zeta_n$ with all $R'_{h,<1}$ short, let us label factors of $\rho_i$ as described above. Then the accumulated duration of $F$-labeled ones  is bounded by $h\#E_F$. This property can be proved by induction on $h$, and implies the global bound of $d=K\cdot \#E_F$.

For each $e$, on level $h+1$, the accumulated duration of new $F$-factors is $\leq 1$ due to Zeno property  (\cref{lem:st-non-obese&fast=>st-zeno,lem:st-non-obese&st-zeno=>zeno}); summing over all possible $e$ this gives $\#E_F$. The accumulated duration  of old factors (inherited from level $h$) is at most $h\#E_F$. This implies the announced bound.

Consider now the set of runs $R_T$ for an arbitrary $T$. Applying $\Runs(\cdot)\cap R_T$ to the recurrence \eqref{eq:paths} we obtain a recurrence for runs:
\[
    R'_{0T} =R_{BT}\cup R_{I};\qquad
    R'_{h+1,T} =\left(R'_{hT}\circ R'_{hT} \cup \bigcup_{e\in E_F} (R_e\cap R'_{hT})^*\cup \bigcup_{e\in E_S} (R_e\cap R'_{hT})^*\right)\cap R_T,
\]
 and $R_T\subset R'_{KT}$. 
We overapproximate it as follows
\[
    R''_{0T} =R_{BT}\cup R_{<1};\qquad
    R''_{h+1,T} =R_{hT}^{''\leq T+2}\cap R_T.
\]
We omitted the union over $E_F$ because it is included in $R_{<1}$ already injected at level 0. The star in the union over $E_S$ has been truncated up to power $T$ since each run  $R_e$ with $e\in E_S$ has a duration of at least $1$.

Again we have $R_T\subset R''_{KT}$. Now we can produce the factorization and labeling of each run in $R_T$:
\begin{itemize}
    \item for the base case  $\rho\in R''_{0,<1}$  such that $\rho\in R_{BT}$ we take one factor, and label it $B$;
    \item otherwise, in the base case such that $\rho\in R_{<1}$ we apply the factorization for short runs described above;
    \item for the inductive case 
     $\rho\in (R''_{hT})^{\leq T}\cap R_T$, in this case $\rho=\rho_1\rho_2\dots \rho_n$ with $n\leq T+2$ with $\rho_i\in R''_{hT}$ satisfying the inductive hypothesis. We factorize and label  each of $\rho_i$ and obtain a factorization of $\rho$.  
\end{itemize}

We will now prove  by induction over $h$ that the size of factorization obtained $N\leq 2^{h+K}T+2^{h+K}$.

Indeed, for the base level, the statement is evident. 
For level $h+1$, whenever  $\rho\in R''_{h+1,T}$, it can be factorized as $\rho=\rho_1\rho_2\dots \rho_n$ with $n\leq T$, and $\rho_i$ of some durations $T_i$ with $\sum_{i=1}^n T_i \leq T$ and satisfying the inductive hypothesis. Combining factorizations of $\rho_i$ we get a factorization of $\rho$ of size 
\[
\sum_{i=1}^n(2^{h+K}T_i+2^{h+K})\leq 2^{h+K}T+2^{h+K}n\leq 2^{h+K}T+2^{h+K}(T+2)=2^{h+1+K}T+2^{h+1+K},
\]
 which concludes the induction.
Hence each run  $\rho\in R_{T}\subset R''_{KT}$ has a factorization of size $aT+a$ with $a=2^{2K}$. Factors labeled by $I$ have duration $0$, and factors labeled by $B$ correspond to one transition. As observed above, the sum  of durations of factors labeled by $F$ does not exceed $d$.
Each word $w\in L_T(\aut)$ corresponds to some run $\rho\in R_T$, projecting its factorization yields a required factorization of $w$.
\end{proof}
We can now construct an $\varepsilon$-net for $L_T$.

\begin{con}\label{def:net:nonobese} 
The set $\mathcal{N}(\Sigma,T,\varepsilon,a,d)$ of timed words with duration $\leq T$ is constructed as follows:
\begin{itemize}
    \item at some $aT+a$ positions multiple of $\varepsilon$ we put a set of letters in $\Sigma$;
    \item at some $d/\varepsilon$ positions multiple of $\varepsilon$ we put an element of $\mathcal{N}_{U\varepsilon}$,
\end{itemize}
where $\mathcal{N}_{U\varepsilon}$ is a finite $\varepsilon$-net for the set of all  timed words of duration $\leq 1$ as can be found in \cite[Thm 2]{distance}, its size is $2^{\#\Sigma/\varepsilon}$.
\end{con}

\begin{lemma}\label{lem:net:nonobese}
The set $\mathcal{N}(\cdot)$ is an $\varepsilon$-net for $L_T(\aut)$.
\end{lemma}
\begin{proof}
    By virtue of \cref{lem:factorize}, every word $w\in L_T$ can  be factorized in $\leq aT+a$ factors with labels $B,I,F$.
The total duration of factors with label  $F$ is bounded by $d$, thus only $d/\varepsilon$ among them can have a ``non-negligible'' duration $\geq \varepsilon$.

The former bullet of  \cref{def:net:nonobese} allows approximating  letters labeled by  $B$, instant words labeled by $I$, and factors  with label $F$ shorter than $\varepsilon$; the latter ---  factors with label $F$ longer than $\varepsilon$. 
\end{proof}

\begin{lemma}\label{lem:net:size}
The cardinality $S(T,\varepsilon)$ of the set $\mathcal{N}(\cdot)$ satisfies the estimate
\[
\lim_{T\to\infty}\log S(T,\varepsilon)/T= O(\log(1/\varepsilon)).
\]
\end{lemma}
\begin{proof}
The following bound is immediate from  \cref{def:net:nonobese}:
\[
S(T,\varepsilon)\leq\binom{T/\varepsilon}{aT+a}2^{\#\Sigma(aT+a)}\cdot \binom{T/\varepsilon}{d/\varepsilon}2^{(\#\Sigma/\varepsilon)(d/\varepsilon)}.
\]
This gives an upper bound for the bandwidth:
\begin{multline}\label{eq:size}
\lim_{T\to\infty}\frac{\log S(T,\varepsilon)}{
T}\leq
\lim_{T\to \infty}
\frac{\log \binom{T/\varepsilon}{aT+a}+\#\Sigma(aT+a)+ \log\binom{T/\varepsilon}{d/\varepsilon}+\log\left( \frac{\#\Sigma}{\varepsilon}\cdot\frac{d}{\varepsilon}\right)}{T}=
\\
  =\lim_{T\to \infty}\frac{\log \binom{T/\varepsilon}{aT+a}}{T}+\lim_{T\to \infty}\frac{\log \binom{T/\varepsilon}{d/\varepsilon}}{T}+\#\Sigma\cdot a.
\end{multline}
We notice that whenever $A\gg B\gg 1$, then due to Stirling's formula
$
\binom{A}{B}\leq (Ae/B)^B$.
Thus the first term of the sum in  \eqref{eq:size} does not exceed
\[
 \lim_{T\to \infty}\frac{(aT+a)\log\left(\frac{Te}{\varepsilon}\cdot\frac{1}{aT+a}\right)}{T}=a \log\frac{1}{\varepsilon}+a\log\frac{e}{a},   
\]
and the latter term gives 
\[
\lim_{T\to \infty}\frac{\frac{d}{\varepsilon}\cdot\log\left(\frac{Te}{\varepsilon}\cdot\frac{\varepsilon}{d}\right)}{T}=0.
\]
Gathering all the estimates we obtain the required bound:
\[
\lim_{T\to\infty}\frac{\log S(T,\varepsilon)}{
T}\leq a\log\frac{1}{\varepsilon}+a\log\frac{e}{a}+\#\Sigma\cdot a=O\left(\log\frac{1}{\varepsilon}\right). \qedhere
\]
\end{proof}
From \cref{lem:net:nonobese,lem:net:size} follows the upper bound of \cref{prop:obese:upper}.
\subsection{Lower bound}
\begin{proposition}\label{prop:obese-cycle}%
If $\pi$ is a structurally obese cycle whose obesity pattern involves vertices $u$ and $v$, then  the language $L_{\bar{\pi}^*}(u,\cdot)$ is obese.
\end{proposition}
\begin{proof}
We exhibit a construction for the two subcases of structural obesity. The second construction is actually the same as the first, with some runs through a foreign cycle inserted regularly along the runs of the $\varepsilon$-separated set.

\proofsubparagraph{Case of Type~I structural obesity.}

For a cycle $\pi$ exhibiting Type~I structural obesity, vertex $u$  can be the starting point of some realizations of duration 0, going back to $u$, and of some realizations of duration $\geq 1$, also going back to $u$. Hence by convex combination, it is possible to go from $u$ to $u$ with runs of any duration $\lambda$ with $0\leq\lambda\leq 1$.

In particular, for any $\varepsilon$, $0<\varepsilon<\frac{1}{3|\pi|}$, there exists at least a run $\sigma(x)$ of $\bar\pi$ of length $\varepsilon$ and a run $\rho_0(x)$ of length $3\varepsilon|\pi|$, both going back to $u$. 
Let us define $\rho_1(x)=\sigma(x)^{3|\pi|}$ (so that $\rho_0(x)$ and $\rho_1(x)$ have the same duration) and construct a function $f_x: \{0,1\}^*\to \Runs_{\bar{\pi}^*}$, inductively: $f_x(\epsilon)=\epsilon$; for $b\in\{0,1\}$, $f_x(b)=\rho_b(x)$ and, for $b\in\{0,1\}$, $f_x(wb)=f_x(w)\rho_b(\dst(f_x(w))$. We show the language of the labels of runs in the image is $\varepsilon$-separated. Indeed, for 2 different words of this language, there are 2 different runs that are labeled by these words, and these runs are of the form $f_x(w)$ and $f_x(w')$ with $w\neq w'$; i.e.~there is $i\in\nat$ such that $w[i]=0$ and $w'[i]=1$. This means that the slices of $f_x(w)$ and $f_x(w')$ for the time interval $\left[3i\varepsilon|\pi|, 3(i+1)\varepsilon|\pi|\right]$ are $\rho_0(\dst(f_x(w[0..i-1]))$ and $\rho_1(\dst(f_x(w[0..i-1]))$.
But $\rho_0(\dots)$ necessarily has a gap of length $3\varepsilon$ without any event, while $\rho_1(\dots)$ has at least one event in every interval of length $\varepsilon$, so at least one event of $\rho_1(\dots)$ cannot be matched in $\rho_0(\dots)$ by an event with a date closer than $\varepsilon$. Therefore this is also true of $f_x(w')$ and $f_x(w)$, thus $d(\Word(f_x(w)),\Word(f_x(w'))>\varepsilon$. Hence $\{\Word(f_u(w))\mid w\in\{0,1\}^*\}_{\leq T}$ is an $\varepsilon$-separated subset of $L(\bar{\pi}^*)_{\leq T}$ and its size is at least $2^{\left\lfloor \frac{T}{3\varepsilon|\pi|}\right\rfloor}$.

\proofsubparagraph{Case of Type~II structural obesity.}

For a Type~II cycle, the construction is similar. For technical reasons, for a given $\varepsilon$, we will build an $\varepsilon'$-separated set with $\varepsilon'$ the next value (slightly) larger than $\varepsilon$ such that $\frac{1}{3\varepsilon'|\pi|}$ is integer. The difference with previous construction is that as long as only $\pi$ is executed, instead of going back to $u$, we go farther and farther from $u$ and closer and closer to $v$: a run of duration $t$ from $(1-\lambda)u+\lambda v$ goes to $(1-\lambda-t)u+(\lambda+t)v$, provided $\lambda+t\leq 1$. So the analogs of runs $\rho_0$ and $\rho_1$ can only be executed $\frac{1}{3\varepsilon'|\pi|}$ times at most before it is necessary to ``recharge'' the clocks through the execution of a resetting cycle.

More precisely, the vertex $u$ can be the starting point of a run of duration 0 through $\pi$, going back to $u$, as well of runs of duration $\geq 1$, going to $v$; while, through $\pi$, vertex $v$ can only be the starting point of realizations of duration $0$, all going back to $v$.
By convex combinations, from any $x(\lambda)=(1-\lambda) u + \lambda v$ it is thus possible to start a run of any duration $\mu\leq 1-\lambda$ and reach the point $x(\lambda+\mu)$.

In particular, with $\varepsilon'$ such that $0<\varepsilon'<\frac{1}{3|\pi|}$, for any point $x(\lambda)$ of the segment $[u,\tilde{v}]$ (where $\tilde{v}=x(1- 3\varepsilon'|\pi|)$), there exists at least a run $\sigma(x(\lambda))$ of $\bar\pi$ of duration $\varepsilon'$ and a run $\rho_0(x(\lambda))$ of duration $3\varepsilon'|\pi|$. Let us define $\rho_1(x)=\sigma(x)^{3|\pi|}$ (so that $\rho_0(x)$ and $\rho_1(x)$ have the same duration) and remark that both $\rho_0(x(\lambda))$ and $\rho_1(x(\lambda))$ reach $x(\lambda+3\varepsilon'|\pi|)$. Before we construct the new function $f_x: \{0,1\}^*\to \Runs_{\bar{\pi}^*}$, let us choose a resetting cycle $\pi_r$ such that its p-orbit has an edge from $v$ to $u$, a run $\rho_r$ that goes from $v$ to $u$ through $\pi_r$ and $\tau_r$ the duration of this run.
%

Now we construct $f_x$: $f_x(\epsilon)=\epsilon$; with $b\in\{0,1\}$, if $x\in [u,\tilde v]$,
$f_x(b)=\rho_b(x)$ and $f_x(bw)=\rho_b(x)f_{\dst(\rho_b(x)}(w)$; 
otherwise $f_v(w)=\rho_r(v)f_u(w)$ (note that values of $x$ in $(\tilde v, v)$ cannot be reached with this strategy; hence $f$ is left undefined for this interval).

Here again, the language of the labels of the image is $\varepsilon'$-separated for the same reason as in the previous case. 
So $S_{\varepsilon}(T)\triangleq \{\Word(f_u(w))\mid w\in\{0,1\}^*\}_{\leq T}$ is an $\varepsilon'$-separated (thus $\varepsilon$-separated) subset of $L((\bar{\pi}^{\frac{1}{3\varepsilon' |\bar{\pi}|}}\bar{\pi_r})^*)_{\leq T}$ equal to the product of its slices corresponding to intervals $[ip,i(p+1)]$, $i\in\nat$, with $p=\left\lfloor\frac{1}{3\varepsilon'|\pi|}\right\rfloor3\varepsilon'|\pi|+\tau_r=1+\tau_r$.
Each slice contains exactly $2^{\frac{1}{3\varepsilon'|\pi|}}$ elements, hence $\log\#S_\varepsilon(T) \geq {\frac{1}{3\varepsilon'|\pi|} \left\lfloor\frac{T}{p}\right\rfloor} \geq {\frac{1}{3\varepsilon'|\pi|} \left\lfloor\frac{T}{1+\tau_r}\right\rfloor} \geq {(\frac{1}{3\varepsilon|\pi|} - 1) \left\lfloor\frac{T}{1+\tau_r}\right\rfloor}$. Hence $\mathcal{BH}_\varepsilon(L(\bar{\pi}^*)_{\leq T})=\Theta(\frac{1}{\varepsilon})$ (the upper bound is given by comparison with the universal language, having a bandwidth  $\Theta(\frac{1}{\varepsilon})$).

\end{proof}

\thmcomplexityobese*
\begin{proof}[Proof sketch]
The  \RTA\ contains at most $n$ locations. Each orbit can be represented as a couple of locations and an $(m+1)\times(m+1)$-matrix with elements in $\Lambda_d$ (a semiring of 4 elements).
Thus
$
\# \mon_d \leq n^2\cdot 4^{(m+1)^2}=\poly(n)\myexp(m).
$
An automaton is structurally obese whenever  it contains a path $\pi$ with orbit $\gamma_d(\pi)$ having an obesity pattern. According to \cref{lem:simon-minimal-length}, only cycles up to length $\#\mon_d$ need to be checked. 
 We use now the reachability method   \cite[Sect.~7.3]{papa} to detect the obesity pattern in a  space-efficient way, see \cref{algo:obese}. 
An auxiliary function \lstinline{isAPathOrbit(e,h)}, true iff there is a path $\pi\in \gamma_d^{-1}(e)$ of length $\leq 2^h$ in  the \RTA, 
is computed by a divide-and-conquer algorithm. 
\begin{lstlisting}[mathescape=true,language=Java,float,caption=Deciding structural obesity,label=algo:obese]
Boolean function isAPathOrbit(e, h)
    if h = 0 return (e=$\1$ || $\exists \delta  : \gamma_d(\delta)=$e)
    for all orbits e1, e2 
        if e1$\cdot$e2 = e && isAPathOrbit(e1, h-1) && isAPathOrbit(e2, h-1)
            return true
    return false
 Boolean function isStrucObese()
    for all orbits e
        if (obeseType1(e) && isAPathOrbit(e, $\log \#\mon_d$)) return true
    for all states q, int i,j $\leq d_q$
        for all orbits e1,e2 from q to q
            if (obeseType2(e1,i,j) && resettingPath(e2,j,i)&& 
               isAPathOrbit(e1, $\log \#\mon_d$) && isAPathOrbit(e2, $\log \#\mon_d$))
                  return true
    return false
\end{lstlisting}
The algorithm requires a call stack of depth $\log\#\mon_d$, the size of each stack frame is $O(\log n)+\poly(m)$, thus the 
 computation takes  space   $\left(O(\log n)+\poly(m)\right)\cdot \log\#\mon_d=\log^2(n)\poly(m)$.
\end{proof}

\corstrobeseisobese*
\begin{proof}
    This is a direct corollary of \cref{prop:obese-cycle}. We obtain an $\varepsilon$-separated set $S'$ from $S$, the $\varepsilon$-separated set of $L_{\bar{\pi}^*,T}(u,\cdot)$ constructed in the previous proof, by prepending a common prefix $w$ to all its elements, such that $w$ is the labeling of a run $\rho$ going from an initial vertex of an initial location to the vertex $u$, through some path $p$.

    This shows that $L_{\bar{p}\bar{\pi}^*}$ is obese, but since for each $\varepsilon$-separated subset of $L_{\bar{p}\bar{\pi}^*,T}$, $L_{{p}{\pi}^*,T}$ admits a $\varepsilon'$-separated subset of the same size (according to \cref{lem:closed-semantics}), this actually also proves that $L_{{p}{\pi}^*}$ is  obese and hence its superset $L(\aut)$.
\end{proof}

\section{Proofs on meagerness}
\subsection{Monoid properties}
\propForbit*
\begin{proof}
Here again, we only need to prove, for any two edge sequences $\pi_1,\pi_2$, that \linebreak $\gamma_f(\pi_1\circ\pi_2)=\gamma_f(\pi_1)\gamma_f(\pi_2)$.
It is still trivially true if any of the two sequences is $\epsilon$, or if $\pi_1\circ\pi_2$ is not a path.
So assume again both $\pi_1$ and $\pi_2$ are paths of length $\geq 1$ such that $\pi_1\circ\pi_2$ also is a path and let us call $R$ the destination region of $\pi_1$.

Let us choose any source vertex $u$ and any destination vertex $v$ for $\pi_1\circ\pi_2$. Then $\left(\gamma_f(\pi_1)\gamma_f(\pi_2)\right)[u,v]$ has to take one of the three values of $\Lambda_d$:

\begin{itemize}
\item If it is $\none$: the proof is the same as the case $\none$ in \cref{prop:dorbit-morphism} (actually $\gamma_f(\pi)[u,v]=\none$ iff $\gamma_d(\pi)[u,v]=\none$ iff $\gamma(\pi)=0$).
\item If it is $\narrow$, there is a single vertex $w$ of $R$ such that $\gamma_f(\pi_1)[u,w]=\gamma_f(\pi_2)[w,v]\neq\none$, and this value has to be $\narrow$. But according to \cref{lem:region-reach}, a full region must be reachable from $u$ through $\pi_1$, and $v$ must be reachable from a full region through $\pi_2$; the intersection of these two regions must also be a region, but this region contains only the vertex $w$, hence it is equal to $\{w\}$. So the only run from $u$ to $v$ through $\pi_1\circ\pi_2$ is the concatenation of the only run from $u$ to $w$ and the only run from $w$ to $v$. Hence $\gamma_f(\pi_1\circ\pi_2)[u,v]=\narrow=\left(\gamma_f(\pi_1)\gamma_f(\pi_2)\right)[u,v]$.
\item If it is $\wide$, because of semiring rules \cref{fig:freedom-monoid-plus}, there are 2 possibilites. Either $\wide$ is obtained by virtue of f-orbit multiplication, meaning there is $w$ such that $\gamma_f(\pi_1)[u,w]=\wide$
and $\gamma_f(\pi_2)[w,v]\neq \none$; since, by definition, $L_{\overline{\pi_1}}(u,w)$ contains several words, and thus so is the case of $L_{\overline{\pi_1\circ\pi_2}}(u,v)$.
Or we obtained $\wide$ by virtue of f-orbit addition, meaning that there are $w_1$ and $w_2$, $w_1\neq w_2$, such that $\gamma_f(\pi_1)[u,w_1]=\narrow$,
$\gamma_f(\pi_2)[w_1,v]=\narrow$,
$\gamma_f(\pi_1)[u,w_2]=\narrow$ and
$\gamma_f(\pi_2)[w_2,v]=\narrow$;
this implies you can go from $u$ to $v$ by at least two different orbit paths, with different timings, hence $L_{\bar{\pi}(u,v)}$ contains at least two distinct words. Hence $\gamma_f(\pi_1\circ\pi_2)[u,v]=\wide=\left(\gamma_f(\pi_1)\gamma_f(\pi_2)\right)[u,v]$.\qedhere
\end{itemize}
\end{proof}

\subsection{Upper bound}

Throughout this section, we reason about a structurally meager region-split automaton, we will abbreviate it as \stA.
As we said, the case when the language $L_{\bar\pi}(x,x)$ is a singleton (typical for structurally meager automata) will play the key role. First, we prove that whenever $\pi$ is an ``idempotent cycle'', the language mentioned above is indeed a singleton.
\begin{lemma}\label{lem:meager-singleton}
In an \stA, for any  path $\pi$  with idempotent orbit and a clock vector $x$ in the closure of starting region of $\pi$,  the language $L_{\bar{\pi}}(x,x)$ is either empty or a singleton. 
\end{lemma}
\begin{proof}
First note that it is already known \cite{puri}, for p-orbits, that $L_{\bar{\pi}}(x,x)$ is non-empty iff $x$ is in the convex hull of some set of vertices with self-loops. So, it remains to prove that, in the non-empty case, the language is a singleton, which we do by using the fact that, in the f-orbit, all these self-loops are $\narrow$ when $\pi$ is structurally meager.

Before going on, we show it is possible to assume w.l.o.g. that, for all factorization $\pi=\pi_1\circ\pi_2$, the ``shifted'' cycle $\pi'=\pi_2\circ\pi_1$ always has an idempotent f-orbit.  If not, reason with $\pi^p$, where $p$ is sufficiently large, observe it is still structurally meager and then, prove its polytope is a singleton. It implies that $L_{\bar{\pi}}(x,x)$ also is a singleton: indeed, if it were some non-singleton polytope, then $L_{\bar{\pi}^p}(x,x)$ would  also include several words.  

So we assume that $\pi$ and all its ``shifted'' versions have idempotent p-orbits.

{
Suppose, towards a contradiction, that $L_{\bar{\pi}}(x,x)$ is neither empty nor a singleton. Then, this language has at least two words with different timings and, thus, with a first disagreeing timing, so there exist two paths $\pi_1$ and $\pi_2$ and an edge $\delta\in\Delta$ such that $\pi=\pi_1\circ \delta \circ\pi_2$ a word $w_1$ through $\pi_1$, two delays $t$ and $t'$ with $0<t<t'$ and two words $w_2$ and $w_2'$ through $\pi_2$ such that both $w_1\circ (\lbl_\delta,t) \circ w_2$ and $w_1\circ (\lbl_\delta,t') \circ w'_2$ belong to $L_{\bar\pi}(x,x)$. Let us call $y$ and $y'$ the two clock vectors reached after reading respectively $w_1\circ(\lbl_\delta,t)$ and $w_1\circ(\lbl_\delta,t)$ through $\pi_1\circ\delta$ from $x$.

If $y\neq y'$, then observe both vectors are reachable from each other through $\pi'=\pi_2\circ\pi_1\circ\delta$. The smallest affine set that contains the two vectors has dimension 1, hence, according to \cref{cor:mut-reach-dim}, the starting region of $\pi'$  has at least one non-singleton SCC $S$.
Since $\gamma_f(\pi')$ is idempotent, all self-loops $\gamma_f(\pi')[u,u]$ for $u\in S$ are $\wide$, which is impossible in a structurally meager automaton. 

If $y=y'$, this means the edge $\delta$ can reach the same vector with different delays from the same originating vector, and hence that $\guard_\delta$ is not punctual and $\reset_\delta$ contains at least all clocks such that their value is not constrained inside a singleton in $\guard_\delta$. In this case $S(\dst_\delta)$ also is a singleton, of the form $\{v\}$ with $v=(0,\dots,0)$. So $v$ is the only vertex of $\gamma_d(\pi_2\circ\pi_1\circ\delta)$. It necessarily has a self-loop, and since $L_{\overline{\pi_2\circ\pi_1\circ\delta}}(v,v)$ contains several words, this loop is $\wide$, which is the non-meagerness pattern.
}
\end{proof}

We establish now  a continuity result for the vicinity of a single-word path:
\begin{lemma}\label{lem:narrow-variations}
Given a  path $\pi$ in an \RTA\ and clock vectors $x^0,y^0$, such that $L_{\bar\pi}(x^0,y^0)=\{w^0\}$ (a singleton), for any clock vectors $x^1,y^1$ such that both $||x^0-x^1||_\infty<\varepsilon$, $||y^0-y^1||_\infty<\varepsilon$,  
 and for any $w^1\in L_{\bar\pi}(x^1,y^1)$ we have $d(w^0,w^1)<3\varepsilon$.
\end{lemma}
\begin{proof}
We fix notations for start and end regions of the path: $S(\src_\pi)=R_s$, $S(\dst_\pi)=R_d$,  clearly $x^0,x^1\in \bar{R_s}$ and $y^0,y^1\in \bar{R_d}$ (otherwise the languages $L_{\bar\pi}(\cdot)$ would be empty.) 

For any $x\in \bar {R_s}, y\in \bar {R_d}$ all words in $L_{\bar\pi}(x,y)$ have the form $(a_1,t_1),\dots (a_n, t_n)$, where the untiming is always the same $a_1\dots a_n=\lbl_w$, and the set of timings is defined by a parametric \DBM{} $A(x,y)$, as described in \cref{prop:timingDBM}, and its projection to $t_j$ is an interval $[L_j(x,y);U_j(x,y)]$ as described in \cref{prop:proj-DBM}. As we know $L_{\bar\pi}(x^0,y^0)$ is a singleton $w^0$, thus for every $j$, the interval for $t^0_j$ is also a singleton: $L_j(x^0,y^0)=U_j(x^0,y^0)=t^0_j$. \DBM{}s $A(x^1,y^1)$ and $A(x^0,y^0)$ satisfy the hypotheses of \cref{cor:Lipshitz}: they are equal in the middle and $\varepsilon$-close on border, thus by this Corollary,
\[
t^0_j-3\varepsilon< L_j(x^1,y^1)\leq t^1_j \leq U_j(x^1,y^1)< t^0_j+3\varepsilon,
\]
hence $|t^0_j-t^1_j|< 3\varepsilon$ (for any $j$). We conclude that $d(w^0,w^1)< 3\varepsilon$.
\end{proof}

In order to characterize the duration of  singletons we will  naturally use the monoid $\mon_d$ again and define a new structural property. An idempotent $e$ in $\mon_d$ will be called  \emph{quick} whenever all its self-loops are $\instant$.

\begin{lemma}\label{lem:singletons}In an \stA, if for a cycle $\pi$ with idempotent $e=\gamma_d(\pi)$ and a clock vector  $x$ the language $L_{\bar\pi}(x,x)$ is a singleton $\{w\}$, then 
\begin{itemize}
    \item   if $e$ is quick then  $\tau(w)=0$;
    \item otherwise  $\tau(w)\geq 1$.
\end{itemize}
   \end{lemma}   
\begin{proof}
Let $\pi$ be a cycle such that $\gamma_f(\pi)$ is idempotent and $x$ a clock vector such that $L_{\bar\pi}(x,x)$ is not empty.
Necessarily, $x$ is in a facet of the starting region of $\pi$, delimited by vertices all having self-loops in the $(p,f,d)$-orbits (the convex sums of triples $(x,t,x)$ that satisfy the linear system of constraints of $\pi$ also satisfy it; hence if $x$ is a convex combination of vertices $v$ such that at least one is such that $(v,t,v)$ satisfies the system for several $t$, then $x$ also accepts several $t$ such that $(x,t,x)$ satisfies the system). Let us call $V$ the set of vertices of this facet.

If $\gamma_d(\pi)$ is quick (only has $\instant$ self-loops), then any element $w$ of $L_{\bar\pi}(x,x)$ is such that $\tau(w)=0$.
If $\gamma_d(\pi)$ only has $\slow$ self-loops, then $\tau(w)\geq 1$.
If $\gamma_d(\pi)$ has a $\fast$ self-loop, then $P(x,x)$ is not a singleton.
If $\exists u,v\in V$ s.t. $\gamma_d(\pi)[u,u]=\instant$ and $\gamma_d(\pi)[v,v]=\slow$, then \cref{lem:instant+slow=forgetful} applies, and we deduce that $\gamma_p(\pi)$ is complete and that therefore all the self-loops of $\gamma_f(\pi)$ are $\wide$ and that therefore $L_{\bar\pi}(x, x)$ is not a singleton.
\end{proof}
We will now need a small technical lemma.
\begin{lemma}\label{lem:transient-to-transient}
    Let $e$ be a realizable idempotent of $\mon_f$. If $u$ and $v$ are transient vertices such that $e[u,v]=\narrow$, then there is a non-transient vertex $w$ such that $e[u,w]=\narrow$ and $e[w,v]=\narrow$.
\end{lemma}
\begin{proof}
    Assume there is no such vertex $w$. Then all paths from $u$ to $v$ only pass through transient vertices. Paths through transient vertices cannot go back or stay in place, thus by taking $n$ the length of the longest such path, $e^{n+1}$ does not contain a path from $u$ to $v$, which contradicts $e$ being idempotent.
\end{proof}
The following lemma states that in  any singleton run along a structurally meager idempotent cycle,  all iterations except the first and the last one  loop on  a same clock configuration.
\begin{lemma}\label{lem:xzzy} In an \stA,  let $e\neq 0$ be an idempotent, let $\pi=\pi_1\dots \pi_k$ such that all $\gamma(\pi_i)=e$ and $L_{\bar\pi}(x,y)=\{w\}$. Then there exists a factorization $w=w_1\dots w_k$  and a clock vector $z$ such that:
$L_{\bar\pi_1}(x,z)=\{w_1\}$,  for $i=2..k-1$ it holds that $L_{\bar\pi_i}(z,z)=\{w_i\}$, and finally  $L_{\bar\pi_k}(z,y)=\{w_k\}$.
\end{lemma}

\begin{proof}
    Assume such a path $\pi$ and call $\rho = \rho_1\dots\rho_k$ its unique run such that $\Runs_{\bar{\pi}}(x,y)=\{\rho\}$ and $\rho_i\in \Runs_{\bar{\pi_i}}$.

    Necessarily there is $x_0=x, x_1, x_2, \dots x_k=y$ such that $\Runs_{\bar{\pi}}(x_i,x_{i+1})=\{\rho_{i+1}\}$.

    Moreover, the timing vector of $\rho_i$ necessarily is a convex combination of timings of runs from a vertex $u$ to another $v$ in the support of $e$. Note that if there were several such runs from $u$ to $v$, there would be several runs from $x_i$ to $x_{i+1}$ through $\pi_{i+1}$, hence several from $x$ to $y$ through $\pi$, which would contradict the hypothesis. Hence the convex combination has positive coefficients only for runs from $u$ to $v$ such that $e[u,v]=\narrow$.

    Additionally, we show that, for $i<k$, $\rho_i$ cannot be a convex combination including runs from vertex $u$ to vertex $v$ when $v$ is transient.

    Indeed, $\narrow$ edges can only go from $u$ to $v$ such that either $u=v$ or at least one of the two vertices is transient (otherwise it would violate the idempotency of $e$).

    Assume $\rho_i$ can be written as a linear combination of vertex-to-vertex runs such that one goes from $u$ to $v$ with $v$ transient. Hence $x_i$ has a positive barycentric coordinate in $v$. So $\rho_{i+1}$ is a convex combination of vertex-to-vertex runs which must include a run from $v$ to another vertex $w$. Remark that $\rho_i\rho_{i+1}$ is also a realization of $\pi$ and so it can be written as a linear combination of vertex-to-vertex runs too, this one such that one run goes from $u$ to $w$ (connecting the run from $u$ to $v$ and the one from $v$ to $w$). By  \cref{lem:transient-to-transient}, there are non-transient vertices in the path from $u$ to $v$ and from $v$ to $w$. But in this case, $e[u,w]=\wide$, so it means other runs than $\rho_i\rho_{i+1}$ can go from $x_{i-1}$ to $x_{i+1}$ which contradicts the hypothesis.

    This means that the barycentric coordinate of $x_i$, $0<i<k$, for any transient vertex is $0$. So for $1<i<k$, $\rho_{i}$ is a convex combination of vertex-to-vertex runs starting from each non-transient vertex $u$ and going to themselves (going to another non-transient is impossible because the edge is $\wide$, going to a transient is forbidden, as we just proved), and the coefficient of these runs is necessarily equal to the weight of $u$ in the barycentric coordinates of $x_{i-1}$, hence $x_{i}=x_{i-1}$.
\end{proof}

Given a timed word $w = (a_1,t_1)\dots (a_n,t_n)$ we denote $\vvvert w\vvvert=\#\{i\mid t_{i-1}<t_i\}$, i.e.~the number of non-instant events in $w$.
We will now estimate the number of non-instant delays in singletons w.r.t.~their duration. 
\begin{lemma}\label{lem:freq}
    For any path $\pi$ and  clock vectors $x,y$ in an \stA, if the language $L_{\bar\pi}(x,y)=\{w\}$  (a singleton), then   $\vvvert w\vvvert\leq  a\tau(w)+a$, with the constant $a$ depending on the automaton.
\end{lemma}
 \begin{proof}
 We will use Simon's theorem for monoid $\mon_d$ and prove by induction, that whenever $\pi\in X_h$ (with $X_h$ as in \eqref{eq:simon}) satisfies the Lemma's hypothesis for some $x,y,w$, then $\vvvert w\vvvert\leq 2^h\tau(w)+2^h$. Since $h\leq 3\#\mon_d$, this would imply the required result with $a=2^{3\#\mon_d}$.

\begin{itemize}
    \item The base case $\pi\in X_0$ is evident, $\vvvert w\vvvert\leq 1$.
    \item The first inductive case is $\pi=\pi_1\pi_2$ with $\pi_1,\pi_2\in X_h$ satisfying the inductive hypothesis. On the unique run along $\pi$ from $x$ to $y$, let $z$ be the vector of clock values at the end of the prefix $\pi_1$. It is easy to see that languages $L_{\bar\pi_1}(x,z)$ and $L_{\bar\pi_2}(z,y)$ should be some singletons $w_1$ and $w_2$. Thus by inductive hypothesis $\vvvert w_1\vvvert\leq 2^h \tau(w_1)+2^h$, similarly for $w_2$. As required
    \[\vvvert w\vvvert =\vvvert w_1\vvvert +\vvvert w_2\vvvert \leq (2^h \tau(w_1)+2^h)+ (2^h \tau(w_2)+2^h)=2^h \tau(w)+2^{h+1}.\]
    \item The second inductive case is $\pi=\pi_1\dots \pi_k$ with all $\pi_i\in X_h$ and $\gamma(\pi_i)=e$ for some  non-quick idempotent $e$.
    By \cref{lem:xzzy} we find a clock vector $z$, and timed words $w_1,\dots w_k$ such that $w=w_1\dots w_k$ and  $L_{\bar{\pi}_1}(x,z) =\{w_1\}$, $L_{\bar{\pi}_k}(z,y) =\{w_k\}$, and 
    $L_{\bar{\pi}_i}(z,z) =\{w_i\}$ for each $i=2..k-1$. $\vvvert w_i\vvvert \leq 2^h\tau(w_i)+2^h$. By  \cref{lem:singletons}, $\tau(w_i)\geq 1$ for all $i=2..k-1$, hence $k-2\leq \tau(w)$. By inductive hypothesis, $\vvvert w_i\vvvert \leq 2^h\tau(w_i)+2^h$ for any $i$. Summarizing, we get the required estimate:
    \[
    \vvvert w\vvvert =\sum_{i=1}^k\vvvert w_i\vvvert \leq \sum_{i=1}^k(2^h\tau(w_i)+2^h)
    \leq 2^h\tau(w)+2^h(\tau(w)+2)=2^{h+1}\tau(w)+2^{h+1}.
\]
    \item The last inductive case is $\pi=\pi_1\dots \pi_k$ with all $\pi_i\in X_h$ and $\gamma(\pi)=\gamma(\pi_i)=e$ for a quick idempotent $e$. Let us factorize $w=w_1\dots w_k$ according to \cref{lem:xzzy}. 
    
    By \cref{lem:singletons}, $\tau(w_i)=0$ for $i=2..k-1$, thus $\vvvert w\vvvert =\vvvert w_1\vvvert +\vvvert w_k\vvvert $. Applying the inductive hypothesis to $w_1$ and $w_k$, we obtain again
    \[\vvvert w\vvvert =\vvvert w_1\vvvert +\vvvert w_2\vvvert \leq (2^h \tau(w_1)+2^h)+ (2^h \tau(w_2)+2^h)=2^h \tau(w)+2^{h+1}.\qedhere \]     
\end{itemize}
 \end{proof}
 \begin{definition}
      Given a run $\rho$ (on $w$) we define its \emph{signature} $\sigma(\rho)=I_0\delta_0I_1\delta_1\dots \delta_kI_k$ with $k=\vvvert w\vvvert $, $\forall j\in \{0.. k\}\, I_j \subset \Sigma$, $\delta_j\in\Delta$ as follows:
\begin{itemize}
    \item for every (maximal) sequence of instant  (0 time) transitions in $\rho$ we keep only the set  of labels $I\subset\Sigma$;
    \item for every non-instant transition we keep only the edge $\delta\in \Delta$.
\end{itemize}
 \end{definition}
 There are at most $2^{(k+1)\#\Sigma}\cdot \#\Delta^k$ different signatures of size $k$. Consequently, for some constant $c$ the number of signatures of size $\leq k$ can be bounded by $2^{ck}$.

 \begin{lemma}\label{lem:sig}
 In an  \stA, for an idempotent $e\in \mon_d$, the following holds.
 If we have two runs $\rho_1,\rho_2$, with $\gamma(\Path(\rho_i))=e$, $\src_{\rho_i}=\dst_{\rho_i}=x$, and with the same signature $\sigma(\rho_1)=\sigma(\rho_2)$, then $d(\Word(\rho_1),\Word(\rho_2))=0$. 
 \end{lemma}

 \begin{proof}
    The factors of the runs $\rho_i$ corresponding to the $I_j$ have exactly the same effects on clocks: they reset all clocks that are not 0 in the ending region of $\delta_{j-1}$ and become 0 in the starting region of $\delta_j$ and leave all other clocks unchanged (an instant run cannot change a clock it does not reset).

    So, after the prefix of $\rho_1$ up to $\delta_{j-1}$, it is possible to run the sequence of instant transitions corresponding to $I_j$ in $\rho_2$, instead of the instant transitions corresponding to $I_j$ in $\rho_1$, and still reach the same state (and thus continue with the rest of the suffix of $\rho_1$ and still go back to $x$). By induction it is thus possible to build a run $\rho_1'$ from $x$ to $x$ such that $d(\Word(\rho_1),\Word(\rho_1'))=0$ and $\Path(\rho_1')=\Path(\rho_2)$. But since $L_{\Path(\rho_2)}(x,x)$ is a singleton (structural thinness), it is necessary that $\rho_1'=\rho_2$. Hence  $d(\Word(\rho_1),\Word(\rho_2))=0$.
 \end{proof}

We are ready to upper bound information in singletons.
\begin{lemma}\label{lem:net1}
For any \stA, there exists  a constant $b$ such that, given any idempotent non-zero $e \in \mon_d$ and clock vector $x$, the language $L_{eT}(x)\triangleq \bigcup_{\pi\in\gamma^{-1}(e)}L_{\pi T}(x,x)$ admits a $0$-net $\mathcal{N}_1(e,x,T)$ of cardinality no more than $2^{bT+b}$.
\end{lemma}
\noindent By definition, a $0$-net of a set $S$ has an element at distance $0$ from each element of $S$.
\begin{proof} 
By virtue of \cref{lem:meager-singleton}, $L_{\pi T}(x,x)$ is empty or singleton, thus 
 for any timed word  $w\in L_{eT}(x)$,  \cref{lem:freq} can be applied and $\vvvert w\vvvert \leq {aT+a}$. Thus for any run $\rho$ on $w$, the size of the signature $\sigma(\rho)$ is also bounded by ${aT+a}$.

Now, by \cref{lem:sig}, to build the $0-$net $\mathcal{N}_1(e,x,T)$, it suffices to take the trace of one representative run for each possible signature, there are $2^{c(aT+a)}$ possibilities.
\end{proof}

\begin{con}[Partition of regions] Given an \stA,  an $\varepsilon$, and an idempotent non-zero $e \in \mon_d$, take a path $\pi\in\gamma^{-1}(e)$ and its initial  region $R=S(\src_\pi)$. According to \cref{con:lyapu,lem:independent-lyapunov}, we select a family $\Lambda$ of $\dim R$  independent Lyapunov functions on $R$.  We partition the closed region $\bar R$ into cells of clock vectors having the same $\varepsilon$-integer part for each Lyapunov function (i.e.~the same integer vector $\left\lfloor \Lambda(x)/\varepsilon\right\rfloor$). We notice now that for an idempotent $e$ the partition does not depend  on the choice of the path $\pi$. 
\end{con}
Note that each cell can be characterized by a tuple of integers $\varkappa = (\varkappa_1,\ldots,\varkappa_d) \in \nat^d$ with $d = dim(R)$ 
such that for each $x \in c$ and $1\leq i\leq d$, $\varkappa_{i}=\left\lfloor {\ell_i(x) }/{\varepsilon} \right\rfloor -\left\lfloor {\ell_{i-1}(x) }/{\varepsilon} \right\rfloor$ where by convention we put $\ell_{0}(x) = 0$.

Given a vertex $v$, we denote by  $\ell_{v}$ the smallest function of $\Lambda$ such that $\ell_t(v_t)=1$, and $I_v$ its support.  Abusively, we denote $\ell_{v-1}$ the element of vector $\Lambda$ preceding $\ell_{v}$ (in the topological order of \cref{con:lyapu}), we have that 
 (for an \stA) $\lambda_v=\ell_v-\ell_{v-1}$. Finally, we denote the corresponding component of $\varkappa$ by $\varkappa_v$.

\begin{lemma}
With the hypotheses and notations from the construction above, if $t$ is a transient vertex,  and if $x$ and $y$ are clock vectors such that $P_{\bar{\pi}}(x,y)\neq \emptyset$ then $\ell_t(y)\leq \ell_t(x) - \lambda_{t}(x)$.
\end{lemma}
\begin{proof}
By virtue of \cref{lem:puri}, there exists a stochastic matrix $P$  such that $\lambda(y)=\lambda(x)P$, with $P_{uv}$ being positive only if there is an edge $(u,v)$ in $\gamma(\pi)$, where $\lambda(x)$ is the vector of barycentric coordinates of the clock vector $x$ relative to the vertices of $S(\src_\pi)$.
\begin{multline*}
    1-\ell_t(y)=\sum_{u\not\in I_t} \lambda_u(y)
        =\sum_{u\not\in I_t}\sum_{v} \lambda_v(x)P_{vu}
        =\sum_{v}\lambda_v(x)\left(\sum_{u\not\in I_t} P_{vu}\right)\\
        =\sum_{v\in I_t\setminus \{t\}}\lambda_v(x)\left(\sum_{u\not\in I_t} P_{vu}\right)
        +\lambda_t(x)\left(\sum_{u\not\in I_t} P_{tu}\right)
        + \sum_{v\not\in I_t}\lambda_v(x)\left(\sum_{u\not\in I_t} P_{vu}\right)\\
        =\sum_{v\in I_t\setminus \{t\}}\lambda_v(x)\left(\sum_{u\not\in I_t} P_{vu}\right)
        +\lambda_t(x)
        + \sum_{v\not\in I_t}\lambda_v(x)
        \geq \lambda_t(x) + \sum_{v\not\in I_t}\lambda_v(x)\\ = \lambda_t(x) + 1 - \sum_{v\in I_t}\lambda_v(x).
\end{multline*}
We conclude.
\end{proof}

\begin{lemma}\label{lem:tra-rec}
For any cell $\ck$ the following alternative holds:
\begin{itemize}
    \item if for some transient  vertex $t$  we have  $\varkappa_t \geq 1$, then for any $x,y \in \ck$ it holds that $L_{\pi}(x,y) =\emptyset$. We call such cell $\ck$ \emph{transient}.
    \item Otherwise (if for all transient vertices $v_t$ we have $\varkappa_t =0$), then there exists an  $x\in \ck$ having $L_{\bar\pi}(x,x) \neq \emptyset$ (and thus  a singleton). We denote it $\xk$ and  call such cells \emph{recurrent}.  
\end{itemize}
\end{lemma}
\begin{proof}
In the former case, following the previous lemma, we have that $\ell_t(y)  \leq \ell_t(x)- \lambda_{t}(x) = \ell_{t-1}(x)$. Then 
\[
\left\lfloor \frac{\ell_t(y)}{\varepsilon} \right\rfloor \leq  \left\lfloor \frac{\ell_{t-1}(x)}{\varepsilon} \right \rfloor < \left\lfloor \frac{\ell_{t-1}(x)}{\varepsilon} \right\rfloor + 1 \leq \left\lfloor \frac{\ell_{t-1}(x)}{\varepsilon} \right\rfloor + \varkappa_t.
\]
But in order for $y$ to be in $\ck$ , we must have, among other properties, that $\left\lfloor \frac{\ell_{t-1}(y)}{\varepsilon} \right\rfloor = \left\lfloor \frac{\ell_{t-1}(x)}{\varepsilon} \right\rfloor  $
and $\left\lfloor \frac{\ell_t(y)}{\varepsilon} \right\rfloor = \left\lfloor \frac{\ell_{t-1}(y)}{\varepsilon} \right\rfloor + \varkappa_t $, which, as per the above inequality, does not hold. 

In the latter case,  we can choose values of $\Lambda(\cdot)$ corresponding to the cell $\ck$ with $\ell_t(\cdot)=\ell_{t-1}(\cdot)$ (and thus $\lambda_t(\cdot)=0$) for all transient vertices $t$. By \cref{lem:lyapunov-derivatives} there exists  a point $x$ in the cell corresponding to these values of $\Lambda$. Such an $x$ is a combination of vertices with self-loops, which gives $L_{\bar\pi}(x,x) \neq \emptyset$, as explained in the proof of \cref{lem:meager-singleton}.
\end{proof}

\begin{lemma}\label{lem:prop:cells}
    For fixed $e$ and $\varepsilon$, cells $\ck$  have the following properties:
    \begin{itemize}
        \item the number of cells is at most $f_1(\varepsilon)\triangleq\left\lceil \frac{1}{\varepsilon}\right\rceil^{\dim R}$;
        \item the diameter of each cell is at most $K\varepsilon$, with $K$ from  \cref{lem:lyapunov-derivatives};
        \item there is a partial order over cells: $c^1\succeq c^2$, whenever some $x^1\in c^1$, $x^2\in c^2$ satisfy $\Lambda(x^1)\geq \Lambda(x^2)$. The language $L_{\bar\pi}(c^1,c^2)$ can be non-empty only if $c^1\succeq c^2$.
    \end{itemize}   
\end{lemma}
\begin{proof}
The first statement is trivial, the second is a consequence of \cref{lem:lyapunov-derivatives}, and the third follows from  \cref{def:Lyap}. 
\end{proof}

\begin{definition}\label{def:still}
A non-empty run $\rho$ is called $\varepsilon$-\emph{still} if 
\begin{itemize}
\item $e=\gamma_f(\Path(\rho))$  is an idempotent; 
\item and  $\rho$ goes from a recurrent cell $\ck$  to the same cell.
\end{itemize}
\end{definition}
We denote the set $\Word(\rho)$ over all $\varepsilon$-still runs by $\Still_\varepsilon$, and its restriction to words of duration $\leq T$  by $\Still_{\varepsilon T}$.
\begin{con}[Net for $Still$] The net $\mathcal{N}_2(\varepsilon,T)$ is built as follows. For each recurrent cell $\ck$  let $\xk$ be as in \cref{lem:prop:cells}.
    Take 
    \[\mathcal{N}_2\triangleq \bigcup_{\substack{e=ee\\\ck \text{recurrent}}}
    \mathcal{N}_1(e,\xk,T+1).\]
\end{con}
\begin{lemma} Suppose $3K\varepsilon<1$.
    The set $\mathcal{N}_2(\varepsilon,T)$ is  a $3K\varepsilon$-net  for the language  $\Still_{\varepsilon T}$ of size $f_2(\varepsilon)2^{bT}$ (for some function $f_2$).
\end{lemma}
\begin{proof} 
Let $w\in \Still_{\varepsilon T}$, by definition $w=\Word(\rho)$ with $\rho$ from some $x\in\ck$ to some $y\in \ck$  (in the same recurrent cell). Let  $\pi=\Path(\rho)$. By \cref{lem:tra-rec}, $L(\pi,\xk,\xk)$ is some singleton $\{w'\}$. By \cref{lem:prop:cells},  $||x-\xk||_\infty\leq K\varepsilon$ and $||y-\xk||_\infty\leq K\varepsilon$, thence by \cref{lem:narrow-variations} $d(w,w')\leq 3K\varepsilon$. This implies that $\tau(w')\leq T+3K\varepsilon<T+1$,  and by \cref{lem:net1} there exists $w''\in \mathcal{N}_1(e,\xk,T+1)$, such that $d(w',w'')=0$. Hence $d(w,w'')\leq 2K\varepsilon$ and by construction $w''\in \mathcal{N}_2$. Thus $\mathcal{N}_2$ is indeed a  $3K\varepsilon$-net  for the language  $\Still_{\varepsilon T}$. The size estimation is obtained by multiplying the size of $\mathcal{N}_1$ by the number of cells.
\end{proof}
We proceed now with factorizing words in $\Still_\varepsilon$ fragments and transitions between them.
\begin{lemma}\label{lem:factor:meager}
For an \stA\ $\aut$ there exists a function $f_3(\varepsilon)$, such that for any $\varepsilon$ every word $w\in L(\aut)$ admits a factorization $w=w_1\dots w_N$ with $N\leq f_3(\varepsilon)$, with each $w_i$ either having  one letter or belonging to $\Still_\varepsilon$.  
\end{lemma}
\begin{proof}
    We will use Simon's  theorem (for monoid $\mon_d$) and prove by induction that for every path $\pi \in X_h$, every $w\in L_\pi$ admits a factorization with $N\leq (2f_1(\varepsilon))^h$. Since the height $h$ is bounded by $K$, this yields the required bound $N<f_3(\varepsilon)$ with $f_3(\varepsilon)=(2f_1(\varepsilon))^K$
\begin{itemize}
    \item In the base case $\pi\in X_0$, the word $w$ has one letter and its factorization is trivial ($N=1$).
    \item In the first inductive case $\pi=\pi_1\pi_2$, with $\pi_1,\pi_2\in X_h$. Consequently $w=w_1w_2$ with $w_i\in L_{\pi_i}$, the words $w_1,w_2$ admitting factorizations with $(2f_1(\varepsilon))^h$ factors. Concatenating those we obtain a factorization for $w$ with  $2(2f_1(\varepsilon))^h<(2f_1(\varepsilon))^{h+1}$ factors.
    \item In the other inductive case, there exists an idempotent $e$ and $\pi=\pi_1\dots \pi_J$, with $\gamma_d(\pi_i)=e$ and $\pi_i\in X_h$. Passing to runs, we get $\rho=\rho_1\dots\rho_J$, with $w=\Word(\rho)$, and $\Path(\rho_i)=\pi_i$. The factorization in $\rho_i$ is too long and too precise, we will partition it into fewer clusters. Let $z_i$ be clock vectors (traversed by the run $\rho$), such that $\src_{\rho_i}=z_{i-1}$ and $\dst_{\rho_i}=z_{i}$. Let $c_i$ be the cell of $z_i$, that is the unique $\ck$  containing $z_i$. By \cref{lem:prop:cells} $c_0\succeq c_1 \succeq \cdots \succeq c_J$, and there are at most $f_1(\varepsilon)$ different cells in this sequence, let us denote them $c^i$. It means that the sequence can be written as 
    \[c^0=c^0=\dots c^0 \succ c^1=\dots c^1 \succ c^2 \dots c^M\]  with $M<f_1(\varepsilon)$. Notice that only recurrent cells can be repeated. We will concatenate together all runs $\rho_j$ from $c^i$ to itself, and call the result of concatenation $\zeta^i$; we also denote $\rho_j$ from $c^i$ to $c^{i+1}$ by $\rho^i$. This yields a new factorization:
    \[
    \rho=\zeta^0\rho^0\zeta^1\rho^1\dots \zeta^M, \quad M<f_1(\varepsilon)
    \]
    with all $\zeta^i$ from a cell to the same cell, and all $\Path(\rho^i)\in X_h$.  Applying the projection $\Word(\cdot)$ we get a factorization
    $w=u^0v^0u^1v^1\dots u^M$ with all $u^i\in \Still_\varepsilon$ (or equal to the empty word)  and all $v^i$ satisfying the inductive hypothesis. Factorizing each $v^i$ according to the inductive hypothesis, we get  a required factorization of $w$ into $f_1(\varepsilon)+f_1(\varepsilon)(2f_1(\varepsilon))^h<(2f_1(\varepsilon))^{h+1}$ factors. This concludes the proof. \qedhere
\end{itemize}
\end{proof}

At the last step we build an $O(\varepsilon)$-net for all the words $w\in L_T(\aut)$:
\begin{con}[Net for $L_T$]
The elements of the  net  $\mathcal{N}_3(\varepsilon,T)$ are generated as follows
\begin{itemize}
    \item choose $N \leq f_3(\varepsilon)$ and $N$ dates $0 \leq t_1\leq t_2\leq\cdots t_N \leq T$ which are multiples of $\varepsilon$;
    \item at some of those dates put a letter in $\Sigma$;
    \item at all other $t_i$ insert an element of   $\mathcal{N}_2(\varepsilon,t_{i+1}-t_{i}+\varepsilon)$.
\end{itemize}
\end{con}
We remark that elements of the second kind can slightly overlap with the rest of the word, i.e.,~go beyond $t_{i+1}$.

\begin{lemma}
    $\mathcal{N}_3(\varepsilon,T)$ is a $(3K+1/2)\varepsilon$-net  for $L_T$ (for $\varepsilon$ small enough).
\end{lemma}
\begin{proof}
Let $w\in L_T(\aut)$,  we will explain how to find a world in $\mathcal{N}_3(\varepsilon,T)$ at a small distance from $w$.

To that aim, consider the factorization  $w=w_1\dots w_N$ provided by  \cref{lem:factor:meager}. 
For each factor $w_i$ of the first kind (containing a delay and one letter $a_i$) let $t_i=\tau (w_1\dots w_i)$, i.e.~the date  when this letter happens in the word $w$. For each factor of the second kind ($w_i\in\Still_\varepsilon$) let $t_i=\tau (w_1\dots w_{i-1})$ the starting date of factor $w_i$. We notice that $0\leq t_1\leq t_2\leq t_N$, and that for any factor $w_i$ of the second kind $\tau(w_i)\leq t_{i+1}-t_i$.

For every $t_i$, let $t'_i$ be the closest $\varepsilon$-integer. We have that $0\leq t'_1\leq t'_2\leq t'_N$, and for any factor $w_i$ of the second kind $\tau(w_i)\leq t'_{i+1}-t'_i+\varepsilon$. For any such $w_i$ of the second kind let $w'_i$ be its $2K\varepsilon$-approximation from the net $\mathcal{N}_2(\varepsilon,t'_{i+1}-t'_{i}+\varepsilon)$.

We are ready to build an approximated word $w'$. For each factor $w_i$ of the first kind, the word $w'$ contains $a_i$ at date $t'_i$. For each factor $w_i$ of the second kind, the word $w'$ contains the approximation $w'_i$ starting at $t'_i$.

By construction, $w'\in \mathcal{N}_3(\varepsilon,T)$. On the other hand, $d(w,w')\leq \varepsilon$, indeed each factor of the first kind in $w$ is shifted by at most $\varepsilon/2$, and each factor of the second kind in $w$ is $3K\varepsilon$-approximated and shifted by at most $\varepsilon/2.$ 
\end{proof}

\begin{lemma}
The cardinality of $\mathcal{N}_3(\varepsilon,T)$ is bounded by $T^{f_3(\varepsilon)}f_4(\varepsilon) 2^{bT}$.
\end{lemma}
\begin{proof}
There are $\binom{T/\varepsilon}{f_3(\varepsilon)}$ ways of choosing $f_3(\varepsilon)$ dates, which can be  bounded by $(T/\varepsilon)^{f_3(\varepsilon)}$. For each of those one can choose a letter in $\Sigma$, or decide to put an element of $\mathcal{N}_2$, or put nothing (if there are less than $f_3(\varepsilon)$ factors), this gives $(\#\Sigma+2)^{f_3(\varepsilon)}$ possibilities. For elements of $\mathcal{N}_2$ the number of possibilities is
\[
\prod_{\text{some }j<f_3(\varepsilon)}f_2(\varepsilon/2)2^{b(t_{j+1}-t_j+\varepsilon)}\leq 
2^{bT} 2^{b\varepsilon f_3(\varepsilon)}f_2(\varepsilon/2)^{f_3(\varepsilon)}.
\]
Multiplying everything together we get an expression of the required form $T^{f_3(\varepsilon)}f_4(\varepsilon) 2^{bT}$ for an appropriate $f_4$.
\end{proof}

\propmeagerupper*
\begin{proof}  Let $K'=3K+1/2$. By the previous lemma, the bandwidth is bounded by 
\[ 
\lim_{T\to \infty}\frac{\log \left(T^{f_3(\varepsilon/K')}f_4(\varepsilon/K') 2^{bT}\right) }{T}= b. \qedhere
\]
\end{proof}
\subsection{Proof details for lower bound}
\begin{lemma}
    If a cyclic path $\pi$ in an \RTA\ is   not structurally meager, then $\bandc_\varepsilon(L_{\pi^*})=\Omega(\log(1/\varepsilon))$, and hence $L_{\pi^*}$ is not a meager language.
\end{lemma}
\begin{proof} 
Without loss of generality, we assume $\gamma_f(\pi)$ is idempotent. It is not structurally meager, hence it has a $\wide$ self-loop on some vertex $v$, implying $\dim P_\pi(v,v)\geq 1$. Now consider $\gamma_d(\pi)[v,v]$. It cannot be $\instant$
or $\none$, since $\gamma_f[v,v]=\wide$. If it is $\fast$, then $\pi$ is structurally obese, thus its language is obese (and not meager).
So we consider the remaining case: $\gamma_d(\pi)[v,v]=\slow$: i.e.~the cycle cannot be executed in less than one time unit. So not only the set of possible timings has a dimension of at least 1, but also all timings have a total duration of at least 1. Without loss of generality, we can assume we can select a segment of realizations such that all its timing vectors have the same total duration (if not, take any segment: take 2 runs $\rho_1$ and $\rho_2$ in this segment, going from $v$ to itself, with different total duration; from now on instead of $\pi$, we consider $\pi^2$ and the segment delimited by the timings of the runs $\rho_1\rho_2$ and $\rho_2\rho_1$). 

We consider the set of runs $\sigma_k=\frac{k}{n} \rho_1 + \frac{n-k}{n}\rho_2$ and look at the first transition occurring at a different date in $\rho_1$ and $\rho_2$ (call its index $m$), assuming wlog it occurs earlier in $\rho_1$ and let us call $\tau$ the difference between the date of $m$-th event in $\rho_2$ and in $\rho_1$. If $k>j$, in $\sigma_j$, the $m$th event cannot be matched in $\sigma_k$ with a time difference smaller than $\tau/n$, implying that $d(\sigma_j,\sigma_k)\geq\vec{d}(\sigma_j,\sigma_k)\geq \tau/n$. Now, considering arbitrary repetitions, we define the language $S_n\triangleq\left(\left(\Word(\sigma_k)\right)_{k=1..n}\right)^*$). We remark there is also a distance of at least $\tau/n$ between any two words of $S_n$: indeed the factors $\sigma_k$ having all the same duration, the $i|\pi|+m$th letter of a word cannot be matched with any letter of index smaller than $i|\pi|$ in another word. Also, no letter of index greater than $(i+1)|\pi|$ can be matched with this letter. Hence $S_n$ is $\tau/n$ separated. So, for any arbitrary $\varepsilon$ (smaller than $\tau$), there is an $\varepsilon$-separated subset of $L(\bar{\pi}^*)_T$ of size $\left\lfloor\frac{\tau}{\varepsilon}\right\rfloor^{\lfloor \frac{T}{\tau(\rho_1)}\rfloor}$.

Using \cref{lem:closed-semantics}, it follows this is also true of $L(\pi^*)_T$.
\end{proof}
\cref{prop:meager:lower} follows immediately.
\section{\PSPACE-hardness is not so hard}
\thmclassif*
We prove here only the \PSPACE-hardness of the three problems, denoted \Meager, \Normal\, and \Obese. \PSPACE\ membership is proved in \cref{sec:classif}.
\begin{proof} We proceed by reduction from the following \PSPACE-complete problem \cite{AD}:
\begin{quote}
\myreach: given a \DTA\ $\aut$ and two locations $p,q$, does there exist a run from state $(p,0)$ to location $q$?
\end{quote}
From an instance of \myreach\ we build two \DTA\ $\mathcal B$ and $\mathcal C$ as follows:
\begin{itemize}
    \item $\mathcal B$ is obtained from $\aut$ by adding one more location $f$, and two edges $q\to f$ and $f\to f$ with a new label $d$, no reset and the guard \textbf{true}. The initial condition $I(p)$ is $x=0$, the final condition $F(f)$ is \textbf{true}, and all the other initial and final conditions are \textbf{false}.
    \item $\mathcal C$ is obtained from $\mathcal B$ by adding  one more location $s$, and  two edges $(s,s,d,x_1>1,\{x_1\})$, $(s,p,d,x_1=1,X)$ where $X$ is the set of clocks in $\mathcal{A}$. The initial state is  $I(s)=(x=0)$, the final ones are $F(f)=F(s)=\text{\textbf{true}}$. 
\end{itemize}
By construction 
\begin{itemize}
    \item whenever $\myreach(\aut, p,q)$, both automata $\mathcal B$ and $\mathcal C$ are obese: indeed,  accepting runs can traverse $\aut$ from $p$ to $q$ and then go to the obese loop on location $f$;
    \item otherwise, $\mathcal B$ has empty language, and $\mathcal C$ is normal: indeed, the former automaton has no accepting run, and the latter can loop on $s$ producing a normal language. 
\end{itemize}
Thanks to the polynomial reduction $(\aut, p,q)\mapsto \mathcal B$  we have $\myreach\preceq \Obese$ and $(\neg \myreach)\preceq \Meager$. The other reduction $(\aut, p,q)\mapsto \mathcal C$ yields $(\neg \myreach)\preceq \Normal$. Thus the three problems \Obese, \Meager\, and \Normal\ are \PSPACE-hard.
\end{proof}
\section{Extending to more general automata}
\corextend*
\begin{proof}[Proof sketch]
Given a non-deterministic timed automaton $\mathcal B$ (without $\epsilon$-transitions) with alphabet $\Sigma$,
it is always possible to relabel its transitions with a larger alphabet $\Gamma$ to make it deterministic (for example choosing distinct letters for each transition), let us call this deterministic automaton $\aut$.  It holds that $L(\mathcal{B}) = \mu(L(\aut))$, where $\mu$ is a letter-to-letter renaming from $\Gamma$ to $\Sigma$. \cref{thm:classif} applies to  $\aut$, and 
we claim that $L(\mathcal{B})$ belongs to the same class as $L(\aut)$.

To prove the claim we notice first that if $\mathcal{N}$ is an $\varepsilon$-net for $L_T(\aut)$, then $\mu(\mathcal{N})$ is a smaller $\varepsilon$-net for $L_T(\mathcal{B})$, hence
\begin{equation}\label{eq:nondet}
    \bandh(L(\mathcal{B}))\leq \bandh(L(\aut)).
\end{equation}
\begin{description}
\item[if  $\aut$ is meager,] then  $\mathcal{B}$ is also meager by virtue of \eqref{eq:nondet};
\item[if  $\aut$ is obese,] then it is structurally obese by \cref{thm:obese}, and thus the proof of \cref{prop:strobese} can be applied, yielding one cycle $\pi$ in $\aut$  and a big $\varepsilon$-separated subset $\mathcal{S}\subset L_{T,\pi^*}(\aut)$. Since all elements of $S$ follow the same cycle, the restriction of $\mu$ to $S$ is injective, moreover  $\mu(S)$ is also a big $\varepsilon$-separated subset in $L_T(\mathcal{B})$, thus $\mathcal B$ is obese.
\item[if  $\aut$ is normal,] then by  \eqref{eq:nondet} $\bandh(\mathcal{B})= O(\log 1/\varepsilon)$. On the other hand, $\aut$ is not structurally meager, and the proof of \cref{prop:meager:lower} yields a big $\varepsilon$-separated subset along one cycle. Similarly to the previous case this implies that $L_T(\mathcal{B})$ is big enough: $\Omega(\log 1/\varepsilon)$. \qedhere
\end{description}  
\end{proof}

\section{Completing the proof of the last theorem}
\thmthinvsmeager*
The only missing element in the proof sketch in \cref{sec:comparing} concerned the incompatibility of obesity Type I with thinness. 
\begin{proof}
    Consider a thick automaton, with its forgetful cycle $\pi$. If its orbit has a single vertex, since $\pi$ has a non-punctual transition it can be realized by several runs, hence the only edge of $\gamma_f(\pi)$ is $\wide$. 
    Otherwise, $\gamma(\pi)$ has a complete SCC of size at least 2, and then all its vertices have $\wide$ self-loops in $\gamma_f(\pi^2)$.
    
    Now consider an automaton with structurally obese cycle $\pi$ of Type I. By definition, $\gamma_d(\pi)$  has both an $\instant$ and a $\slow$ self-loop. 
  According to \cref{lem:instant+slow=forgetful}, $\pi$ must be forgetful and the automaton cannot be thin.
\end{proof}

%% file: mainArxiv.bbl
\begin{thebibliography}{10}

\bibitem{AD}
Rajeev Alur and David~L. Dill.
\newblock A theory of timed automata.
\newblock {\em Theoretical Computer Science}, 126:183--235, 1994.
\newblock \href {https://doi.org/10.1016/0304-3975(94)90010-8}
  {\path{doi:10.1016/0304-3975(94)90010-8}}.

\bibitem{timedCoding}
Eugene Asarin, Nicolas Basset, Marie-Pierre B{\'e}al, Aldric Degorre, and
  Dominique Perrin.
\newblock Toward a timed theory of channel coding.
\newblock In {\em Proc.~FORMATS}, volume 7595 of {\em LNCS}, pages 27--42,
  2012.
\newblock \href {https://doi.org/10.1007/978-3-642-33365-1_4}
  {\path{doi:10.1007/978-3-642-33365-1_4}}.

\bibitem{entroJourn}
Eugene Asarin, Nicolas Basset, and Aldric Degorre.
\newblock Entropy of regular timed languages.
\newblock {\em Information and Computation}, 241:142--176, 2015.
\newblock \href {https://doi.org/10.1016/j.ic.2015.03.003}
  {\path{doi:10.1016/j.ic.2015.03.003}}.

\bibitem{distance}
Eugene Asarin, Nicolas Basset, and Aldric Degorre.
\newblock Distance on timed words and applications.
\newblock In {\em Proc.~FORMATS}, volume 11022 of {\em LNCS}, pages 199--214,
  2018.
\newblock \href {https://doi.org/10.1007/978-3-030-00151-3_12}
  {\path{doi:10.1007/978-3-030-00151-3_12}}.

\bibitem{wordgen}
Beno{\^{\i}}t Barbot, Nicolas Basset, and Alexandre Donz{\'{e}}.
\newblock Wordgen: a timed word generation tool.
\newblock In {\em Proc.~{HSCC}}, pages 16:1--16:7. {ACM}, 2023.
\newblock \href {https://doi.org/10.1145/3575870.3587116}
  {\path{doi:10.1145/3575870.3587116}}.

\bibitem{Bengtsson2004}
Johan Bengtsson and Wang Yi.
\newblock {\em Timed Automata: Semantics, Algorithms and Tools}, pages 87--124.
\newblock Springer Berlin Heidelberg, 2004.
\newblock \href {https://doi.org/10.1007/978-3-540-27755-2_3}
  {\path{doi:10.1007/978-3-540-27755-2_3}}.

\bibitem{perrin}
Jean Berstel, Dominique Perrin, and Christophe Reutenauer.
\newblock {\em Codes and Automata}.
\newblock Cambridge University Press, 2009.

\bibitem{bojanczyk}
Miko{\l}aj Boja{\'n}czyk.
\newblock Factorization forests.
\newblock In {\em Proc.~DLT}, volume 5583 of {\em LNCS}, pages 1--17, 2009.
\newblock \href {https://doi.org/10.1007/978-3-642-02737-6_1}
  {\path{doi:10.1007/978-3-642-02737-6_1}}.

\bibitem{growth1}
Martin~R. Bridson and Robert~H. Gilman.
\newblock Context-free languages of sub-exponential growth.
\newblock {\em Journal of Computer and System Sciences}, 64(2):308--310, 2002.
\newblock \href {https://doi.org/10.1006/jcss.2001.1804}
  {\path{doi:10.1006/jcss.2001.1804}}.

\bibitem{ChomskyMiller}
Noam Chomsky and George~A. Miller.
\newblock Finite state languages.
\newblock {\em Information and Control}, 1(2):91 -- 112, 1958.
\newblock \href {https://doi.org/10.1016/S0019-9958(58)90082-2}
  {\path{doi:10.1016/S0019-9958(58)90082-2}}.

\bibitem{DBM}
David~L. Dill.
\newblock Timing assumptions and verification of finite-state concurrent
  systems.
\newblock In {\em Proc.~{CAV}}, volume 407 of {\em LNCS}, pages 197--212, 1989.
\newblock \href {https://doi.org/10.1007/3-540-52148-8_17}
  {\path{doi:10.1007/3-540-52148-8_17}}.

\bibitem{growth2}
Pawe{\l} Gawrychowski, Dalia Krieger, Narad Rampersad, and Jeffrey Shallit.
\newblock Finding the growth rate of a regular or context-free language in
  polynomial time.
\newblock In {\em Proc. DLT}, volume 5257 of {\em LNCS}, pages 339--358, 2008.
\newblock \href {https://doi.org/10.1007/978-3-540-85780-8_27}
  {\path{doi:10.1007/978-3-540-85780-8_27}}.

\bibitem{grigorchuk}
Rostislav Grigorchuk and Antonio Mach{\'\i}.
\newblock An example of an indexed language of intermediate growth.
\newblock {\em Theoretical Computer Science}, 215(1):325--327, 1999.
\newblock \href {https://doi.org/10.1016/S0304-3975(98)00161-3}
  {\path{doi:10.1016/S0304-3975(98)00161-3}}.

\bibitem{tubes}
Vineet Gupta, Thomas~A. Henzinger, and Radha Jagadeesan.
\newblock Robust timed automata.
\newblock In {\em Hybrid and Real-Time Systems}, volume 1201 of {\em LNCS},
  pages 331--345. Springer, 1997.
\newblock \href {https://doi.org/10.1007/BFb0014736}
  {\path{doi:10.1007/BFb0014736}}.

\bibitem{efmplus}
K.A.S. Immink.
\newblock {EFMPlus}: The coding format of the multimedia compact disc.
\newblock {\em IEEE Transactions on Consumer Electronics}, 41(3):491--497,
  1995.
\newblock \href {https://doi.org/10.1109/30.468040}
  {\path{doi:10.1109/30.468040}}.

\bibitem{imminkBook}
K.A.S. Immink.
\newblock {\em Codes for Mass Data Storage Systems}.
\newblock Shannon Foundation Publ., 2004.

\bibitem{formats2022}
Bernardo Jacobo~Incl{\'{a}}n, Aldric Degorre, and Eugene Asarin.
\newblock Bounded delay timed channel coding.
\newblock In {\em Proc. {FORMATS}}, volume 13465 of {\em LNCS}, pages 65--79,
  2022.
\newblock \href {https://doi.org/10.1007/978-3-031-15839-1_4}
  {\path{doi:10.1007/978-3-031-15839-1_4}}.

\bibitem{3approaches}
Andrei Kolmogorov.
\newblock Three approaches to the quantitative definition of information.
\newblock {\em International Journal of Computer Mathematics}, 2(1-4):157--168,
  1968.
\newblock \href {https://doi.org/10.1080/00207166808803030}
  {\path{doi:10.1080/00207166808803030}}.

\bibitem{kolmoEpsilon}
Andrei Kolmogorov and Vladimir Tikhomirov.
\newblock $\varepsilon$-entropy and $\varepsilon$-capacity of sets in function
  spaces.
\newblock {\em Uspekhi Matematicheskikh Nauk}, 14(2):3--86, 1959.
\newblock URL: \url{http://mi.mathnet.ru/eng/umn7289}, \href
  {https://doi.org/10.1007/978-94-017-2973-4_7}
  {\path{doi:10.1007/978-94-017-2973-4_7}}.

\bibitem{Marcus}
Douglas Lind and Brian Marcus.
\newblock {\em An Introduction to Symbolic Dynamics and Coding}.
\newblock Cambridge University Press, 1995.

\bibitem{papa}
Christos~H. Papadimitriou.
\newblock {\em Computational Complexity}.
\newblock Addison-Wesley, 1994.

\bibitem{puri}
Anuj Puri.
\newblock Dynamical properties of timed automata.
\newblock {\em Discrete Event Dynamic Systems}, 10(1-2):87--113, 2000.
\newblock \href {https://doi.org/10.1023/A:1008387132377}
  {\path{doi:10.1023/A:1008387132377}}.

\bibitem{rhodes}
John Rhodes and Pedro~V. Silva.
\newblock {\em Boolean Representations of Simplicial Complexes and Matroids}.
\newblock Springer International Publishing, 2015.
\newblock \href {https://doi.org/10.1007/978-3-319-15114-4}
  {\path{doi:10.1007/978-3-319-15114-4}}.

\bibitem{ocan}
Ocan Sankur.
\newblock {\em Robustness in timed automata: analysis, synthesis,
  implementation}.
\newblock PhD thesis, {\'{E}}cole normale sup{\'{e}}rieure de Cachan, Paris,
  France, 2013.
\newblock URL: \url{https://tel.archives-ouvertes.fr/tel-00910333}.

\bibitem{shannon48}
Claude~E. Shannon.
\newblock A mathematical theory of communication.
\newblock {\em The Bell System Technical Journal}, 27(3):379--423, 1948.
\newblock \href {https://doi.org/10.1002/j.1538-7305.1948.tb01338.x}
  {\path{doi:10.1002/j.1538-7305.1948.tb01338.x}}.

\bibitem{simon}
Imre Simon.
\newblock Factorization forests of finite height.
\newblock {\em Theoretical Computer Science}, 72(1):65--94, 1990.
\newblock \href {https://doi.org/10.1016/0304-3975(90)90047-L}
  {\path{doi:10.1016/0304-3975(90)90047-L}}.

\bibitem{ambiguity}
Andreas Weber and Helmut Seidl.
\newblock On the degree of ambiguity of finite automata.
\newblock {\em Theoretical Computer Science}, 88(2):325--349, 1991.
\newblock \href {https://doi.org/10.1016/0304-3975(91)90381-B}
  {\path{doi:10.1016/0304-3975(91)90381-B}}.

\end{thebibliography}
